\documentclass[11pt]{amsart}
\usepackage{amsmath}
\usepackage{graphicx}
\usepackage{bm}
\usepackage{enumerate}
\theoremstyle{plain}
\newtheorem{theorem}{Theorem}[section]
\newtheorem{lemma}[theorem]{Lemma}
\newtheorem{prop}[theorem]{Proposition}
\newtheorem{cor}[theorem]{Corollary}

\theoremstyle{definition}

\numberwithin{equation}{section}

\newcommand{\R}{\mathbb R}

\newcommand{\tr}{\operatorname{Tr}}

\newcommand{\abs}[1]{\left|#1\right|}
\begin{document}




\title[Relating mass to angular momentum and charge]{Relating mass to angular momentum and charge in 5-dimensional minimal supergravity}

\author{Aghil Alaee}
\address{Department of Mathematical and Statistical Science\\
University of Alberta\\
Edmonton AB T6G 2G1, Canada}
\email{khangha@ualberta.ca}


\author{Marcus Khuri}
\address{Department of Mathematics, Stony Brook University, Stony Brook, NY 11794, USA}
\email{khuri@math.sunysb.edu}

\author{Hari Kunduri}
\address{Department of Mathematics and Statistics\\
Memorial University of Newfoundland\\
St John's NL A1C 4P5, Canada}
\email{hkkunduri@mun.ca}


\thanks{A. Alaee acknowledges the support of a PIMS Postdoctoral Fellowship.  M. Khuri acknowledges the support of NSF Grant DMS-1308753. H. Kunduri acknowledges the support of NSERC Grant 418537-2012.}

\begin{abstract}
We prove a mass-angular momentum-charge inequality for a broad class of maximal, asymptotically flat, bi-axisymmetric initial data within the context of five-dimensional minimal supergravity.  We further show that the charged Myers-Perry black hole initial data are the unique minimizers.  In addition, we establish a rigidity statement for the relevant BPS bound, and give a variational characterization of BMPV black holes.
\end{abstract}
\maketitle

\section{Introduction}

As is well known, asymptotically flat stationary black holes in 4D Einstein-Maxwell theory are characterized by their mass $m$, angular momentum $\mathcal{J}$, and electric charge $q$. In order to avoid a naked singularity, these physical parameters must satisfy the inequality
\begin{equation}\label{4dinequality}
m^2 \geq \frac{q^2 + \sqrt{q^4 + 4 \mathcal{J}^2}}{2}
\end{equation}
More generally, it has been shown that for maximal, simply connected, axisymmetric initial data sets with two ends, one designated asymptotically flat and the other either asymptotically flat or asymptotically cylindrical, the above inequality holds
\cite{Chrusciel2009,Costa2010,Dain2008,schoen2013convexity} (see \cite{dain2012geometric} for a thorough review). As was partially shown in these references, and completed in \cite{khuri2015positive}, the bound is saturated if and only if the initial data is that corresponding to the extreme Kerr-Newman black hole. This result has also been generalized to the setting of multiple black holes in \cite{chrusciel2008mass,khuri2015positive}.

It is natural to ask whether the above inequality \eqref{4dinequality} admits a generalization to dimensions greater than four.  This program was initiated in our recent article \cite{alaee2015global}, in which we considered geometric inequalities satisfied by a broad class of asymptotically flat initial data for the \emph{vacuum} Einstein equations in spacetime dimension five.  The initial data $(M^4,g, k)$, consisting of a Riemannian 4-manifold $M^4$ with metric $g$, and second fundamental form tensor $k$, was assumed to be bi-axisymmetric so that it admits a $U(1)^2$ action by isometries.  Such data possess two independent angular momenta $\mathcal{J}_l$, $l=1,2$ in addition to the ADM mass $m$.  The inequality reads
\begin{equation}\label{vacuumineq}
m^3\geq \frac{27\pi}{32}\left(|\mathcal{J}_1|+|\mathcal{J}_2|\right)^2,
\end{equation}
and holds for $M^4$ diffeomorphic to $\mathbb{R}^4\setminus\{0\}$.  The class of data treated in \cite{alaee2015global} includes that of the Myers-Perry black hole family \cite{Myers2011}, which is the natural generalization of the Kerr solution to $D>4$. Indeed, we have established that the lower bound of \eqref{vacuumineq} is achieved if and only if the initial data set is the canonical slice of an extreme Myers-Perry spacetime.

The present article is concerned with charged generalizations of \eqref{vacuumineq}, where the charge arises from an Abelian (Maxwell) gauge field.  Initial data of the five-dimensional Einstein-Maxwell theory would be the obvious candidate geometries.  However, pure Einstein-Maxwell theory restricted to spacetimes with $U(1)^2$ isometry does not admit a coset structure similar to that which is found in four dimensions.  The lack of associated solution-generating techniques is perhaps the reason that there is no known `charged' Myers-Perry solution to pure Einstein-Maxwell theory.
For this reason, we will instead consider five-dimensional, $\mathcal{N}=1$ minimal supergravity \cite{Cremmer:1980gs}, which admits a harmonic map structure and is thus the natural theory to study \cite{Cremmer:1997ct, Cremmer:1998px, Cremmer:1999du} (indeed four-dimensional Einstein-Maxwell theory is itself a supergravity theory).  The theory is `minimal' in the sense that it is the simplest supersymmetric extension of general relativity.   As we discuss below, the only relevant additional feature of the supergravity theory is that the Maxwell equation is now self-sourced, that is $d \star_{5} F \neq 0$ where $F$ is the field strength.  Stationary black hole solutions to minimal supergravity have been studied extensively over the past decade, motivated by developments in string theory (see e.g. \cite{emparan2008black} for a review). This is because the action arises via a simple Kaluza-Klein type dimensional reduction of ten-dimensional type IIB supergravity on $\mathbb{T}^5$ (see  \cite{Elvang:2004ds} for details on the explicit compactification).  In particular, an important achievement of string theory was the microscopic calculation of the Bekenstein-Hawking entropy of five-dimensional black hole solutions of supergravity in terms of an underlying weakly-coupled string theory \cite{Breckenridge:1996is,strominger1996microscopic}.

The model geometry in our analysis is the charged Myers-Perry black hole solution \cite{Chong:2005hr,Cvetic:1996xz}. This is a four-parameter family of asymptotically flat, bi-axisymmetric stationary black hole solutions characterized by $(m, \mathcal{J}_1,\mathcal{J}_2, Q)$, where $Q$ is the electric charge. When $Q=0$, it reduces to the vacuum Myers-Perry black hole discussed above, while for $\mathcal{J}_l =0$, $l=1,2$ it reduces to the familiar two-parameter family of Reissner-Nordstr\"{o}m black hole solutions. These latter solutions are in fact also solutions of pure Einstein-Maxwell theory.  The charged Myers-Perry solution also contains a 3-parameter subset of extreme black holes, which play an important role in the proof of our inequality.  These properties indicate that the charged Myers-Perry solutions are the natural five-dimensional generalization of the Kerr-Newmann family.

A crucial ingredient in the proof of \eqref{vacuumineq} is to show that the mass of the initial data is bounded below by a certain mass functional \cite{alaee2015global}. This mass functional is itself a regularization of the (divergent) Dirichlet energy for singular maps between $\mathbb{R}^3$ and $SL(3,\mathbb{R})/SO(3)$, where the target space is equipped with a metric of nonpositive sectional curvature.  The critical points of this latter energy functional are precisely the stationary, bi-axisymmetric solutions of the vacuum Einstein equations in $D=4+1$ \cite{maison1979ehlers}.  Remarkably, it can be shown that for minimal supergravity, the stationary bi-axisymmetric solutions arise as critical points for an energy functional with target $G_{2(2)}/SO(4)$  \cite{Cremmer:1997ct, Cremmer:1998px, Cremmer:1999du}.  This space\footnote{$G_{2(2)}$ refers to the noncompact real Lie group whose complexification is $G_2$; the notation $2(2)$ refers respectively to the rank and character of the group.}  is eight-dimensional and again carries a metric with nonpositive curvature.  This allows for the convexity arguments used in the proof of \eqref{vacuumineq} to be applied.  In particular we will construct an appropriate mass functional for a large class of initial data of minimal supergravity, and show how it can be interpreted as a regularization of  the Dirichlet energy for singular harmonic maps taking $\mathbb{R}^3 \to G_{2(2)} / SO(4)$.

In order to state the main result we first discuss the appropriate setting. In addition to a Riemannian 4-manifold $M^4$ with metric $g$ and extrinsic curvature $k$, an initial data set for 5-dimensional minimal supergravity comes equipped with a 1-form $E$ and 2-form $B$ which represent the electric and magnetic field, respectively. These quantities are related to one another, as well as the energy density $\mu_{SG}$ and momentum density $J_{SG}$ of the nonelectromagnetic matter fields, through the constraint equations \eqref{Const1}, \eqref{Const2}, and \eqref{Const3} derived in Section \ref{secsugra}. As with \eqref{vacuumineq}, the data are assumed to be bi-axisymmetric with the $U(1)^2$ symmetry generated by Killing fields $\eta_{(l)}$, $l=1,2$. Associated with each Killing field is the ADM angular momentum
\begin{equation}\label{admam}
\mathcal{J}_{l}=\frac{1}{8\pi}\int_{S^{3}_{\infty}}(k_{ij}-(\tr_{g} k)g_{ij})\nu^{i}\eta_{(l)}^{j},
\text{ }\text{ }\text{ }\text{ }\text{ }l=1,2,
\end{equation}
where $S^{3}_{\infty}$ indicates the limit as $r\rightarrow\infty$ of integrals over coordinate spheres $S^{3}_{r}$, with unit outer normal $\nu$, in a designated asymptotically flat end. Typically, enhanced asymptotics beyond the usual definitions of asymptotic flatness are needed to guarantee that this limit exists. However, here it will be assumed that the momentum density vanishes in the Killing directions $J_{SG}(\eta_{(l)})=0$, $l=1,2$, and as is shown in Section \ref{secpot} this is sufficient to guarantee that \eqref{admam} is finite and well-defined. Furthermore, the ADM mass is given by
\begin{equation}\label{admmass}
m=\frac{1}{16\pi}\int_{S^{3}_{\infty}}(g_{ij,i}
-g_{ii,j})\nu^{j},
\end{equation}
and the total electric charge takes the usual form
\begin{equation}\label{totalcharge}
Q = \frac{1}{16\pi} \int_{S^3_\infty} \star E,
\end{equation}
where $\star$ denotes the Hodge star operation. Due to the fact that the magnetic field is represented by a 2-form, there is no meaningful notion of total magnetic charge for our purposes; this is examined in more detail in Section \ref{secpot}.
Certain combinations of the mass, angular momenta, and charge, which will be labeled $a$, $b$, and $q$, appear naturally in the explicit expression for the charged Myers-Perry spacetime and play a role in the statements below. In particular, the charged Myers-Perry solution has a naked singularity precisely when $ab+q=0$. These quantities are given implicitly through the relations
\begin{equation}\label{relations}
\mathcal{J}_{1}=\frac{2m}{3}a+\frac{Q}{\sqrt{3}}b,
\text{ }\text{ }\text{ }\text{ }\text{ }
\mathcal{J}_{2}=\frac{Q}{\sqrt{3}}a+\frac{2m}{3}b,
\text{ }\text{ }\text{ }\text{ }\text{ }
Q=\frac{\sqrt{3}\pi}{4}q.
\end{equation}
Observe that these relations define $a$ and $b$ uniquely in terms of the mass, angular momenta, and charge whenever $m^2\neq \frac{3}{4}Q^2$, which in light of \eqref{maininequality} is always satisfied unless $m=0$. However, the mass cannot vanish in the context of our results as there are two ends, and thus there is no obstruction to inverting the relations \eqref{relations}.

\begin{theorem}\label{MainTheorem}
Let $(M^4,g,k,E,B)$ be a smooth, complete, bi-axially symmetric, maximal initial data set for the 5-dimensional minimal supergravity equations satisfying $\mu_{SG}\geq 0$
and $J_{SG}(\eta_{(l)})=0$, $l=1,2$ and with two ends, one designated asymptotically flat and the other either asymptotically flat or
asymptotically cylindrical.
If $M^{4}$ is diffeomorphic to $\mathbb{R}^{4}\setminus\{0\}$ and admits a global system of Brill coordinates then
\begin{equation}\label{maininequality}
m\geq \frac{27\pi}{8}\frac{\left(\mathcal{J}_1+\mathcal{J}_2\right)^2}{\left(2m+\sqrt{3}|Q|\right)^2}
+\sqrt{3}|Q|.
\end{equation}
Moreover if $ab+q\neq 0$, then equality holds if and only if $(M^4,g,k,E,B)$ is isometric to the canonical slice of an extreme charged Myers-Perry spacetime.
\end{theorem}

Brill coordinates, defined in Section \ref{secini}, are a system of cylindrical coordinates in which the metric on the orbit space $M^4/U(1)^2$ takes an isothermal form. They played an indispensable role in the proofs of the $D=3+1$ inequality \eqref{4dinequality}, and were later shown to always exist \cite{chrusciel2008masspositivity,KhuriSokolowsky} as long as the axisymmetric initial data set is simply connected. In the $D=4+1$ case, we strongly suspect that generalizations of \cite{chrusciel2008masspositivity,KhuriSokolowsky} also hold, so that in Theorem \ref{MainTheorem} the hypotheses concerning the diffeomorphism type and existence of Brill coordinates may be replaced with the assumption of simple connectivity or another similar condition.

The proof of Theorem \ref{MainTheorem} also yields a slightly different inequality (with corresponding rigidity statement), which in some circumstances produces an improved lower bound for the mass, namely
\begin{equation}\label{maininequality1}
m\geq \frac{27\pi}{8}\frac{\left(|\mathcal{J}_1|+|\mathcal{J}_2|\right)^2}{\left(2m+\sqrt{3}Q\right)^2}
+\sqrt{3}Q.
\end{equation}
As mentioned above $m^2\neq\frac{3}{4}Q^2$, and so the denominator on the right-hand side does not vanish. It should be pointed out that this inequality reduces to \eqref{vacuumineq} when $Q=0$, and that \eqref{maininequality} does not necessarily have this property. However \eqref{maininequality} implies the so called BPS bound in supergravity, which in our conventions reads
\begin{equation}\label{BPSbound}
m \geq \sqrt{3} |Q|.
\end{equation}
The BPS bound has previously been established \cite{Gibbons:1993xt} using completely different methods, specifically the spinorial approach developed in Witten's proof of the positive mass theorem \cite{witten1981new}. In particular, this bound is known to hold without any symmetry assumptions or restrictions on $M^4$ apart from the existence of a spin structure.  Solutions saturating \eqref{BPSbound} must be `supersymmetric', that is, they admit Killing spinor fields; note that supersymmetric black holes are necessarily extreme. It turns out that the spinor proof of the BPS bound has not yielded an associated rigidity statement as in Theorem \ref{MainTheorem}, and indeed there are distinct families of solutions which saturate the bound. Thus, our result, which does treat the case of equality may be viewed as a refinement of the BPS bound in the setting of bi-axisymmetry. If \eqref{BPSbound} is saturated, then $\mathcal{J}_{1}=-\mathcal{J}_{2}$ and there are two cases to consider. When $\mathcal{J}_{1}=\mathcal{J}_{2}= 0$ the initial data must arise from an extreme Reissner-Nordstr\"{o}m spacetime, while if $\mathcal{J}_{l}\neq 0$ the initial data are a special subclass of the extreme charged Myers-Perry solutions in which the two angular momenta differ by a sign. These latter spacetimes form a two-parameter family of supersymmetric solutions known as the BMPV black holes \cite{Breckenridge:1996is}. It follows that we obtain a new characterization of these solutions.

\begin{cor}
Under the hypotheses of Theorem \ref{MainTheorem}, the BPS bound \eqref{BPSbound} holds and is saturated if and only if the initial data set is isometric to the canonical slice of an extreme Reissner-Nordstr\"{o}m spacetime (vanishing angular momentum) or BMPV black hole (nonvanishing angular momentum).
\end{cor}

There is another known and important class of solutions with vanishing angular momenta that saturate \eqref{BPSbound}, namely the supersymmetric multi-black hole spacetimes \cite{Gibbons:1994vm} which generalize the Majumdar-Papapetrou solutions of the 4D Einstein-Maxwell equations. Their associated initial data are not covered in our analysis, and would require a strengthening of our results to the situation when $M^4$ is diffeomorphic to $\mathbb{R}^4$ with multiple points removed. Such an analysis should be possible, and has already been carried out in the $D=3+1$ case \cite{chrusciel2008mass,khuri2015positive}. Another direction for possible improvement of
Theorem \ref{MainTheorem} would be to remove the maximal assumption $\tr_{g} k=0$. Again, progress has already been made in the $D=3+1$ case \cite{cha2014deformations,cha2015deformations}, and perhaps similar methods can be applied in the current setting.

This paper is organized as follows. Section \ref{secsugra} introduces 5-dimensional minimal supergravity, and summarizes the appropriate initial data constraint equations. Section \ref{secini} discusses in detail the hypotheses imposed upon, and consequences for, the initial data, and Section \ref{secpot} establishes the existence of potentials and properties of the charges.
In Section \ref{massfun} we derive a lower bound for the mass in terms of a functional  related to the Dirichlet energy of a map from $\mathbb{R}^{3}\rightarrow G_{2(2)}/SO(4)$. Sections \ref{dirichlet}, \ref{seccut}, and \ref{secproof} are then
dedicated to proving that the extreme charged Myers-Perry harmonic map realizes the absolute minimum of the mass functional. Lastly an appendix is included to record, among other things, important properties of the charged Myers-Perry black holes.

\section{Five Dimensional Minimal Supergravity}
\label{secsugra}

In this section the relevant concepts of 5D minimal supergravity will be presented. In particular we derive the constraint equations satisfied by initial data. The bosonic field content of this theory consists of a spacetime metric $\tilde{g}_{ab}$ and a closed 2-form Maxwell field $F_{ab}$. It will be assumed that the spacetime $M^5$ possesses no nontrivial 2-cycles, so that $dF=0$ implies the existence of a globally defined vector potential $F = dA$.  The action \cite{Tomizawa:2009ua} is that of Einstein-Maxwell theory together with a Chern-Simons term, and is given by
\begin{equation}\label{sugraaction}
S = \int_{M^5} \tilde{R} \star_5 1 - \frac{1}{2} F \wedge \star_5 F - \frac{1}{3\sqrt{3}} F \wedge F \wedge A \;.
\end{equation}
where $\tilde{R}$ is the scalar curvature of $\tilde{g}$ and $\star_5$ denotes the spacetime Hodge star operation. The field equations are
\begin{equation}\label{FIELDeqns}
\tilde{R}_{ab}-\frac{1}{2}\tilde{R}\tilde{g}_{ab}
=\frac{1}{2}F_{ac}F_{b}^{\phantom{b}c}-\frac{1}{8}|F|^{2}\tilde{g}_{ab},
\end{equation}
\begin{equation}\label{FIELDeqns1}
d \star_5 F + \frac{1}{\sqrt{3}} F \wedge F=0,\text{ }\text{ }\text{ }\text{ }\text{ }
dF=0,
\end{equation}
where $\tilde{R}_{ab}$ denotes the Ricci tensor. Note that in contrast to pure Einstein-Maxwell theory, $d \star_5 F \neq 0$. It will be convenient to define
$H = \star_5 F$.  With this the field equations may be rewritten as
\begin{equation}\label{Fieldeqns1}
\tilde{R}_{ab}-\frac{1}{2}\tilde{R}\tilde{g}_{ab}
= \frac{1}{8}H_{acd}H_b^{\phantom{b}cd} + \frac{1}{4}F_{ac}F_b^{\phantom{b}c},
\end{equation}
\begin{equation}\label{fieldeqns2}
\tilde{\nabla}^b F_{ba}+ \frac{1}{2\sqrt{3}} F^{bc} H_{abc} = 0 ,\text{ }\text{ }\text{ }\text{ }\text{ }
\tilde{\nabla}^a H_{abc}=0,
\end{equation}
where $\tilde{\nabla}$ is the metric connection associated to $\tilde{g}$. Before proceeding further, we mention that throughout this section and the next, there will be numerous computations involving differential forms; the relevant conventions and useful formulae are recorded in Appendix \ref{app4}.

Let $M^4$ be a spacelike hypersurface with unit normal $n$ and induced Riemannian metric $g_{ab}=\tilde{g}_{ab}+n_a n_b$. The constraint equations associated with this surface are the nondynamical equations of \eqref{Fieldeqns1} and \eqref{fieldeqns2}, and are found by contracting each of the three sets of equations with the normal $n$. Thus, from the Einstein equations \eqref{Fieldeqns1} we obtain the following constraints from the Gauss-Codazzi relations
\begin{equation}
R + (\operatorname{Tr}_g k)^2 - |k|_{g}^2 = 2 T_{ab}n^a n^b,\text{ }\text{ }\text{ }\text{ }\text{ }
\nabla^i\left (k_{ij} - (\operatorname{Tr}_gk) g_{ij}\right) = T_{aj}n^a
\end{equation}
where $R$ and $\nabla$ are the scalar curvature and metric connection of $g$, $k$ is the extrinsic curvature or second fundamental form of $M^4$, $T$ is the stress-energy tensor given by the right-hand side of \eqref{Fieldeqns1}, and the indices $i$ and $j$ represent directions tangential to $M^4$. The electric field 1-form and magnetic field 2-form may be extracted from the field strength tensor in the usual manner
\begin{equation}\label{EandB}
E = \iota_n F, \qquad B = \iota_n \star_5 F = \iota_n H,
\end{equation}
that is $E_b = n^a F_{ab}$ and $B_{ab} = n^c H_{cab}$.  Observe that by construction $E$ and $B$ are spatial, $n_a E^a = n_a B^{ab} = 0$.  We then have
\begin{equation}
T_{ab}n^a n^b = \frac{1}{4}E_aE^a + \frac{1}{8}B_{ab}B^{ab}= \frac{1}{4}|E|_g^2
+\frac{1}{8} |B|_g^2,
\end{equation}
so that the scalar constraint becomes
\begin{equation}\label{constraint1}
R + (\operatorname{Tr}_g k)^2 - |k|_{g}^2 = \frac{1}{2}|E|_g^2 + \frac{1}{4} |B|_g^2
\end{equation}
Moreover the entire Maxwell field may be expressed in terms of the electric and magnetic fields
\begin{equation}
F = -(n \wedge E) + \star_5 (n \wedge B)\;, \qquad H = -\star_5(n \wedge E) - (n \wedge B),
\end{equation}
or equivalently
\begin{equation}\label{E0}
F_{ab} = -2n_{[a}E_{b]} + \frac{1}{2} \epsilon_{abcde}n^c B^{de} \;, \qquad
H_{abc} = -3n_{[a}B_{bc]} - \epsilon_{abcde}n^d E^e,
\end{equation}
with $\epsilon_{abcde}$ the volume form for $\tilde{g}$.
A calculation now shows that
\begin{equation}
T_{j a}n^a = \frac{1}{4}n^a \epsilon_{aj cde}E^c B^{de} =\frac{1}{4} \epsilon_{j cde}E^c B^{de},
\end{equation}
where $\epsilon_{\beta cde}$ represents the volume form of $g$ \footnote{Here we use the convention that the pullback to $M^4$ of the spacetime volume form satisfies $\iota_n \text{Vol}(\tilde{g}) = \text{Vol} (g)$.}. It follows that the momentum constraint is
\begin{equation}\label{constraint2}
\nabla^i \left(k_{ij} - (\operatorname{Tr}_g k) g_{ij}\right) = \frac{1}{2} \star(E \wedge B)_j,
\end{equation}
in which $\star$ is the Hodge star operation on the slice.

Now consider the constraints arising from the Maxwell equations.  First contract the second equation of \eqref{fieldeqns2} with the normal to obtain
\begin{equation}\label{E1}
0=n^b \tilde{\nabla}^a H_{abc} = -(\tilde{\nabla}^{a}n^{b})H_{abc}
+\tilde{\nabla}^{a}(n^{b}H_{abc}).
\end{equation}
Since $H$ is antisymmetric $n^b n^{d}n_{l} \tilde{\nabla}^{l}H_{dbc}=0$, and so the sum over $a$ in \eqref{E1} need only be performed for directions $i$ tangential to $M^4$. It follows that $\tilde{\nabla}^{i}n^{b}$ represents the second fundamental form $k$ which is symmetric, and this implies that $(\tilde{\nabla}^{i}n^{b})H_{i bc}=0$. Then using \eqref{E0} produces the desired constraint
\begin{equation}\label{constraint3}
0=\tilde{\nabla}^{a}B_{ac}=\nabla^{i}B_{i c}-n^{a}n^{b}\tilde{\nabla}_{a}B_{bc}
-n_{c}n^{b}\tilde{\nabla}^{a}B_{ab}=\nabla^{\mathfrak{i}}B_{i c}=(\operatorname{div}B)_{c}.
\end{equation}
Similarly, contract the first equation of \eqref{fieldeqns2} with the normal to obtain
\begin{equation}\label{constraint4}
\operatorname{div}E = \nabla^i E_i=\frac{1}{2\sqrt{3}}F^{bc}B_{bc}
 = \frac{1}{4\sqrt{3}}n^a \epsilon_{abcde}B^{bc}B^{de} = \frac{1}{4\sqrt{3}}\epsilon_{bcde}B^{bc}B^{de}=\frac{1}{\sqrt{3}}\star (B \wedge B).
\end{equation}

Equations \eqref{constraint1}, \eqref{constraint2}, \eqref{constraint3}, and \eqref{constraint4} comprise the full set of constraint equations for the pure minimal supergravity theory. However for the purposes of this paper the presence of addition matter fields will be taken into account, as long as the additional matter is neutral and hence does not source the Maxwell field. It then becomes convenient to separate out contributions from the Maxwell field to the energy and momentum densities. We thus rewrite the constraint equations as
\begin{equation}\label{Const1}
16\pi\mu_{SG} = R+(\tr_{g}k)^{2}-|k|_{g}^{2}
-\frac{1}{2}|E|_g^2 - \frac{1}{4} |B|_g^2,
\end{equation}
\begin{equation}\label{Const2}
8\pi J_{SG} = \operatorname{div}_{g}(k-(\tr_{g}k)g)-\frac{1}{2} \star(E \wedge B),
\end{equation}
\begin{equation}\label{Const3}
\operatorname{div}_{g}E =\frac{1}{\sqrt{3}}\star (B \wedge B),\text{ }\text{ }\text{ }\text{ }\text{ }\text{ }\operatorname{div}_{g}B=0,
\end{equation}
where $\mu_{SG}$ and $J_{SG}$ are, respectively, the energy and momentum densities of the nonelectromagnetic matter fields. The equations \eqref{Const3} may be interpreted as stating that charged matter is not present.

\section{The Initial Data}
\label{secini}

An initial data set $(M^4,g,k,E,B)$ for the 5-dimensional minimal supergravity equations
consists of a Riemannian manifold $M^4$, with metric $g$, a symmetric 2-tensor $k$ denoting the second fundamental form, a 1-form and 2-form $E$ and $B$ representing the electric and magnetic fields, respectively, all of which satisfy the constraint equations \eqref{Const1}, \eqref{Const2}, and \eqref{Const3}. The initial data set is assumed to possess a $U(1)^2$ symmetry generated by two Killing fields $\eta_{(l)}$, $l=1,2$, that is
\begin{equation}\label{Killing}
\mathfrak{L}_{\eta_{(l)}}g=\mathfrak{L}_{\eta_{(l)}}k
=\mathfrak{L}_{\eta_{(l)}}\mu_{SG}=\mathfrak{L}_{\eta_{(l)}}J_{SG}
=\mathfrak{L}_{\eta_{(l)}}E=\mathfrak{L}_{\eta_{(l)}}B=0,
\end{equation}
where $\mathfrak{L}_{\eta_{(l)}}$ denotes Lie differentiation. In order to incorporate the presence of a black hole, the manifold $M^4$ will have two ends, one asymptotically flat and the other either asymptotically flat or asymptotically cylindrical.
We will also postulate that
$M^4$ has two ends, with one designated end being asymptotically flat, and the other being either asymptotically
flat or asymptotically cylindrical. A region $M^{4}_{\text{end}}\subset M^4$ diffeomorphic to $\mathbb{R}^{4}\setminus\text{Ball}$ is called asymptotically flat, if there exist coordinates such that the following fall-off conditions hold
\begin{equation}\label{AF1}
g_{ij}=\delta_{ij}+O_{1}(r^{-1-\kappa}),\text{ }\text{ }\text{ }\text{ }\text{ }
k_{ij}=O(r^{-2-\kappa}),\text{
}\text{ }\text{ }\text{ }
\text{ }E_{i}=O(r^{-2-\kappa}),\text{ }\text{ }\text{ }\text{ }\text{ }
B_{ij}=O(r^{-2-\kappa}),
\end{equation}
\begin{equation}\label{AF2}
\mu_{SG}\in L^{1}(M^{4}_{\text{end}}),\text{ }\text{ }\text{ }\text{ }\text{ }
J_{SG}^{i}\in L^{1}(M^{4}_{\text{end}}),\text{ }\text{ }\text{ }\text{ }\text{ }
\operatorname{div}_{g}E-\frac{1}{\sqrt{3}}\star(B\wedge B)\in L^{1}(M^{4}_{\text{end}}),
\end{equation}
for some $\kappa>0$.
These asymptotics guarantee that the ADM energy and linear momentum, as well as the total electric charge are all well-defined. Due to the simple topology of the initial data, in particular the lack of nontrivial 2-cycles, the total magnetic charge always vanishes.

The asymptotics for cylindrical ends are most easily described in Brill coordinates, which we now describe. A basic hypothesis of Theorem \ref{MainTheorem} is the existence of a Brill coordinate system \cite{alaee2014thesis,Alaeeremarks2015}, that is, a global set of cylindrical coordinates $(\rho,z,\phi^{1},\phi^{2})$ in which the metric takes the form
\begin{equation}\label{BrillMetric}
g=\frac{e^{2U+2\alpha}}{2\sqrt{\rho^2+z^2}}\left(d\rho^2+d z^2\right)+e^{2U}\lambda_{ij}\left(d\phi^i+A^i_l d y^l\right)\left(d\phi^j+A^j_l d y^l\right),
\end{equation}
for some functions $U$, $\alpha$, $A_{l}^{i}$, and a symmetric positive definite matrix $\lambda=(\lambda_{ij})$ with $\det\lambda=\rho^{2}$, $i,j,l=1,2$, $(y^1,y^2)=(\rho,z)$. All of the quantities involved satisfy the asymptotics \eqref{Asym1}-\eqref{Asym12}, and are independent of $(\phi^{1},\phi^{2})$, which are the coordinates for the $U(1)^2$ generators $\eta_{(l)} = \partial_{\phi^l}$, $l=1,2$. Furthermore, the values of the coordinate functions are restricted to the ranges $\rho\in[0,\infty)$, $z\in \mathbb{R}$, and $\phi^i\in [0,2\pi]$, $i=1,2$. Brill coordinates may also be expressed in polar form through the transformation
\begin{equation}\label{pcoordinates}
\rho=\frac{1}{2}r^{2}\sin(2\theta),\text{ }\text{ }\text{ }\text{ }\text{ }\text{ }\text{ }
z=\frac{1}{2}r^{2}\cos(2\theta),\text{ }\text{ }\text{ }\text{ }\text{ }\text{ }\text{ }
r^{2}=2\sqrt{\rho^{2}+z^{2}},
\end{equation}
where $r\in [0,\infty)$, $\theta\in[0,\pi/2]$. For instance, the flat metric on $\mathbb{R}^{4}$ is given in these two sets of coordinates by
\begin{equation}\label{FlatMetric}
\delta_4= \frac{d\rho^2 + d z^2}{2\sqrt{\rho^2 + z^2} } + \sigma_{ij}d\phi^{i} d\phi^{j}=
d r^2 + r^2d\theta^2 +r^{2}\left(\sin^{2}\theta (d\phi^1)^{2}+\cos^{2}\theta (d\phi^2)^{2}\right),
\end{equation}
where $\sigma_{ij}$ is defined by the second equality.

There are three different asymptotic regimes of interest, namely near infinity, the origin, and the axes $\Gamma_{\pm}=\{\rho=0, \pm z>0\}$. Consider first the asymptotics near infinity as $r\rightarrow\infty$. In this region the initial data set is asymptotically flat, which motivates the requirements
\begin{equation}\label{Asym1}
U=O_{1}(r^{-1-\kappa}),\text{ }\text{ }\text{ }\text{ }\text{ }
\alpha=O_1(r^{-1-\kappa}),\text{ }\text{ }\text{ }\text{ }\text{ }
A_{\rho}^i=\rho O_{1}(r^{-5-\kappa}),
\text{ }\text{ }\text{ }\text{ }\text{ }
A_z^i=O_{1}(r^{-3-\kappa}),
\end{equation}
\begin{equation}\label{Asym2}
\lambda_{ii}=\left(1+(-1)^{i}c_{0}r^{-1-\kappa}+O_{1}(r^{-2-\kappa})\right)\sigma_{ii},\text{ }\text{ }\text{ }\text{ }\text{ }\text{ }\text{ }\lambda_{12}=\rho^2 O_1(r^{-5-\kappa}),
\end{equation}
\begin{equation}\label{Asym3}
|k|_{g}=O(r^{-2-\kappa}),\text{ }\text{ }\text{ }\text{ }\text{ }\text{ }\text{ }
|E|_{g}=O(r^{-2-\kappa}),\text{ }\text{ }\text{ }\text{ }\text{ }\text{ }\text{ }
|B|_{g}=O(r^{-2-\kappa}),
\end{equation}
where $c_{0}$ is a function of $\theta$. For the asymptotics as $r\rightarrow 0$ we require, in the asymptotically flat case
\begin{equation}\label{Asym4}
U=-2\log r+O_{1}(1),\text{ }\text{ }\text{ }\text{ }\text{ }
\alpha=O_1(r^{1+\kappa}),\text{ }\text{ }\text{ }\text{ }\text{ }
A_{\rho}^i=\rho O_{1}(r^{1+\kappa}),
\text{ }\text{ }\text{ }\text{ }\text{ }
A_z^i=O_{1}(r^{3+\kappa}),
\end{equation}
\begin{equation}\label{Asym5}
\lambda_{ii}=\left(1+(-1)^{i}c_{1}r^{1+\kappa}+O_{1}(r^{2+\kappa})\right)\sigma_{ii},\text{ }\text{ }\text{ }\text{ }\text{ }\text{ }\text{ }\lambda_{12}=\rho^2 O_1(r^{-\frac{1}{2}+\kappa}),
\end{equation}
\begin{equation}\label{Asym6}
|k|_{g}=O(r^{2+\kappa}),\text{ }\text{ }\text{ }\text{ }\text{ }\text{ }\text{ }
|E|_{g}=O(r^{2+\kappa}),\text{ }\text{ }\text{ }\text{ }\text{ }\text{ }\text{ }
|B|_{g}=O(r^{2+\kappa}),
\end{equation}
where $c_{1}$ is a function of $\theta$, and in the asymptotically cylindrical case
\begin{equation}\label{Asym7}
U=-\log r+O_{1}(1),\text{ }\text{ }\text{ }\text{ }\text{ }\text{ }
\alpha=O_{1}(1),\text{ }\text{ }\text{ }\text{ }\text{ }\text{ }A_{\rho}^i=\rho O_{1}(r^{1+\kappa}),
\text{ }\text{ }\text{ }\text{ }\text{ }\text{ }
A_z^i=O_{1}(r^{3+\kappa}),
\end{equation}
\begin{equation}\label{Asym8}
\lambda_{ij}=r^{2}\tilde{\sigma}_{ij}+O_{1}(r^{2+\kappa}),\text{ }\text{ }\text{ }\text{ }\text{ }\text{ }\text{ }|k|_{g}=O(r^{2+\kappa}),\text{ }\text{ }\text{ }\text{ }\text{ }\text{ }
|E|_{g}=O(r^{2+\kappa}),\text{ }\text{ }\text{ }\text{ }\text{ }\text{ }
|B|_{g}=O(r^{2+\kappa}),
\end{equation}
where $\tilde{\sigma}$ is a positive definite metric on the 2-torus depending only on $\theta$. Lastly, the asymptotics near the axes as $\rho\rightarrow 0$ are required to satisfy
\begin{equation}\label{Asym10}
U=O_{1}(1),\text{ }\text{ }\text{ }\text{ }\text{ }\text{ }
\alpha=O_1(1),\text{ }\text{ }\text{ }\text{ }\text{ }\text{ }
A_{\rho}^i= O_{1}(\rho),
\text{ }\text{ }\text{ }\text{ }\text{ }\text{ }
A_z^i=O_{1}(1),\text{ }\text{ }\text{ }\text{ }\text{ }\text{ }
|k|_{g},|E|_{g},|B|_{g}=O(1),
\end{equation}
\begin{equation}\label{Asym12}
\lambda_{11},\lambda_{12}=O(\rho^{2}),\text{ }\text{ }\text{ }\text{ }\text{ }
\lambda_{22}=O(1)
\text{ }\text{ }\text{ on }\text{ }\text{ }\Gamma_{+},\text{ }\text{ }\text{ }\text{ }\text{ }
\lambda_{22},\lambda_{12}=O(\rho^{2}),\text{ }\text{ }\text{ }\text{ }\text{ }
\lambda_{11}=O(1)
\text{ }\text{ }\text{ on }\text{ }\text{ }\Gamma_{-}.
\end{equation}

It is shown in \cite{alaee2015global} that the Brill coordinate asymptotics \eqref{Asym1}-\eqref{Asym6}, associated with asymptotically flat regions, are consistent  with the asymptotics \eqref{AF1} and \eqref{AF2} used in the definition of asymptotically flat ends. Finally, we mention that the asymptotics \eqref{Asym10} and \eqref{Asym12} are not sufficient to guarantee regularity of the geometry at the axes, that is, the absence of conical singularities. For this, a compatibility condition \cite{alaee2015global} is needed
\begin{equation}\label{ConeCondition}
\alpha(0,z)=\frac{1}{2}\log\left(|z|\partial_{\rho}^{2}\lambda_{ii}(0,z)\right)=:\alpha_{\pm}(z)\text{ }\text{ }\text{ }\text{ on }\text{ }\text{ }\text{ }\Gamma_{\pm},
\end{equation}
where $i=1,2$ corresponds to $\Gamma_{+},\Gamma_{-}$, respectively. Thus, under the assumptions of Theorem \ref{MainTheorem} in which the geometry is smooth, \eqref{ConeCondition} holds.

\section{Potentials and Charges}
\label{secpot}

The $U(1)^2$ symmetry, together with the lack of charged matter and the vanishing of the momentum density in Killing directions, $J_{SG}(\eta_{(l)})=0$ for $l=1,2$, guarantees that potentials exist for portions of the electric and magnetic fields as well as for parts of the second fundamental form $k$. Moreover, these potentials encode the total charge and angular momentum of the data. In this section we will establish the global existence of such potentials and give their relationship with the charges. While this will be carried out here from the `initial data point of view', we note that the same constructions may also be accomplished from the `spacetime perspective' as is demonstrated in Appendix \ref{app3}.

We begin with the magnetic field. Observe that by \eqref{Const3}, \eqref{Killing}, and Cartan's formula
\begin{equation}\label{Bequation}
d \left(\iota_{\eta_{(i)}} \star B\right) = \mathfrak{L}_{\eta_{(i)}}\star B
-\iota_{\eta_{(i)}} d \star B  = 0.
\end{equation}
Since $H_1(M^4)$ is trivial, there exists a globally defined potential function such that
\begin{equation}\label{Bpotential}
d\psi^i  = \iota_{\eta_{(i)}} \star B.
\end{equation}
We emphasize that the index $i$ is a label here, and not a tensor index that is raised and lowered with the metric. Observe further that $d\psi^i (\eta_{(l)}) = 0$.
To see this use \eqref{Killing} and \eqref{Bequation} to compute
\begin{equation}\label{starBconst}
d \left(\iota_{\eta_{(1)}} \iota_{\eta_{(2)}} \star B\right) =
\mathfrak{L}_{\eta_{(1)}}\iota_{\eta_{(2)}} \star B
-\iota_{\eta_{(1)}}d\left(\iota_{\eta_{(2)}} \star B\right)= 0.
\end{equation}
It follows that the function $\iota_{\eta_{(1)}} \iota_{\eta_{(2)}} \star B$ is constant.
Moreover, \eqref{Asym12} implies that $|\eta_{(1)}|_{g}=0$ on $\Gamma_{+}$ and $|\eta_{(2)}|_{g}=0$ on $\Gamma_{-}$, so that this constant is zero. Hence
\begin{equation}
\mathfrak{L}_{\eta_{(l)}} \psi^i=0,\text{ }\text{ }\text{ }\text{ }\text{ }i,l=1,2,
\end{equation}
showing that the potentials must be invariant under the $U(1)^2$ action.

Consider next the electric field. Similar calculations as those used above with the magnetic field, combined with the electric field constraint \eqref{Const3}, produce
\begin{equation}
d\left( \iota_{\eta_{(2)}} \iota_{\eta_{(1)}}  \star E \right)  =  \iota_{\eta_{(2)}} \iota_{\eta_{(1)}}  d \star E =   \iota_{\eta_{(2)}} \iota_{\eta_{(1)}} \left(\frac{-1}{\sqrt{3}}B \wedge B\right).
\end{equation}
Then noting the identity $B \wedge B = \star B \wedge \star B$ for any 2-form $B$ on a 4-manifold, and using \eqref{Bpotential} yields
\begin{equation}
d\left( \iota_{\eta_{(2)}} \iota_{\eta_{(1)}}   \star E\right)  =  -\frac{1}{\sqrt{3}} \iota_{\eta_{(2)}} \iota_{\eta_{(1)}} (\star B \wedge \star B) = \frac{2}{\sqrt{3}} d\psi^1 \wedge d\psi^2 = \frac{1}{\sqrt{3}} d \left(\psi^1 d\psi^2-\psi^2 d \psi^1\right).
\end{equation}
Thus, there exists a globally defined potential function with
\begin{equation}\label{Epotential}
d\chi =  \iota_{\eta_{(2)}} \iota_{\eta_{(1)}}  \star E - \frac{1}{\sqrt{3}} \left( \psi^1 d\psi^2  -\psi^2 d \psi^1 \right).
\end{equation}
Moreover it immediately follows that this potential is invariant under the $U(1)^2$ symmetry
\begin{equation}
\mathfrak{L}_{\eta_{(l)}}\chi=0, \text{ }\text{ }\text{ }\text{ }\text{ }l=1,2.
\end{equation}

Finally we demonstrate the existence of charged twist potentials for the second fundamental form $k$, which encode the angular momentum contained in the initial data.  Define
\begin{equation}
\mathcal{P}^l =2\star\left(p(\eta_{(l)}) \wedge \eta_{(1)} \wedge \eta_{(2)}\right), \qquad p = k - (\operatorname{Tr}_{g} k) g,
\end{equation}
where again $l=1,2$ here is not a tensor index but is rather a label.
It can be shown \cite{alaee2015global} that
\begin{equation}
d\mathcal{P}^l = -2\iota_{\eta_{(1)}} \iota_{\eta_{(2)}} d \star p(\eta_{(l)}) = -2\iota_{\eta_{(1)}} \iota_{\eta_{(2)}}\star \left(\star d \star p(\eta_{(l)})\right).
\end{equation}
Now from the constraint equation \eqref{Const2}, the hypothesis $J_{SG}(\eta_{(l)})=0$, and the fact that $\eta_{(l)}$ is a Killing field, we have
\begin{equation}\label{12345}
-\star d \star p(\eta_{(l)})  = \operatorname{div}_{g}p(\eta_{(l)})   = 8\pi J_{SG}(\eta_{(l)})+\frac{1}{2}\iota_{\eta_{(l)}} \star (E \wedge B)
=\frac{1}{2}\iota_{\eta_{(l)}} \star (E \wedge B)
\end{equation}
and hence
\begin{equation}
d\mathcal{P}^l = \iota_{\eta_{(1)}} \iota_{\eta_{(2)}}\star \left(\iota_{\eta_{(l)}} \star (E \wedge B)\right).
\end{equation}
Let us now compute the right-hand side in terms of the electromagnetic potentials derived above. Observe that
\begin{equation}\label{123456}
\iota_{\eta_{(l)}} \star (E \wedge B) = \frac{1}{2}\eta^{i}_{(l)}\epsilon_{ijns}
E^{j}B^{ns} = E^{j} (\iota_{\eta_{(l)}} \star B)_{j} = d\psi^l(E)
\end{equation}
and
\begin{equation}\label{1234567}
\star d\psi^l(E) = d\psi^l(E) \epsilon = \star E \wedge d\psi^l,
\end{equation}
so that
\begin{align}
\begin{split}
d\mathcal{P}^l =&   \iota_{\eta_{(1)}} \iota_{\eta_{(2)}} \left(\star E \wedge d\psi^l\right) \\
=&  \left(\iota_{\eta_{(1)}} \iota_{\eta_{(2)}} \star E \right) \wedge d\psi^l \\
=&  d\psi^{l} \wedge  \left( d\chi + \frac{1}{\sqrt{3}}\left(\psi^1 d\psi^2 - \psi^2 d\psi^1\right) \right) \\
=&  d\left[ \psi^l \left( d\chi + \frac{1}{3\sqrt{3}}\left(\psi^1 d\psi^2 - \psi^2 d\psi^1\right) \right)\right].
\end{split}
\end{align}
Therefore, for $l=1,2$, a globally defined potential function exists such that
\begin{equation}\label{chargedtwistpotentials}
d\zeta^l = \mathcal{P}^l -  \psi^l \left( d\chi + \frac{1}{3\sqrt{3}}\left(\psi^1 d\psi^2 - \psi^2 d\psi^1\right) \right),
\end{equation}
and it is clear that these potentials are also $U(1)^2$ invariant
\begin{equation}
\mathfrak{L}_{\eta_{(i)}} \zeta^l=0,\text{ }\text{ }\text{ }\text{ }\text{ }i,l=1,2.
\end{equation}

Having constructed a total of five potential functions, one for the electric field and two each for the magnetic field and second fundamental form, we will now examine exactly which components of the initial data are determined by the potentials. In order to do this it will be convenient to introduce the following frame field associated with Brill coordinates
\begin{equation}
e_1= e^{-U-\alpha+\log r}\left(\partial_\rho - A^{i}_{\rho} \partial_{\phi^i}\right), \text{ }\text{ }\text{ }e_2 = e^{-U-\alpha+\log r}\left(\partial_z - A^{i}_{z} \partial_{\phi^i}\right),\text{ }\text{ }\text{ }
 e_{i+2} = e^{-U} \partial_{\phi^i},\text{ }\text{ }i=1,2,
\end{equation}
with dual co-frame
\begin{equation}
\theta^1= e^{U+\alpha-\log r}d\rho,\text{ } \text{ } \text{ } \text{ } \text{ } \text{ }
\theta^2= e^{U+\alpha-\log r}dz,\text{ }\text{ }\text{ }\text{ }\text{ }\text{ }
\theta^{i+2}=e^{U}\left( d\phi^i+A^i_l d y^l\right),\text{ }\text{ }\text{ }\text{ }i=1,2,
\end{equation}
in which the metric takes the form
\begin{equation}
g=(\delta_{2})_{ln}\theta^{l}\theta^{n}+\lambda_{ij}\theta^{i+2}\theta^{j+2}.
\end{equation}
Consider first the magnetic field. In index notation \eqref{Bpotential} becomes
\begin{equation}\label{1.11}
d\psi^{i} = -\frac{1}{2}\epsilon_{jlnt}\eta_{(i)}^{l} B^{nt}
\theta^{j} \\
= -\frac{1}{2}e^U \epsilon_{j(i+2)nt}B^{nt} \theta^{j}.
\end{equation}
Since $\psi^i$ is invariant under the $U(1)^2$ symmetry, the indices $n$ and $t$ for $B$ can only take the values 1 and 2. Thus
\begin{equation}
d\psi^{i} = e^{2U + \alpha - \log r} \epsilon_i^{\phantom{i}l}\left(B_{2(l+2)} d\rho - B_{1 (l+2)}dz\right),
\end{equation}
where $\epsilon_{ij}$ is the volume form associated to the metric $\lambda_{ij}$. Then applying $\epsilon^i_{\phantom{i}j}$ to both sides yields
\begin{align}\label{Bcomponents}
\begin{split}
B_{1(j+2)}=B(e_1,e_{j+2}) =& -e^{-2U - \alpha + \log r}\epsilon^i_{\phantom{i}j}\partial_z \psi^{i},\\
B_{2(j+2)}=B(e_2,e_{j+2}) =& e^{-2U - \alpha + \log r}\epsilon^i_{\phantom{i}j}\partial_\rho \psi^{i}.
\end{split}
\end{align}
Furthermore, the condition $\star B(\eta_{(i)},\eta_{(j)}) =0$ implies that
$B_{12}=0$. Turn now to the electric field. In the frame basis
\begin{equation}
  \iota_{\eta_{(2)}} \iota_{\eta_{(1)}}  \star E = e^{2U} \epsilon_{34ij}E^j \theta^i = e^{2U} \rho (E^2 \theta^1 - E^1 \theta^2),
\end{equation}
from which we obtain
\begin{align}\label{Ecomponents}
\begin{split}
  E_1=E(e_{1}) =& -\frac{e^{-3U - \alpha + \log r}}{\rho} \left(\partial_z \chi + \frac{1}{\sqrt{3}}(\psi^1 \partial_z \psi^2 - \psi^2 \partial_z \psi^1)\right),\\
  E_2=E(e_{2})  =& \frac{e^{-3U - \alpha + \log r}}{\rho} \left(\partial_\rho \chi + \frac{1}{\sqrt{3}}(\psi^1 \partial_\rho \psi^2 - \psi^2 \partial_\rho \psi^1)\right).
\end{split}
\end{align}
Similar considerations applied to \eqref{chargedtwistpotentials} produce expressions for certain components of the second fundamental form
\begin{align}\label{kcomponents}
\begin{split}
k_{2(i+2)} = k(e_2,e_{i+2}) =& \frac{e^{-4U - \alpha + \log r}}{2\rho}\left[ \partial_\rho \zeta^i + \psi^{i}\left(\partial_\rho \chi  + \frac{1}{3\sqrt{3}}(\psi^1\partial_\rho \psi^2 - \psi^2 \partial_\rho \psi^1)\right)\right], \\
k_{1(i+2)} = k(e_1,e_{i+2}) =& -\frac{e^{-4U - \alpha + \log r}}{2\rho}\left[ \partial_z \zeta^i + \psi^{i}\left(\partial_z \chi  + \frac{1}{3\sqrt{3}}(\psi^1\partial_z \psi^2 - \psi^2 \partial_z \psi^1)\right)\right].
\end{split}
\end{align}

It will now be described how the potentials encode the total charge and angular momentum.
First consider the electric charge. The definition \eqref{totalcharge} of total charge may be motivated by the constraint equation \eqref{Const3}
\begin{equation}
d \star E = - \frac{1}{\sqrt{3}}(B \wedge B) = -\frac{1}{\sqrt{3}} (\star B \wedge \star B).
\end{equation}
Since $H_{2}(M^{4})$ is trivial and $d\star B =0$ there exists a globally defined vector potential such that $\star B = d\vec{A}$. It follows that
\begin{equation}
d \left( \star E + \frac{1}{\sqrt{3}} \vec{A} \wedge \star B \right) =0.
\end{equation}
This suggests the following definition of total charge contained within a 3-cycle $\mathcal{S}^{3}$:
\begin{equation}\label{CHARGE}
\tilde{Q}(\mathcal{S}^3) = \frac{1}{16\pi} \int_{\mathcal{S}^{3}} \left(\star E + \frac{1}{\sqrt{3}} \vec{A} \wedge \star B\right).
\end{equation}
Via Stoke's theorem we find that $\tilde{Q}(\mathcal{S}^{3}_{1})=\tilde{Q}(\mathcal{S}^{3}_{2})$ for any two homologous 3-cycles, and thus this definition yields conservation of charge. The total charge is then given by $\tilde{Q}=\lim_{r\rightarrow\infty}\tilde{Q}(S^{3}_{r})$, where $S^{3}_{r}$ are coordinate spheres in the asymptotically flat end; see \cite{Marolf:2000cb} for different notions of charge. Although this appears to differ from the classical notion of total charge $Q$ given in \eqref{totalcharge}, the two actually agree $\tilde{Q}=Q$, since according to the asymptotics of Appendix \ref{app1} the extra term $\vec{A} \wedge \star B$ decays sufficiently fast in the limit so as not to yield a contribution.
Note also that even though the expression \eqref{CHARGE} involves the vector potential, it is still gauge invariant. To see this,
consider the gauge transformation $\vec{A}\mapsto \vec{A}+df$ and observe that
\begin{equation}
\int_{S^{3}_{r}}df\wedge \star B=\int_{S^{3}_{r}}\left(d(f\star B)-fd\star B\right)=0,
\end{equation}
since $d\star B=0$ on $M^4$ and $d$ commutes with pullback. In order to relate the charge to the electric potential $\chi$, use Stokes' theorem to find
\begin{align}
\begin{split}
\tilde{Q} =& \lim_{r \to 0} \frac{1}{16\pi} \int_{\partial B(r)} \left(\star E + \frac{1}{\sqrt{3}} \vec{A} \wedge \star B\right) \\
=&  \lim_{r \to 0} \frac{1}{16\pi} \int_{\partial B(r)} \left( E_i  - \frac{1}{\sqrt{3}}\star( \vec{A} \wedge \star B)_i\right)\nu^i dV \\
=& \lim_{r \to 0} \frac{1}{32\pi} \int_{\partial B(1)} \left( E_i  - \frac{1}{\sqrt{3}}\star( \vec{A} \wedge \star B)_i\right)\nu^i  e^{3U + \alpha} r^3 \sin 2\theta \; d\theta \; d\phi^1 \; d\phi^2,
\end{split}
\end{align}
where $\nu$ is the unit normal pointing towards spatial infinity and $B(r)$ is a coordinate ball of radius $r$. From \eqref{Bpotential} it follows that
\begin{equation}
\star B = d\phi^i \wedge d \psi^i = -d\left( \psi^i d\phi^i\right)\text{ }\text{ }\text{ }\text{ }\text{ }\text{ } \Rightarrow
\text{ }\text{ }\text{ }\text{ }\text{ }\text{ }\vec{A} = -\psi^i d\phi^i,
\end{equation}
and so
\begin{equation}
\star (\vec{A} \wedge \star B) = \frac{re^{-3U - \alpha}}{\rho}\left[ \left(\psi^1 \partial_\rho \psi^2  - \psi^2 \partial_\rho \psi^1\right) \theta^2-\left(\psi^1 \partial_z \psi^2  - \psi^2 \partial_z \psi^1\right) \theta^1\right].
\end{equation}
Combining this with \eqref{Ecomponents} produces
\begin{equation}\label{asdfg}
E_1 - \frac{1}{\sqrt{3}}\star(\vec{A} \wedge \star B)_1 =-\frac{e^{-3U - \alpha}r}{\rho} \partial_z \chi, \qquad
E_2 - \frac{1}{\sqrt{3}}\star(\vec{A} \wedge \star B)_2 = \frac{e^{-3U - \alpha}r}{\rho} \partial_\rho \chi,
\end{equation}
and thus
\begin{equation}\label{ChargePotential}
Q =\tilde{Q}= -\lim_{r \to 0} \frac{1}{32\pi} \int_{\partial B(1)}  2 \partial_\theta \chi \; d\theta \; d\phi^1 \; d\phi^2 = \frac{\pi}{4} \left[ \chi(\Gamma_+) - \chi(\Gamma_-)\right].
\end{equation}
The notation $\chi(\Gamma_{\pm})$ suggests that the potential $\chi$ is constant on the axes $\Gamma_{\pm}$, and indeed this is the case in light of \eqref{asdfg}.

None of the results of this paper take into account magnetic charge, for the following reason. The most natural way it seems to define a magnetic charge would be to integrate $\star B$ over a 2-cycle. However since $H_{2}(M^4)=0$, all 2-cycles bound a 3-domain, and since $\star B$ is closed, we find that this integral is zero. This is consistent with there being no magnetic charge from the constraint equations \eqref{Const3}. Thus, in the current setting magnetic charge vanishes. If, on the other hand, the initial data has more complex topology, then each nontrivial homology class in $H_2(M^4)$ yields a well-defined magnetic charge by integrating $\star B$ over a representative. The total magnetic charge may then be defined by summing these `local' charges over all homology classes.

We now discuss angular momentum. Combining \eqref{12345}, \eqref{123456}, and \eqref{1234567} produces
\begin{equation}
d\star p(\eta_{(l)})=-\frac{1}{2}\star E\wedge d\psi^{l}.
\end{equation}
From \eqref{Epotential} we have
\begin{equation}\label{starE}
\star E = d\phi^1 \wedge d\phi^2 \wedge \left( d\chi + \frac{1}{\sqrt{3}}(\psi^1 d \psi^2 - \psi^2 d\psi^1)\right) + \overline{\omega},
\end{equation}
where $\overline{\omega}$ is a closed 3-form with the property that $i_{\eta_{(2)}} i_{\eta_{(1)}}  \overline{\omega}= 0$. It follows that this form has the structure
\begin{equation}\label{,,,}
\overline{\omega} = \overline{\omega}_{ijn} dy^i \wedge dy^j \wedge d\phi^n
\end{equation}
where $y^1=\rho$ and $y^2 = z$, and so $\overline{\omega} \wedge d\psi^l =0$.  Therefore
\begin{equation}
\star E \wedge d\psi^l = -d\phi^1 \wedge d\phi^2 \wedge d\left[ \psi^l \left(d\chi + \frac{1}{3\sqrt{3}}\left(\psi^1 d\psi^2  - \psi^2 d\psi^1\right)\right) \right],
\end{equation}
and as a consequence
\begin{equation}
d\left[\star p(\eta_{(l)})-\frac{1}{2}d\phi^{1}\wedge d\phi^{2}\wedge
\psi^l \left(d\chi + \frac{1}{3\sqrt{3}}\left(\psi^1 d\psi^2  - \psi^2 d\psi^1\right)\right)\right]=0.
\end{equation}
In analogy with charge, and in similarity to the $D=3+1$ case \cite{dain2013lower}, this suggests the following definition of total angular momenta contained within a 3-cycle $\mathcal{S}^{3}$:
\begin{equation}\label{newam}
\tilde{\mathcal{J}}_{l}(\mathcal{S}^3) = \frac{1}{8\pi} \int_{\mathcal{S}^{3}}\left[\star p(\eta_{(l)})-\frac{1}{2}d\phi^{1}\wedge d\phi^{2}\wedge
\psi^l \left(d\chi + \frac{1}{3\sqrt{3}}\left(\psi^1 d\psi^2  - \psi^2 d\psi^1\right)\right)\right],\text{ }\text{ }\text{ }\text{ }\text{ }l=1,2.
\end{equation}
By Stoke's theorem we then find that $\tilde{\mathcal{J}}_{l}(\mathcal{S}^{3}_{1})=\tilde{\mathcal{J}}_{l}(\mathcal{S}^{3}_{2})$ for any two homologous 3-cycles, yielding conservation of angular momentum. The total angular momentum is given by $\tilde{\mathcal{J}}_{l}=\lim_{r\rightarrow\infty}\tilde{\mathcal{J}}_{l}(S^{3}_{r})$. Although this definition of total angular momentum appears to differ from \eqref{admam}, the ADM definition
\begin{equation}
\mathcal{J}_{l}=\frac{1}{8\pi}\int_{S^{3}_{\infty}}\star p(\eta_{(l)}),
\text{ }\text{ }\text{ }\text{ }\text{ }l=1,2,
\end{equation}
the two actually agree $\tilde{\mathcal{J}}_{l}=\mathcal{J}_{l}$. This is due to the fact that the asymptotics (Appendix \ref{app1}) satisfied by the electromagnetic potentials appearing in \eqref{newam}, guarantee that the corresponding integral vanishes in the limit as $r\rightarrow\infty$. Furthermore, using \eqref{kcomponents} we may relate the total angular momenta to the twist potentials as follows
\begin{align}
\begin{split}
\mathcal{J}_{l}=\tilde{\mathcal{J}}_{l}=&
\lim_{r\rightarrow 0}\frac{1}{8\pi} \int_{\partial B(r)}\left[\star p(\eta_{(l)})-\frac{1}{2}d\phi^{1}\wedge d\phi^{2}\wedge
\psi^l \left(d\chi + \frac{1}{3\sqrt{3}}\left(\psi^1 d\psi^2  - \psi^2 d\psi^1\right)\right)\right]\\
=&\lim_{r\rightarrow 0}\frac{1}{8\pi}
\int_{\partial B(r)}\left[k(\partial_{\phi^{l}},\nu)dV-\frac{1}{2}d\phi^{1}\wedge d\phi^{2}\wedge
\psi^l \left(d\chi + \frac{1}{3\sqrt{3}}\left(\psi^1 d\psi^2  - \psi^2 d\psi^1\right)\right)\right]\\
=&\lim_{r\rightarrow0}\frac{1}{16\pi}\int_{\partial B(1)}\partial_{\theta}\zeta^{l} d\theta d\phi^{1}d\phi^{2}\\
=&\frac{\pi}{4}[\zeta^{l}(\Gamma_{-})-\zeta^{l}(\Gamma_{+})].
\end{split}
\end{align}
In similarity to the electric charge, the notation $\zeta^{l}(\Gamma_{\pm})$ suggests that the potentials $\zeta^{l}$ are constant on the axes $\Gamma_{\pm}$, and indeed this is implied by combining \eqref{Ecomponents}, \eqref{kcomponents}, and \eqref{asdfg}.

Lastly, we remark that the integrand in \eqref{newam} constructed from the scalar potentials, may be rewritten in terms of the physical fields $E$, $B$, and $\vec{A}$. To see this, observe that $\vec{A}(\eta_{(l)}) = -\psi^l$ and
\begin{equation}
\vec{A} \wedge \star B = -  d\phi^1 \wedge d\phi^2 \wedge (\psi^1 d\psi^2 - \psi^2 d\psi^1).
\end{equation}
Using this and \eqref{starE} produces
\begin{equation}
d\phi^1 \wedge d\phi^2 \wedge d\chi = \star E + \frac{1}{\sqrt{3}} \vec{A} \wedge \star B - \overline{\omega}.
\end{equation}
Recall that the closed 3-form $\overline{\omega}$ satisfies \eqref{,,,}, and therefore vanishes when pulled back to any coordinate sphere $r=const$. In particular, this implies that it integrates to zero on any 3-cycle and is therefore exact. We now have
\begin{equation}\label{newam1}
\tilde{\mathcal{J}}_{l} = \frac{1}{8\pi} \int_{S_{\infty}^{3}}\left[\star p(\eta_{(l)})+\frac{1}{2}\vec{A}(\eta_{(l)})\left(\star E  + \frac{2}{3\sqrt{3}}(\vec{A} \wedge \star B)\right)\right],\text{ }\text{ }\text{ }\text{ }\text{ }l=1,2.
\end{equation}

\section{The Mass Functional}
\label{massfun}

A key step in the proof of the main theorem is to relate the ADM mass to the energy of a certain harmonic map, described in detail in the next section. Heuristically the ADM mass arises as the boundary term obtained from integrating the scalar curvature by parts. Thus, a lower bound for the mass may be achieved by estimating the scalar curvature from below. By virtue of the energy condition $\mu_{SG}\geq 0$, such a lower bound may be achieved in terms of the potentials constructed in the previous section, and these together these potentials will form a large part of the harmonic map data; the remaining part of the harmonic map data will come from the metric.

We begin with the observation that \eqref{Const1} and the maximality condition $\operatorname{Tr}_{g}k=0$ imply
\begin{equation}\label{sc}
R=16\pi \mu_{SG}+|k|_{g}^{2}+\frac{1}{2}|E|_{g}^{2}+\frac{1}{4}|B|_{g}^{2}.
\end{equation}
In order to estimate the squared terms on the right-hand side, it will be convenient to adopt the notation
\begin{equation}
\Theta^{i}=\nabla\zeta^i+ \psi^{i} \left( \nabla\chi + \frac{1}{3\sqrt{3}}(\psi^1 \nabla \psi^2 - \psi^2 \nabla \psi^1) \right),\text{ }\text{ }\text{ }\text{ }\text{ }i=1,2,
\end{equation}
\begin{equation}
\Upsilon=\nabla \chi + \frac{1}{\sqrt{3}}(\psi^1 \nabla \psi^2 - \psi^2 \nabla \psi^1),
\end{equation}
and
\begin{equation}\label{d}
\delta_{3}=r^{2}\left(d r^{2}+r^{2}d\theta^{2}\right)+\frac{r^{4}\sin^{2}(2\theta)}{4} d\phi^{2}
=d\rho^{2}+d z^{2}+\rho^{2} d\phi^{2}.
\end{equation}
The metric $\delta_{3}$ is flat on $\mathbb{R}^{3}$ and involves a new auxiliary variable $\phi\in[0,2\pi]$ on which no quantities have a dependence. The reason for introducing this metric is to simplify expressions within the mass functional.
Now using \eqref{Bcomponents}, \eqref{Ecomponents}, \eqref{kcomponents} produces
\begin{equation}
|k|_g^2 =  \frac{e^{-8U - 2\alpha + 2 \log r}}{2\rho^2} \Theta^T \lambda^{-1}\Theta
+\delta_{2}^{ij}\delta_{2}^{ln}k(e_i,e_l)k(e_{j},e_{n}) +\lambda^{ij}\lambda^{ln}k(e_{i+2},e_{l+2})k(e_{j+2},e_{n+2}),
\end{equation}
\begin{equation}
|E|_g^2 = \frac{e^{-6U - 2\alpha + 2\log r}}{\rho^2}|\Upsilon|^2 + \lambda^{ij}E_i E_j,
\end{equation}
and
\begin{equation}
|B|_g^2 = 2e^{-4U -2\alpha + 2 \log r}\nabla \psi^T \lambda^{-1} \nabla \psi+
\lambda^{ij}\lambda^{ln}B_{il}B_{jn},
\end{equation}
where the norm $|\cdot|$ is taken with respect to $\delta_{3}$. Here $\Theta^T=(\Theta^1,\Theta^2)$, $\psi^T=(\psi^1,\psi^2)$, and the upper index $T$ represents the transpose operation so that
\begin{equation}
\Theta^T \lambda^{-1}\Theta=\sum_{i,j=1,2}\lambda^{ij}\delta_{3}(\Theta^{i},\Theta^{j}),
\end{equation}
with a similar expression for $\nabla \psi^T \lambda^{-1} \nabla \psi$. It follows that \eqref{sc} becomes
\begin{align}\label{sc1}
\begin{split}
R=& 16\pi\mu_{SG}+ \frac{e^{-8U - 2\alpha + 2 \log r}}{2\rho^2} \Theta^T \lambda^{-1}\Theta +\frac{e^{-6U - 2\alpha + 2\log r}}{2\rho^2}|\Upsilon|^2
+\frac{e^{-4U -2\alpha + 2 \log r}}{2}\nabla \psi^T \lambda^{-1} \nabla \psi\\
&+\delta_{2}^{ij}\delta_{2}^{ln}k(e_i,e_l)k(e_{j},e_{n}) +\lambda^{ij}\lambda^{ln}k(e_{i+2},e_{l+2})k(e_{j+2},e_{n+2})
+\frac{1}{2}\lambda^{ij}E_i E_j
+\frac{1}{4}\lambda^{ij}\lambda^{ln}B_{il}B_{jn}.
\end{split}
\end{align}

The scalar curvature may also be computed directly from the components of the metric $g$. It is here that the existence of Brill coordinates plays a significant role, namely as shown in \cite{Alaeeremarks2015} we have
\begin{align}\label{scalarcurvature}
\begin{split}
e^{2U+2\alpha-2\log r}R=&-6\Delta U-2\Delta_{\rho,z}\alpha-6|\nabla U|^2
+\frac{\det\nabla\lambda}{2\rho^{2}}\\
&-\frac{1}{4}e^{-2\alpha+2\log r}\lambda_{ij}(A_{\rho,z}^{i}-A_{z,\rho}^{i})
(A_{\rho,z}^{j}-A_{z,\rho}^{j}),
\end{split}
\end{align}
where $\Delta$ is the Euclidean Laplacian with respect to $\delta_{3}$
and $\Delta_{\rho,z}$ is the Euclidean Laplacian with respect to
$\delta_{2}=d\rho^{2}+dz^{2}$ on the orbit space, and
\begin{equation}
\det\nabla\lambda
=\delta_{3}(\nabla\lambda_{11},\nabla\lambda_{22})-|\nabla \lambda_{12}|^2.
\end{equation}
An expression for the mass \cite{Alaeeremarks2015} is obtained by integrating this formula by parts
\begin{align}\label{mass scalar}
\begin{split}
m=&\frac{1}{8}\int_{\mathbb{R}^{3}}\left(e^{2U+2\alpha-2\log r}R
+6|\nabla U|^{2}-\frac{\det\nabla\lambda}{2\rho^{2}}\right)dx\\
&+\frac{1}{32}\int_{\mathbb{R}^{3}}
e^{-2\alpha+2\log r}\lambda_{ij}(A_{\rho,z}^{i}-A_{z,\rho}^{i})
(A_{\rho,z}^{j}-A_{z,\rho}^{j})dx
+\frac{\pi}{2}\sum_{\varsigma=\pm}\int_{\Gamma_{\varsigma}}\alpha_{\varsigma}dz,
\end{split}
\end{align}
where the volume form
\begin{equation}
dx=\frac{1}{2}r^{5}\sin(2\theta)dr \wedge d\theta \wedge d\phi
=\rho d\rho \wedge dz \wedge d\phi
\end{equation}
arises from $\delta_{3}$. By combining \eqref{sc1} and \eqref{mass scalar} we obtain
\begin{align}\label{POI}
\begin{split}
m=&\mathcal{M}+\frac{1}{8}\int_{\R^3}e^{2U+2\alpha-2\log r}\left(16\pi\mu_{SG}
+\frac{1}{2}\lambda^{ij}E_i E_j +  \frac{1}{4}\lambda^{ij}\lambda^{ln}B_{il}B_{jn} \right)dx\\
&+\frac{1}{8}\int_{\R^3}e^{2U+2\alpha-2\log r}\left(\delta_{2}^{ij}\delta_{2}^{ln}k(e_i,e_l)k(e_{j},e_{n}) +\lambda^{ij}\lambda^{ln}k(e_{i+2},e_{l+2})k(e_{j+2},e_{n+2})\right)dx\\
&+\frac{1}{32}\int_{\R^3}
e^{-2\alpha+2\log r}\lambda_{ij}(A^i_{\rho,z}-A^i_{z,\rho})(A^j_{\rho,z}-A^j_{z,\rho})dx,
\end{split}
\end{align}
where the mass functional is given by
\begin{align}\label{massfunctional}
\begin{split}
	\mathcal{M}
	=&\frac{1}{8}\int_{\mathbb{R}^{3}}
	\left(6|\nabla U|^{2}-\frac{\det\nabla\lambda}{2\rho^{2}}+\frac{e^{-6U}}{2\rho^{2}}\Theta^{T}
	\lambda^{-1}\Theta\right)dx\\
	 &+\frac{1}{8}\int_{\mathbb{R}^{3}}\left(\frac{e^{-4U}}{2\rho^2}|\Upsilon|^2
+\frac{e^{-2U}}{2}\nabla \psi^{T}
	\lambda^{-1}\nabla \psi\right)dx
+\frac{\pi}{2}\sum_{\varsigma=\pm}\int_{\Gamma_{\varsigma}}\alpha_{\varsigma}dz.
\end{split}
\end{align}

It turns out that the mass function $\mathcal{M}$ may be expressed as a sum of squares, and related to a harmonic energy. In order to see this it is necessary to perform a change of variables $(\lambda_{11},\lambda_{22},\lambda_{12})\rightarrow (V,W)$ where
\begin{equation}\label{variablesvw}
	 V=\frac{1}{2}\log\left(\frac{\lambda_{11}\cos^{2}\theta}{\lambda_{22}\sin^{2}\theta}\right),
	\text{ }\text{ }\text{ }\text{ }\text{ }\text{ }W=\sinh^{-1}\left(\frac{\lambda_{12}}{\rho}\right),
\end{equation}
with inverse
\begin{equation}\label{inversevariables}
	\lambda_{11}=\left(\sqrt{\rho^2+z^2}-z\right)e^{V}\cosh W, \text{ }\text{ }\text{ }\text{ }\text{ }
	\lambda_{22}=\left(\sqrt{\rho^2+z^2}+z\right)e^{-V}\cosh W,\text{ }\text{ }\text{ }\text{ }\text{ }
	\lambda_{12}=\rho\sinh W.
\end{equation}
A computation then shows that
\begin{equation}\label{fff}
	-\frac{\det\nabla\lambda}{\rho^2}=|\nabla V|^2+|\nabla W|^2+\sinh^{2}W\abs{\nabla\left(V+h_2\right)}^2+2\delta_{3}(\nabla h_2,\nabla V),
\end{equation}
where
\begin{equation}\label{HARMONICFUN}
	h_{1}=\frac{1}{2}\log\rho,\text{ }\text{ }\text{ }\text{ }\text{ }\text{ }\text{ }\text{ }\text{ }h_{2}=\frac{1}{2}\log\left(\frac{\sqrt{\rho^{2}+z^{2}}-z}{
		\sqrt{\rho^2+z^{2}}+z}\right)=\log(\tan\theta),
\end{equation}
are harmonic functions on $(\mathbb{R}^{3}\setminus\Gamma,\delta_{3})$. Next observe that the last term of \eqref{fff} may be related to the boundary integral in \eqref{massfunctional} by
\begin{align}
	\begin{split}
		\frac{1}{8}\int_{\mathbb{R}^{3}}\delta_{3}(\nabla h_{2},\nabla V)dx
		=&-\lim_{\varepsilon\rightarrow 0}\frac{1}{8}\int_{\rho=\varepsilon}V\partial_{\rho}h_{2}\\
		 =&\frac{\pi}{4}\left(\int_{\Gamma_{-}}Vdz-\int_{\Gamma_{+}}Vdz\right)\\
		 =&-\frac{\pi}{2}\sum_{\varsigma=\pm}\int_{\Gamma_{\varsigma}}\alpha_{\varsigma}dz,
	\end{split}
\end{align}
where we have used
\begin{equation}\label{formulaV}
	V=2\alpha_{+}\text{ }\text{ }\text{ on }\text{ }\text{ }\Gamma_{+},\text{ }\text{ }\text{ }\text{ }\text{ }\text{ } V=-2\alpha_{-}\text{ }\text{ }\text{ on }\text{ }\text{ }\Gamma_{-},
\end{equation}
which follows from \eqref{ConeCondition} and \eqref{inversevariables}.
Putting this altogether yields the desired expression for the mass functional
\begin{align}\label{massfunc}
\begin{split}
16\mathcal{M}=&\int_{\mathbb{R}^{3}}12|\nabla U|^{2}+|\nabla V|^{2}+|\nabla W|^{2}
+\sinh^{2}W|\nabla (V+h_{2})|^{2}
+\frac{e^{-6h_{1}-6U-h_{2}-V}}{\cosh W}|\Theta^{1}|^{2}dx\\
&+\int_{\mathbb{R}^{3}}
e^{-6h_{1}-6U+h_{2}+V}\cosh W
\left|e^{-h_{2}-V}\tanh W\Theta^{1}-\Theta^{2}\right|^{2}
+\frac{e^{-2h_{1}-2U-h_{2}-V}}{\cosh W}|\nabla\psi^{1}|^{2}
dx\\
&+\int_{\mathbb{R}^{3}}e^{-2h_{1}-2U+h_{2}+V}\cosh W|e^{-h_{2}-V}\tanh W \nabla\psi^{1}-\nabla\psi^{2}|^{2}+
e^{-4h_{1}-4U}
|\Upsilon|^{2}dx.
\end{split}
\end{align}

\section{The Dirichlet Energy and Global Minimization}
\label{dirichlet}

Dimensional reduction of 5-dimensional minimal supergravity from five to three dimensions results in three dimensional gravity coupled to a nonlinear sigma model \cite{Compere,Mizoguchi}, which is invariant under the the exceptional Lie group $G_{2(2)}$. This is the noncompact real form of $G_{2}$, where the first 2 indicates the rank and the 2 inside parentheses indicates the character. The target space for the nonlinear sigma model is $G_{2(2)}/SO(4)\cong\mathbb{R}^{8}$ \cite{Yokota}, and from the dimensional reduction it comes equipped with a complete Riemannian metric of nonpositive curvature given by
\begin{equation}\label{1}
ds^{2}=\frac{1}{2}\text{Tr}\left[\left(\Phi^{-1}d\Phi\right)^2\right],
\end{equation}
where $\Phi$ is a $7\times 7$ matrix coset representative defined in (3.4) of \cite{Tomizawa:2009ua} (and denoted by $M$ in this reference). A calculation shows that
\begin{equation}\label{2}
\text{Tr}\left[\left(\Phi^{-1}d\Phi\right)^{2}\right]=
\text{Tr}\left[\left(\Lambda^{-1}d\Lambda\right)^2\right]+\left(d\log\det\Lambda\right)^2\nonumber\\
	+2d\psi^T\Lambda^{-1}d\psi
+\frac{2}{\det\Lambda}\Theta^T\Lambda^{-1}\Theta
+\frac{2}{\det\Lambda}\Upsilon^2
\end{equation}
where $\Lambda_{ij}=e^{2U}\lambda_{ij}$. By parameterizing the target space with respect to the variables $u$, $v$, $w$, $\zeta^1$, $\zeta^2$, $\chi$, $\psi^1$, and $\psi^2$ where $u=U+h_1$, $v=V+h_2$, and $w=W$ we obtain
\begin{align}\label{3}
\begin{split}
ds^2=&12du^{2}
		+\cosh^{2}w dv^{2}
		+dw^{2}+\frac{e^{-6u-v}}{\cosh w}(\Theta^{1})^{2}
+e^{-6u+v}\cosh w (e^{-v}\tanh w\Theta^{1}-\Theta^{2})^{2}\\
&+\frac{e^{-2u-v}}{\cosh w}(d\psi^{1})^{2}
+e^{-2u+v}\cosh w(e^{-v}\tanh w d\psi^{1}-d\psi^{2})^{2}
+e^{-4u}\Upsilon^{2}.
\end{split}
\end{align}
It follows that the harmonic energy, in a domain $\Omega\subset\mathbb{R}^{3}$, of a map $\tilde{\Psi}=(u,v,w,\zeta^{1},\zeta^{2},\chi,\psi^{1},\psi^{2}):\mathbb{R}^{3}\rightarrow
G_{2(2)}/SO(4)$ takes the form
\begin{align}\label{4}
\begin{split}
E_{\Omega}(\tilde{\Psi})=&\int_{\Omega}12|\nabla u|^{2}
		+\cosh^{2}w|\nabla v|^{2}
		+|\nabla w|^{2}+\frac{e^{-6u-v}}{\cosh w}|\Theta^{1}|^{2}
+e^{-6u+v}\cosh w |e^{-v}\tanh w\Theta^{1}-\Theta^{2}|^{2}dx\\
&+\int_{\Omega}\frac{e^{-2u-v}}{\cosh w}|\nabla\psi^{1}|^{2}
+e^{-2u+v}\cosh w|e^{-v}\tanh w \nabla\psi^{1}-\nabla\psi^{2}|^{2}
+e^{-4u}\left|\Upsilon\right|^{2}dx.
\end{split}
\end{align}
This harmonic energy is related to the mass functional \eqref{massfunc} through
an integration by parts. In particular, on a domain
$\Omega$ which does not intersect the axes $\Gamma$ we have
\begin{equation}\label{5}
\mathcal{I}_{\Omega}(\Psi)=E_{\Omega}(\tilde{\Psi})
-\int_{\partial\Omega}12(h_{1}+2U)\partial_{\nu}h_{1}
-\int_{\partial\Omega}(h_{2}+2V)\partial_{\nu}h_{2},
\end{equation}
where $\mathcal{I}=\mathcal{I}_{\mathbb{R}^{3}}=16\mathcal{M}$, $\Psi=(U,V,W,\zeta^{1},\zeta^{2},\chi,\psi^{1},\psi^{2})$, and
$\nu$ denotes the unit outer normal to the boundary $\partial\Omega$. Note that since $E$ and $\mathcal{I}$ agree up to boundary integrals, they must have the same critical points.
The functional $\mathcal{I}$ is referred to as the reduced energy, since it is a regularization in the sense that the two infinite terms $\int|\nabla h_{1}|^{2}$ and $\int \cosh^{2}W|\nabla h_{2}|^{2}$, which appear in $E$, have been eliminated.

The purpose of the remainder of the paper is to establish a lower bound for the mass functional, and to compute its value. A critical point and natural candidate minimizer is the map $\Psi_{0}=(U_{0},V_{0},W_{0},\zeta^{1}_{0},\zeta^{2}_{0},\chi_{0},
\psi_{0}^{1},\psi_{0}^{2})$, the renormalization of the extreme charged Myers-Perry harmonic map
$\tilde{\Psi}_{0}=(u_{0},v_{0},w_{0},\zeta^{1}_{0},\zeta^{2}_{0},
\chi_{0},\psi_{0}^{1},\psi_{0}^{2})$ described in Appendix \ref{app2}. In order to establish this map as the global minimizer, we will employ the basic observation that
the target space $G_{2(2)}/SO(4)$ is nonpositively curved, and hence the harmonic energy is convex along geodesic deformations. The use of energy convexity to minimize mass related functionals was introduced in \cite{schoen2013convexity} and applied in \cite{alaee2015global}, \cite{khuri2015positive}. The difficulty here is that the harmonic map $\tilde{\Psi}_{0}$ is singular along the axes, and so it is not clear that convexity of the harmonic energy is inherited by the reduced energy. This difficulty is typically overcome by adopting a cut-and-paste procedure, whereby portions of a given map near the axes and the designated asymptotically flat end, are replaced by corresponding parts of the candidate minimizer with a Lipschitz transition between different regions. In more detail, define $\Omega_{\delta,\varepsilon}=\{\delta< r<2/\delta;
\rho>\varepsilon\}$ and $\mathcal{A}_{\delta,\varepsilon}=B_{2/\delta}\setminus
\Omega_{\delta,\varepsilon}$, where $\delta,\varepsilon>0$ are
small parameters and $B_{2/\delta}$ is the ball of radius $2/\delta$ centered at the origin.
Given a map $\Psi$, let $\Psi_{\delta,\varepsilon}$ be the resulting map obtained from the cut-and-paste procedure so that its components satisfy
\begin{align}\label{54}
\begin{split}
&\text{ }\text{ }\text{ }\text{ }\text{ }\text{ }\text{ }\text{ }\text{ }\text{ }\text{ }\text{ }\text{ }\text{ }\text{ }\text{ }\text{ }\text{ }\text{ }\text{ }\text{ }\text{ }\text{ }\text{ }\text{ }\text{ }\text{ }\text{ }\text{ }\text{ }\text{ }\text{ }\text{ }\text{ }\text{ }\text{ }\text{ }\text{ }\text{ }\operatorname{supp}(U_{\delta,\varepsilon}-U_{0})\subset B_{2/\delta},\\
&\operatorname{supp}(V_{\delta,\varepsilon}-V_{0},W_{\delta,\varepsilon}-W_{0},
\zeta_{\delta,\varepsilon}^{1}-\zeta^{1}_{0},\zeta_{\delta,\varepsilon}^{2}-\zeta^{2}_{0},
\chi_{\delta,\varepsilon}-\chi_{0},
\psi_{\delta,\varepsilon}^{1}-\psi^{1}_{0},\psi_{\delta,\varepsilon}^{2}-\psi^{2}_{0})\subset \Omega_{\delta,\varepsilon}.
\end{split}
\end{align}
Consider now a geodesic $\tilde{\Psi}^{t}_{\delta,\varepsilon}$ in $G_{2(2)}/SO(4)$, $t\in[0,1]$, connecting
$\tilde{\Psi}^{1}_{\delta,\varepsilon}=\tilde{\Psi}_{\delta,\varepsilon}$ to $\tilde{\Psi}^{0}_{\delta,\varepsilon}=\tilde{\Psi}_{0}$. Then $\Psi^{t}_{\delta,\varepsilon}\equiv\Psi_{0}$ outside $B_{2/\delta}$ and
\begin{equation}
(V^{t}_{\delta,\varepsilon},W^{t}_{\delta,\varepsilon},\zeta^{1,t}_{\delta,\varepsilon},
\zeta^{2,t}_{\delta,\varepsilon},\chi^{t}_{\delta,\varepsilon},\psi^{1,t}_{\delta,\varepsilon},
\psi^{2,t}_{\delta,\varepsilon})
\equiv(V_{0},W_{0},\zeta^{1}_{0},\zeta^{2}_{0},\chi_{0},\psi^{1}_{0},\psi^{2}_{0})\text{ }\text{ }\text{ }\text{ on }\text{ }\text{ }\text{ }\mathcal{A}_{\delta,\varepsilon},
\end{equation}
so that in particular $U^{t}_{\delta,\varepsilon}=U_{0}+t(U_{\delta,\varepsilon}-U_{0})$ and $V^{t}_{\delta,\varepsilon}=V_{0}$ on these
regions. It is the simple linear dynamics of $U^{t}_{\delta,\varepsilon}$ and constancy of
$V^{t}_{\delta,\varepsilon}$ (in $t$) which guarantee that the boundary terms of \eqref{5} do not obstruct the induced convexity of the renormalized harmonic energy, so that
\begin{equation}\label{55}
\frac{d^{2}}{dt^{2}}\mathcal{I}(\Psi^{t}_{\delta,\varepsilon})
\geq 2\int_{\mathbb{R}^{3}}
|\nabla\operatorname{dist}_{G_{2(2)}/SO(4)}(\Psi_{\delta,\varepsilon},\Psi_{0})|^{2}dx.
\end{equation}
Moreover, as $\Psi_{0}$ is a critical point we have
\begin{equation}\label{56}
\frac{d}{dt}\mathcal{I}(\Psi^{t}_{\delta,\varepsilon})|_{t=0}=0.
\end{equation}
Hence, by integrating \eqref{55} a gap lower bound \eqref{53} is achieved after applying the Gagliardo-Nirenberg-Sobolev inequality and letting $\delta,\varepsilon\rightarrow 0$. The next sections will be dedicated to verifying each of the steps above to prove the following result.

\begin{theorem}\label{minimum}
Suppose that $\Psi=(U,V,W,\zeta^{1},\zeta^{2},\chi,\psi^1,\psi^2)$ is smooth and satisfies the asymptotics \eqref{FALLOFF1}-\eqref{FALLOFF4.1'}
with $\zeta^{i}|_{\Gamma}=\zeta^{i}_{0}|_{\Gamma}$, $i=1,2$, and $\chi|_{\Gamma}=\chi_{0}|_{\Gamma}$, then there exists a constant $C>0$ such that
\begin{equation}\label{53}
\mathcal{I}(\Psi)-\mathcal{I}(\Psi_{0})
\geq C\left(\int_{\mathbb{R}^{3}}
\operatorname{dist}_{G_{2(2)}/SO(4)}^{6}(\Psi,\Psi_{0})dx
\right)^{\frac{1}{3}}.
\end{equation}
\end{theorem}

\section{The Cut-and-Paste Argument}
\label{seccut}

As has already been described, we intend to replace a given map $\Psi$ with a new map $\Psi_{\delta,\varepsilon}$ which essentially agrees with the renormalized extreme Myers-Perry harmonic map $\Psi_{0}$ on certain asymptotic regimes. In this section we describe this construction in detail, and show that the new maps may be used to approximate the original in the context of the reduced energy.
In order to carry this out, $\Psi$ must satisfy the appropriate asymptotics which will be recorded below; the asymptotics for $\Psi_{0}$ will also be stated. In the expressions for the asymptotics,
it may appear that certain derivatives have extra fall-off than is expected. This is due to the fact that the vector norms employed are taken with respect to the flat metric $\delta_{3}$ in the cylindrical and nonstandard polar coordinates \eqref{pcoordinates}.

In what follows, $\kappa>0$ is a fixed small parameter. Let us consider the asymptotically flat end first. We require that as $r\rightarrow\infty$ the following decay occurs
\begin{equation}\label{FALLOFF1}
U,V=O(r^{-1-\kappa}),\text{ }\text{ }\text{ }W=\sqrt{\rho} O(r^{-2-\kappa}),\text{ }\text{ }\text{ }\psi^{1}=\sqrt{\sin\theta}O(r^{-\kappa}),\text{ }\text{ }\text{ }\psi^{2}=\sqrt{\cos\theta}O(r^{-\kappa}),
\end{equation}
\begin{equation}\label{FALLOFF2}
|\nabla U|=O(r^{-3-\kappa}),\text{ }\text{ }\text{ }\text{ }|\nabla V|=O(r^{-3-\kappa}),\text{ }\text{ }\text{ }\text{ }|\nabla W|=\rho^{-\frac{1}{2}} O(r^{-2-\kappa}),
\end{equation}
\begin{equation}\label{FALLOFF3}
|\nabla\psi^{1}|=\sqrt{\sin\theta}O(r^{-2-\kappa}),\text{ }\text{ }\text{ }\text{ }|\nabla\psi^{2}|=\sqrt{\cos\theta}O(r^{-2-\kappa}),\text{ }\text{ }\text{ }
\end{equation}
\begin{equation}\label{FALLOFF3'}
|\nabla\chi|=\rho O(r^{-3-\kappa}),\text{ }\text{ }\text{ }\text{ }\text{ }\text{ }\text{ }|\nabla \zeta^{1}|=\rho\sqrt{\sin\theta} O(r^{-2-\kappa}),\text{ }\text{ }\text{ }\text{ }\text{ }\text{ }\text{ }|\nabla \zeta^{2}|=\rho\sqrt{\cos\theta} O(r^{-2-\kappa}).
\end{equation}
Consider now the nondesignated end, in which the asymptotics are broken up into two cases. In the asymptotically flat case, as $r\rightarrow 0$, we require
\begin{equation}\label{FALLOFF1.0}
(U+2\log r),V=O(1),\text{ }\text{ }\text{ } W=\sqrt{\rho}O(r^{-1}),\text{ }\text{ }\text{ }\psi^{1}=\sqrt{\sin\theta}O(r^{-1}),\text{ }\text{ }\text{ }
\psi^{2}=\sqrt{\cos\theta}O(r^{-1}),
\end{equation}
\begin{equation}\label{FALLOFF2.0}
|\nabla U|=O(r^{-2}),\text{ }\text{ }\text{ }\text{ }\text{ }|\nabla V|=O(r^{-2}),\text{ }\text{ }\text{ }\text{ }\text{ }|\nabla W|=\rho^{-\frac{1}{2}} O(r^{-1}),
\end{equation}
\begin{equation}\label{FALLOFF4.0}
|\nabla\psi^{1}|=\sqrt{\sin\theta}O(r^{-3}),\text{ }\text{ }\text{ }\text{ }|\nabla\psi^{2}|=\sqrt{\cos\theta}O(r^{-3}),\text{ }\text{ }\text{ }\text{ }
\end{equation}
\begin{equation}\label{FALLOFF4.0'}
|\nabla\chi|=\rho O(r^{-7+\kappa}),\text{ }\text{ }\text{ }\text{ }\text{ }\text{ }\text{ }|\nabla \zeta^{1}|=\rho\sqrt{\sin\theta} O(r^{-8+\kappa}),\text{ }\text{ }\text{ }\text{ }\text{ }\text{ }\text{ }|\nabla \zeta^{2}|=\rho\sqrt{\cos\theta} O(r^{-8+\kappa}).
\end{equation}
In the asymptotically cylindrical case, as $r\rightarrow 0$, we require
\begin{equation}\label{FALLOFF5.0}
(U+\log r),V=O(1),\text{ }\text{ }\text{ } W=\sqrt{\rho}O(r^{-1}),\text{ }\text{ }\text{ }\psi^{1}=\sqrt{\sin\theta}O(1),\text{ }\text{ }\text{ }
\psi^{2}=\sqrt{\cos\theta}O(1),
\end{equation}
\begin{equation}\label{FALLOFF6.0}
|\nabla U|=O(r^{-2}),\text{ }\text{ }\text{ }\text{ }\text{ }|\nabla V|=O(r^{-2}),\text{ }\text{ }\text{ }\text{ }\text{ }|\nabla W|=\rho^{-\frac{1}{2}} O(r^{-1}),
\end{equation}
\begin{equation}\label{FALLOFF7.0}
|\nabla\psi^{1}|=\sqrt{\sin\theta}O(r^{-2}),\text{ }\text{ }\text{ }\text{ }|\nabla\psi^{2}|=\sqrt{\cos\theta}O(r^{-2}),\text{ }\text{ }\text{ }\text{ }
\end{equation}
\begin{equation}\label{FALLOFF7.0'}
|\nabla\chi|=\rho O(r^{-5+\kappa}),\text{ }\text{ }\text{ }\text{ }\text{ }\text{ }\text{ }
|\nabla \zeta^{1}|=\rho\sqrt{\sin\theta} O(r^{-5+\kappa}),\text{ }\text{ }\text{ }\text{ }\text{ }\text{ }\text{ }|\nabla \zeta^{2}|=\rho\sqrt{\cos\theta} O(r^{-5+\kappa}).
\end{equation}
Moreover the asymptotics near the axis, that is, as $\rho\rightarrow 0$ with $\delta\leq r\leq \frac{2}{\delta}$, are required to satisfy
\begin{equation}\label{FALLOFF1.1}
U,V=O(1),\text{ }\text{ }\text{ } W=O(\sqrt{\rho}),\text{ }\text{ }\text{ }\psi^{1}=\sqrt{\sin\theta}O(1),\text{ }\text{ }\text{ }
\psi^{2}=\sqrt{\cos\theta}O(1),
\end{equation}
\begin{equation}\label{FALLOFF2.1}
|\nabla U|=O(1),\text{ }\text{ }\text{ }\text{ }\text{ }|\nabla V|=O(1),\text{ }\text{ }\text{ }\text{ }\text{ }|\nabla W|=O(\rho^{-\frac{1}{2}}),
\end{equation}
\begin{equation}\label{FALLOFF4.1}
|\nabla\psi^{1}|=\sqrt{\sin\theta}O(1),\text{ }\text{ }\text{ }\text{ }|\nabla\psi^{2}|=\sqrt{\cos\theta}O(1),\text{ }\text{ }\text{ }\text{ }
\end{equation}
\begin{equation}\label{FALLOFF4.1'}
|\nabla\chi|=O(\rho),\text{ }\text{ }\text{ }\text{ }\text{ }\text{ }\text{ }|\nabla \zeta^{1}|=\sqrt{\sin\theta}O(\rho),\text{ }\text{ }\text{ }\text{ }\text{ }\text{ }\text{ }|\nabla \zeta^{2}|=\sqrt{\cos\theta}O(\rho).
\end{equation}
It should be observed that these asymptotics guarantee a finite reduced energy, and are satisfied by the extreme and non-extreme charged Myers-Perry harmonic maps.

In Appendix \ref{app2} the extreme charged Myers-Perry map $\Psi_{0}$ is described in detail, and from this the asymptotics may be derived. In the designated asymptotically flat end as $r\rightarrow\infty$ we find
\begin{equation}\label{FALLOFF4}
U_{0},V_{0}=O(r^{-2}),\text{ }\text{ }\text{ }W_{0}=\rho O(r^{-6}),\text{ }\text{ }\text{ }\psi_{0}^{1}=\sin^{2}\theta O(r^{-2}),\text{ }\text{ }\text{ }\psi_{0}^{2}=\cos^{2}\theta O(r^{-2}),
\end{equation}
\begin{equation}\label{FALLOFF5}
|\nabla U_{0}|=O(r^{-4}),\text{ }\text{ }\text{ }\text{ }|\nabla V_{0}|=O(r^{-4}),\text{ }\text{ }\text{ }\text{ }|\nabla W_{0}|=O(r^{-6}),
\end{equation}
\begin{equation}\label{FALLOFF6}
|\nabla\psi_{0}^{1}|=\sin\theta O(r^{-4}),\text{ }\text{ }\text{ }\text{ }|\nabla\psi_{0}^{2}|=\cos\theta O(r^{-4}),\text{ }\text{ }\text{ }
\end{equation}
\begin{equation}\label{FALLOFF6'}
|\nabla\chi_{0}|=\rho O(r^{-4}),\text{ }\text{ }\text{ }\text{ }\text{ }\text{ }\text{ }|\nabla \zeta_{0}^{1}|=\rho\sin^{2}\theta O(r^{-4}),\text{ }\text{ }\text{ }\text{ }\text{ }\text{ }\text{ }|\nabla \zeta_{0}^{2}|=\rho\cos^{2}\theta O(r^{-4}).
\end{equation}
In the nondesignated end, as $r\rightarrow 0$, the geometry is asymptotically cylindrical and the asymptotics are given by
\begin{equation}\label{FALLOFF8.0}
(U_{0}+\log r),V_{0}=O(1),\text{ }\text{ }\text{ } W_{0}=\rho O(r^{-2}),\text{ }\text{ }\text{ }\psi_{0}^{1}=\sin^{2}\theta O(1),\text{ }\text{ }\text{ }
\psi_{0}^{2}=\cos^{2}\theta O(1),
\end{equation}
\begin{equation}\label{FALLOFF9.0}
|\nabla U_{0}|=O(r^{-2}),\text{ }\text{ }\text{ }\text{ }\text{ }|\nabla V_{0}|=O(r^{-2}),\text{ }\text{ }\text{ }\text{ }\text{ }|\nabla W_{0}|=O(r^{-2}),
\end{equation}
\begin{equation}\label{FALLOFF10.0}
|\nabla\psi_{0}^{1}|=\sin\theta O(r^{-2}),\text{ }\text{ }\text{ }\text{ }|\nabla\psi_{0}^{2}|=\cos\theta O(r^{-2}),\text{ }\text{ }\text{ }\text{ }
\end{equation}
\begin{equation}\label{FALLOFF10.0'}
|\nabla\chi_{0}|=\rho O(r^{-4}),\text{ }\text{ }\text{ }\text{ }\text{ }\text{ }
|\nabla \zeta_{0}^{1}|=\rho\sin^{2}\theta O(r^{-4}),\text{ }\text{ }\text{ }\text{ }\text{ }\text{ }\text{ }\text{ }|\nabla \zeta_{0}^{2}|=\rho\cos^{2}\theta O(r^{-4}).
\end{equation}
Furthermore when $\rho\rightarrow 0$ with $\delta\leq r\leq\frac{2}{\delta}$ we have
\begin{equation}\label{FALLOFF8.1}
U_{0},V_{0}=O(1),\text{ }\text{ }\text{ } W_{0}=O(\rho),\text{ }\text{ }\text{ }\psi_{0}^{1}=\sin^{2}\theta O(1),\text{ }\text{ }\text{ }
\psi_{0}^{2}=\cos^{2}\theta O(1),
\end{equation}
\begin{equation}\label{FALLOFF9.1}
|\nabla U_{0}|=O(1),\text{ }\text{ }\text{ }\text{ }\text{ }|\nabla V_{0}|=O(1),\text{ }\text{ }\text{ }\text{ }\text{ }|\nabla W_{0}|=O(1),
\end{equation}
\begin{equation}\label{FALLOFF10.1}
|\nabla\psi_{0}^{1}|=\sin\theta O(1),\text{ }\text{ }\text{ }\text{ }|\nabla\psi_{0}^{2}|=\cos\theta O(1),\text{ }\text{ }\text{ }\text{ }
\end{equation}
\begin{equation}\label{FALLOFF10.1'}
|\nabla\chi_{0}|=O(\rho),\text{ }\text{ }\text{ }\text{ }\text{ }\text{ }|\nabla \zeta_{0}^{1}|=\sin^{2}\theta O(\rho),\text{ }\text{ }\text{ }\text{ }\text{ }\text{ }\text{ }\text{ }|\nabla \zeta_{0}^{2}|=\cos^{2}\theta O(\rho).
\end{equation}

Construction of the approximating maps $\Psi_{\delta,\varepsilon}$, satisfying \eqref{54}, is accomplished with a three step cut-and-paste procedure inspired from \cite{schoen2013convexity},
and utilizes the following cut-off functions tailored to the three asymptotic regimes
\begin{equation}
\overline{\varphi}_{\delta}=\begin{cases}
1 & \text{ if $r\leq\frac{1}{\delta}$,} \\
|\nabla\overline{\varphi}_{\delta}|\leq 2\delta^2 &
\text{ if $\frac{1}{\delta}<r<\frac{2}{\delta}$,} \\
0 & \text{ if $r\geq\frac{2}{\delta}$,} \\
\end{cases}
\end{equation}
\begin{equation}
\varphi_{\delta}=\begin{cases}
0 & \text{ if $r\leq\delta$,} \\
|\nabla\varphi_{\delta}|\leq \frac{2}{\delta^{2}} &
\text{ if $\delta<r<2\delta$,} \\
1 & \text{ if $r\geq2\delta$,} \\
\end{cases}
\end{equation}
and
\begin{equation}
\phi_{\varepsilon}=\begin{cases}
0 & \text{ if $\rho\leq\varepsilon$,} \\
\frac{\log(\rho/\varepsilon)}{\log(\sqrt{\varepsilon}/
\varepsilon)} &
\text{ if $\varepsilon<\rho<\sqrt{\varepsilon}$,} \\
1 & \text{ if $\rho\geq\sqrt{\varepsilon}$.} \\
\end{cases}
\end{equation}
These functions should be Lipschitz and take values in the interval $[0,1]$.

\begin{lemma}
Set
\begin{equation}
\overline{F}_{\delta}(\Psi)=\Psi_{0}
+\overline{\varphi}_{\delta}(\Psi-\Psi_{0})=:
(\overline{U}_{\delta},\overline{V}_{\delta},\overline{W}_{\delta},
\overline{\zeta}_{\delta}^{1},\overline{\zeta}_{\delta}^{2},
\overline{\chi}_{\delta},\overline{\psi}^{1}_{\delta},\overline{\psi}^{2}_{\delta}),
\end{equation}
so that $\overline{F}_{\delta}(\Psi)=\Psi_{0}$ on $\mathbb{R}^{3}\setminus B_{2/\delta}$. Then
$\lim_{\delta\rightarrow 0}\mathcal{I}(\overline{F}_{\delta}(\Psi))=\mathcal{I}(\Psi).$
\end{lemma}

\begin{proof}
Observe that
\begin{equation}
\mathcal{I}(\overline{F}_{\delta}(\Psi))
=\mathcal{I}_{r\leq\frac{1}{\delta}}(\overline{F}_{\delta}(\Psi))
+\mathcal{I}_{\frac{1}{\delta}< r<\frac{2}{\delta}}(\overline{F}_{\delta}(\Psi))
+\mathcal{I}_{r\geq\frac{2}{\delta}}(\overline{F}_{\delta}(\Psi)),
\end{equation}
and by the dominated convergence theorem (DCT) $\mathcal{I}_{r\leq\frac{1}{\delta}}(\overline{F}_{\delta}(\Psi))
\rightarrow \mathcal{I}(\Psi)$. Furthermore $\Psi_{0}$ has finite reduced
energy, which implies that $\mathcal{I}_{r\geq\frac{2}{\delta}}(\overline{F}_{\delta}(\Psi))\rightarrow 0$. Now write
\begin{align}
\begin{split}
\mathcal{I}_{\frac{1}{\delta}< r<\frac{2}{\delta}}(\overline{F}_{\delta}(\Psi))
=&\underbrace{\int_{\frac{1}{\delta}<r<\frac{2}{\delta}}
12|\nabla \overline{U}_{\delta}|^{2}}_{I_{1}}
+\underbrace{\int_{\frac{1}{\delta}<r<\frac{2}{\delta}}
|\nabla\overline{V}_{\delta}|^{2}}_{I_{2}}
+\underbrace{\int_{\frac{1}{\delta}<r<\frac{2}{\delta}}
|\nabla\overline{W}_{\delta}|^{2}}_{I_{3}}\\
&+\underbrace{\int_{\frac{1}{\delta}<r<\frac{2}{\delta}}
\sinh^{2}\overline{W}_{\delta}|\nabla(\overline{V}_{\delta}+h_{2})|^{2}}_{I_{4}}
+\underbrace{\int_{\frac{1}{\delta}<r<\frac{2}{\delta}}
\frac{\cos\theta}{\rho^{3}\sin\theta}
\frac{e^{-\overline{V}_{\delta}-6\overline{U}_{\delta}}}{\cosh\overline{W}_{\delta}}
|\overline{\Theta}^{1}_{\delta}|^{2}}_{I_{5}}\\
&+\underbrace{\int_{\frac{1}{\delta}<r<\frac{2}{\delta}}
\frac{\sin\theta}{\rho^{3}\cos\theta}
e^{\overline{V}_{\delta}-6\overline{U}_{\delta}}\cosh\overline{W}_{\delta}
|\overline{\Theta}^{2}_{\delta}-e^{-\overline{V}_{\delta}}\cot\theta \tanh\overline{W}_{\delta}
\overline{\Theta}^{1}_{\delta}|^{2}}_{I_{6}}\\
&+\underbrace{\int_{\frac{1}{\delta}<r<\frac{2}{\delta}}
\frac{\sin\theta}{\rho\cos\theta}
e^{\overline{V}_{\delta}-2\overline{U}_{\delta}}\cosh\overline{W}_{\delta}
|\nabla\overline{\psi}^{2}_{\delta}-e^{-\overline{V}_{\delta}}\cot\theta \tanh\overline{W}_{\delta}
\nabla\overline{\psi}^{1}_{\delta}|^{2}}_{I_{7}}\\
&+\underbrace{\int_{\frac{1}{\delta}<r<\frac{2}{\delta}}
\frac{\cos\theta}{\rho\sin\theta}
\frac{e^{-\overline{V}_{\delta}-2\overline{U}_{\delta}}}{\cosh\overline{W}_{\delta}}
|\nabla\overline{\psi}^{1}_{\delta}|^{2}}_{I_{8}}
+\underbrace{\int_{\frac{1}{\delta}<r<\frac{2}{\delta}}
\rho^{-2}e^{-4\overline{U}_{\delta}}|\overline{\Upsilon}_{\delta}|^{2}
}_{I_{9}}.
\end{split}
\end{align}
A direct computation shows that
\begin{equation}
I_{1}\leq C\int_{0}^{\frac{\pi}{2}}\int_{\frac{1}{\delta}}^{\frac{2}{\delta}}
\left(\underbrace{|\nabla U|^{2}}_{O(r^{-6-2\kappa})}+\underbrace{|\nabla U_{0}|^{2}}_{O(r^{-8})}
+\underbrace{(U-U_{0})^{2}}_{O(r^{-2-2\kappa})}
\underbrace{|\nabla\overline{\varphi}_{\delta}|^{2}}_{ O(\delta^{4})}\right)r^{5}\sin(2\theta)dr d\theta\rightarrow 0,
\end{equation}
and similar considerations yield $I_{2}\rightarrow 0$ as well as $I_{3}\rightarrow 0$.
Moreover, using
$\sinh\overline{W}_{\delta}=\sqrt{\rho}O(r^{-2-\kappa})$ produces
\begin{equation}
I_{4}\leq \int_{0}^{\frac{\pi}{2}}\int_{\frac{1}{\delta}}^{\frac{2}{\delta}}
\rho O(r^{-4-2\kappa})\left(\underbrace{|\nabla V|^{2}}_{O(r^{-6-2\kappa})}+
\underbrace{|\nabla V_{0}|^{2}}_{O(r^{-8})}
+\underbrace{|\nabla h_{2}|^{2}}_{O(\rho^{-2})}
+\underbrace{(V-V_{0})^{2}}_{O(r^{-2-2\kappa})}\underbrace{|\nabla\overline{\varphi}_{\delta}|^{2}}_{ O(\delta^{4})}\right)r^{5}\sin(2\theta)dr d\theta\rightarrow 0.
\end{equation}

Next observe that \eqref{FALLOFF3}, \eqref{FALLOFF3'}, and \eqref{FALLOFF6'} together with $(\chi-\chi_{0})|_{\Gamma}=0$, $\psi^{i}|_{\Gamma_{i}}=0$, and $(\zeta^{i}-\zeta_{0}^{i})|_{\Gamma}=0$, $i=1,2$ give rise to the following estimates for $r\in[\frac{1}{\delta},\frac{2}{\delta}]$:
\begin{equation}\label{002.1}
|(\chi-\chi_{0})(\rho,z,\phi)|
\leq\int_{0}^{\rho}|\partial_{\rho}(\chi-\chi_{0})
(\tilde{\rho},z,\phi)|d\tilde{\rho}
=\rho^{2}O(r^{-3-\kappa}),
\end{equation}
\begin{equation}\label{002.2}
|\psi^{1}(\rho,z,\phi)|
\leq\int_{0}^{\rho}|\partial_{\rho}\psi^{1}
(\tilde{\rho},z,\phi)|d\tilde{\rho}
=\sin^{3/2}\theta O(r^{-\kappa}),
\end{equation}
\begin{equation}\label{002.3}
|\psi^{2}(\rho,z,\phi)|
\leq\int_{0}^{\rho}|\partial_{\rho}\psi^{2}
(\tilde{\rho},z,\phi)|d\tilde{\rho}
=\cos^{3/2}\theta O(r^{-\kappa}),
\end{equation}
\begin{equation}\label{002}
|(\zeta^{i}-\zeta^{i}_{0})(\rho,z,\phi)|
\leq\int_{0}^{\rho}|\partial_{\rho}(\zeta^{i}-\zeta^{i}_{0})
(\tilde{\rho},z,\phi)|d\tilde{\rho}
=\rho^{2}O(r^{-2-\kappa}).
\end{equation}
From this we find that
\begin{equation}
|\nabla\overline{\psi}^{1}_{\delta}|
\leq\underbrace{|\nabla\psi^{1}|}_{\sqrt{\sin\theta}O(r^{-2-\kappa})}
+\underbrace{|\nabla\psi_{0}^{1}|}_{\sin\theta O(r^{-4})}
+\underbrace{|\psi^{1}-\psi^{1}_{0}|}_{\sin^{3/2}\theta O(r^{-\kappa})}\underbrace{|\nabla\overline{\varphi}_{\delta}|}_{O(\delta^{2})}
=\sqrt{\sin\theta}O(r^{-2-\kappa}),
\end{equation}
and similarly
\begin{equation}
|\nabla\overline{\psi}^{2}_{\delta}|
=\sqrt{\cos\theta}O(r^{-2-\kappa}),
\end{equation}
as well as
\begin{align}
\begin{split}
|\overline{\Upsilon}_{\delta}|\leq& C\left(\underbrace{|\nabla\chi|}_{\rho O(r^{-3-\kappa})}+\underbrace{|\nabla\chi_{0}|}_{\rho O(r^{-4})}
+\underbrace{|\chi-\chi_{0}|}_{\rho^{2} O(r^{-3-\kappa})}\underbrace{|\nabla
\overline{\varphi}_{\delta}|}_{O(\delta^{2})}
+\underbrace{(|\psi^{1}|+|\psi^{1}_{0}|)}_{\sin^{3/2}\theta O(r^{-\kappa})}
\underbrace{(|\nabla\psi^{2}|
+|\nabla\psi^{2}_{0}|)}_{\sqrt{\cos\theta}O(r^{-2-\kappa})}\right.\\
&
+\left.\underbrace{(|\psi^{2}|+|\psi^{2}_{0}|)}_{\cos^{3/2}\theta O(r^{-\kappa})}
\underbrace{(|\nabla\psi^{1}|
+|\nabla\psi^{1}_{0}|)}_{\sqrt{\sin\theta}O(r^{-2-\kappa})}
+\underbrace{(|\psi^{1}|+|\psi^{1}_{0}|)}_{\sin^{3/2}\theta O(r^{-\kappa})}
\underbrace{(|\psi^{2}|+|\psi^{2}_{0}|)}_{\cos^{3/2}\theta O(r^{-\kappa})}
\underbrace{|\nabla\overline{\varphi}_{\delta}|}_{O(\delta^{2})}\right)\\
=& \sqrt{\rho} O(r^{-2-\kappa}),
\end{split}
\end{align}
\begin{align}
\begin{split}
|\overline{\Theta}_{\delta}^{1}|\leq& C\left[\underbrace{|\nabla\zeta^{1}|}_{\rho \sqrt{\sin\theta}O(r^{-2-\kappa})}+\underbrace{|\nabla\zeta^{1}_{0}|}_{\rho \sin^{2}\theta O(r^{-4})}
+\underbrace{|\zeta^{1}-\zeta^{1}_{0}|}_{\rho^{2} O(r^{-2-\kappa})}\underbrace{|\nabla
\overline{\varphi}_{\delta}|}_{O(\delta^{2})}\right.\\
&+\underbrace{(|\psi^{1}|+|\psi^{1}_{0}|)}_{\sin^{3/2}\theta O(r^{-\kappa})}\left(
\underbrace{|\nabla\chi|}_{\rho O(r^{-3-\kappa})}+\underbrace{|\nabla\chi_{0}|}_{\rho O(r^{-4})}
+\underbrace{|\chi-\chi_{0}|}_{\rho^{2} O(r^{-3-\kappa})}\underbrace{|\nabla
\overline{\varphi}_{\delta}|}_{O(\delta^{2})}
+\underbrace{(|\psi^{1}|+|\psi^{1}_{0}|)}_{\sin^{3/2}\theta O(r^{-\kappa})}
\underbrace{(|\nabla\psi^{2}|
+|\nabla\psi^{2}_{0}|)}_{\sqrt{\cos\theta}O(r^{-2-\kappa})}\right.\\
&+\left.\left.
\underbrace{(|\psi^{2}|+|\psi^{2}_{0}|)}_{\cos^{3/2}\theta O(r^{-\kappa})}
\underbrace{(|\nabla\psi^{1}|
+|\nabla\psi^{1}_{0}|)}_{\sqrt{\sin\theta}O(r^{-2-\kappa})}
+\underbrace{(|\psi^{1}|+|\psi^{1}_{0}|)}_{\sin^{3/2}\theta O(r^{-\kappa})}
\underbrace{(|\psi^{2}|+|\psi^{2}_{0}|)}_{\cos^{3/2}\theta O(r^{-\kappa})}
\underbrace{|\nabla\overline{\varphi}_{\delta}|}_{O(\delta^{2})}\right)\right]\\
=& \sqrt{\rho}\sin\theta O(r^{-1-\kappa}),
\end{split}
\end{align}
and
\begin{equation}
|\overline{\Theta}^{2}_{\delta}|= \sqrt{\rho} \cos\theta O(r^{-1-\kappa}).
\end{equation}
Therefore
\begin{equation}
I_{5}\leq C\int_{0}^{\frac{\pi}{2}}\int_{\frac{1}{\delta}}^{\frac{2}{\delta}}
\frac{\cos\theta}{\rho^{3}\sin\theta}\underbrace{|\overline{\Theta}_{\delta}^{1}|^{2}}
_{\rho\sin^{2}\theta O(r^{-2-2\kappa})}r^{5}\sin(2\theta)dr d\theta\rightarrow 0,
\end{equation}
\begin{equation}
I_{8}\leq C\int_{0}^{\frac{\pi}{2}}\int_{\frac{1}{\delta}}^{\frac{2}{\delta}}
\frac{\cos\theta}{\rho \sin\theta}\underbrace{|\nabla\overline{\psi}_{\delta}^{1}|^{2}}
_{\sin\theta O(r^{-4-2\kappa})}r^{5}\sin(2\theta)dr d\theta\rightarrow 0,
\end{equation}
\begin{equation}
I_{9}\leq C\int_{0}^{\frac{\pi}{2}}\int_{\frac{1}{\delta}}^{\frac{2}{\delta}}
\rho^{-2}\underbrace{e^{-4\overline{U}_{\delta}}}_{O(1)}
\underbrace{|\overline{\Upsilon}_{\delta}|^{2}}
_{\rho O(r^{-4-2\kappa})}r^{5}\sin(2\theta)dr d\theta\rightarrow 0.
\end{equation}
It may be shown in an analogous way that $I_{6}$ and $I_{7}$ also converge to zero.
\end{proof}

In the next step of the cut-and-paste argument we
consider small balls centered at the origin.

\begin{lemma}
Set
\begin{equation}
F_{\delta}(\Psi)=
(U,V_{\delta},W_{\delta}
,\zeta^{1}_{\delta},\zeta^{2}_{\delta},\chi_{\delta},\psi^{1}_{\delta},
\psi^{2}_{\delta})
\end{equation}
with
\begin{align}
\begin{split}
(V_{\delta},W_{\delta}
,\zeta^{1}_{\delta},\zeta^{2}_{\delta},\chi_{\delta},
\psi^{1}_{\delta},\psi^{2}_{\delta})=&(V_0,W_0,\zeta^{1}_0,\zeta^{2}_0,\chi_{0},
\psi_{0}^{1},\psi_{0}^{2})\\
&
+\varphi_{\delta}(V-V_0,W-W_0,\zeta^{1}-\zeta^{1}_0,\zeta^{2}-\zeta^{2}_0,
\chi-\chi_{0},\psi^{1}-\psi^{1}_{0},\psi^{2}-\psi^{2}_{0}),
\end{split}
\end{align}
so that except for the first component $F_{\delta}(\Psi)$ agrees with $\Psi_{0}$ on $B_{\delta}$.
Then $\lim_{\delta\rightarrow 0}\mathcal{I}(F_{\delta}(\Psi))=\mathcal{I}(\Psi)$, and this also holds in the case that $\Psi\equiv\Psi_{0}$ outside of $B_{2/\delta}$.
\end{lemma}

\begin{proof}
Observe that
\begin{equation}\label{080}
\mathcal{I}(F_{\delta}(\Psi))
=\mathcal{I}_{r\leq\delta}(F_{\delta}(\Psi))
+\mathcal{I}_{\delta< r<2\delta}(F_{\delta}(\Psi))
+\mathcal{I}_{r\geq2\delta}(F_{\delta}(\Psi)),
\end{equation}
and by the dominated convergence theorem
$\mathcal{I}_{r\geq2\delta}(F_{\delta}(\Psi))
=\mathcal{I}_{r\geq2\delta}(\Psi)
\rightarrow \mathcal{I}(\Psi)$.
Furthermore
\begin{align}
\begin{split}
\mathcal{I}_{r\leq\delta}(F_{\delta}(\Psi))&=\int_{r\leq\delta}12|\nabla U|^{2}
+|\nabla V_{0}|^{2}+|\nabla W_{0}|^{2}
+\sinh^{2}W_{0}|\nabla(V_{0}+h_{2})|^{2}+\frac{e^{-6h_{1}-6U-h_{2}-V_{0}}}{\cosh W_{0}}|\Theta_{0}^{1}|^{2}\\
&+\int_{r\leq\delta}
e^{-6h_{1}-6U+h_{2}+V_{0}}\cosh W_{0}
\left|\Theta_{0}^{2}-e^{-h_{2}-V_{0}}\tanh W_{0}\Theta_{0}^{1}\right|^{2}
+\frac{e^{-2h_{1}-2U-h_{2}-h_{2}-V_{0}}}{\cosh W_{0}}|\nabla\psi^{1}_{0}|^{2}\\
&+\int_{r\leq\delta}e^{-2h_{1}-2U+h_{2}+V_{0}}\cosh W_{0}|\nabla\psi_{0}^{2}
-e^{-h_{2}-V_{0}}\tanh W_{0}\nabla\psi_{0}^{1}|^{2}
+e^{-4h_{1}-4U}|\Upsilon_{0}|^{2},
\end{split}
\end{align}
where all but the first term on the right-hand side may be estimated
by the reduced energy of $\Psi_{0}$ (and therefore converge to zero), since
$e^{-U}\leq Ce^{-U_{0}}$ near the origin;
the first term also converges to zero by the DCT.
Now write
\begin{align}\label{070}
\begin{split}
\mathcal{I}_{\delta< r<2\delta}(F_{\delta}(\Psi))
=&\underbrace{\int_{\delta<r<2\delta}
12|\nabla U|^{2}}_{I_{1}}
+\underbrace{\int_{\delta<r<2\delta}
|\nabla V_{\delta}|^{2}}_{I_{2}}
+\underbrace{\int_{\delta<r<2\delta}
|\nabla W_{\delta}|^{2}}_{I_{3}}\\
&+\underbrace{\int_{\delta<r<2\delta}
\sinh^{2}W_{\delta}|\nabla(V_{\delta}+h_{2})|^{2}}_{I_{4}}
+\underbrace{\int_{\delta<r<2\delta}
\frac{\cos\theta}{\rho^{3}\sin\theta}
\frac{e^{-V_{\delta}-6U}}{\cosh W_{\delta}}
|\Theta^{1}_{\delta}|^{2}}_{I_{5}}\\
&+\underbrace{\int_{\delta<r<2\delta}
\frac{\sin\theta}{\rho^{3}\cos\theta}
e^{V_{\delta}-6U}\cosh W_{\delta}
|\Theta^{2}_{\delta}-e^{-V_{\delta}}\cot\theta \tanh W_{\delta}
\Theta^{1}_{\delta}|^{2}}_{I_{6}}\\
&+\underbrace{\int_{\delta<r<2\delta}
\frac{\sin\theta}{\rho\cos\theta}
e^{V_{\delta}-2U}\cosh W_{\delta}
|\nabla\psi^{2}_{\delta}-e^{-V_{\delta}}\cot\theta \tanh W_{\delta}
\nabla\psi^{1}_{\delta}|^{2}}_{I_{7}}\\
&+\underbrace{\int_{\delta<r<2\delta}
\frac{\cos\theta}{\rho\sin\theta}
\frac{e^{-V_{\delta}-2U}}{\cosh W_{\delta}}
|\nabla\psi^{1}_{\delta}|^{2}}_{I_{8}}
+\underbrace{\int_{\delta<r<2\delta}
\rho^{-2}e^{-4U}|\Upsilon_{\delta}|^{2}
}_{I_{9}}.
\end{split}
\end{align}
We have
\begin{equation}\label{00072}
I_{2}\leq C\int_{0}^{\frac{\pi}{2}}\int_{\delta}^{2\delta}
\left(\underbrace{|\nabla V|^{2}}_{O(r^{-4})}+\underbrace{|\nabla V_{0}|^{2}}_{O(r^{-4})}
+\underbrace{(V-V_{0})^{2}}_{O(1)}\underbrace{|\nabla\varphi_{\delta}|^{2}}_{ O(\delta^{-4})}\right)r^{5}\sin(2\theta)dr d\theta\rightarrow 0,
\end{equation}
and similarly for $I_{1}$ as well as $I_{3}$.
Moreover, using
$\sinh W_{\delta}=\sqrt{\rho}O(r^{-1})$
yields
\begin{equation}
I_{4}\leq C\int_{0}^{\frac{\pi}{2}}\int_{\delta}^{2\delta}
\rho r^{-2}\left(\underbrace{|\nabla V|^{2}}_{O(r^{-4})}+
\underbrace{|\nabla V_{0}|^{2}}_{O(r^{-4})}
+\underbrace{|\nabla h_{2}|^{2}}_{O(\rho^{-2})}
+\underbrace{(V-V_{0})^{2}}_{O(1)}\underbrace{|\nabla\varphi_{\delta}|^{2}}_{ O(\delta^{-4})}\right)r^{5}\sin(2\theta)dr d\theta\rightarrow 0.
\end{equation}

Next observe that \eqref{FALLOFF4.0}, \eqref{FALLOFF4.0'}, \eqref{FALLOFF7.0}, \eqref{FALLOFF7.0'}, and \eqref{FALLOFF10.0'} together with $(\chi-\chi_{0})|_{\Gamma}=0$, $\psi^{i}|_{\Gamma_{i}}=0$, and $(\zeta^{i}-\zeta_{0}^{i})|_{\Gamma}=0$, $i=1,2$ give rise to the following estimates for $r\in[\delta,2\delta]$:
\begin{equation}\label{002.1'}
|(\chi-\chi_{0})(\rho,z,\phi)|
\leq\int_{0}^{\rho}|\partial_{\rho}(\chi-\chi_{0})
(\tilde{\rho},z,\phi)|d\tilde{\rho}
=\begin{cases}
\rho^{2}O(r^{-7+\kappa})\text{ }\text{ in the AF case,}\\
\rho^{2}O(r^{-5+\kappa})\text{ }\text{ in the AC case},
\end{cases}
\end{equation}
\begin{equation}\label{002.2'}
|\psi^{1}(\rho,z,\phi)|
\leq\int_{0}^{\rho}|\partial_{\rho}\psi^{1}
(\tilde{\rho},z,\phi)|d\tilde{\rho}
=\begin{cases}
\sin^{3/2}\theta O(r^{-1})\text{ }\text{ in the AF case,}\\
\sin^{3/2}\theta O(1)\text{ }\text{ in the AC case},
\end{cases}
\end{equation}
\begin{equation}\label{002.3'}
|\psi^{2}(\rho,z,\phi)|
\leq\int_{0}^{\rho}|\partial_{\rho}\psi^{2}
(\tilde{\rho},z,\phi)|d\tilde{\rho}
=\begin{cases}
\cos^{3/2}\theta O(r^{-1})\text{ }\text{ in the AF case,}\\
\cos^{3/2}\theta O(1)\text{ }\text{ in the AC case},
\end{cases}
\end{equation}
\begin{equation}\label{002'}
|(\zeta^{i}-\zeta^{i}_{0})(\rho,z,\phi)|
\leq\int_{0}^{\rho}|\partial_{\rho}(\zeta^{i}-\zeta^{i}_{0})
(\tilde{\rho},z,\phi)|d\tilde{\rho}
=\begin{cases}
\rho^{2}O(r^{-8+\kappa})\text{ }\text{ in the AF case,}\\
\rho^{2}O(r^{-5+\kappa})\text{ }\text{ in the AC case}.
\end{cases}
\end{equation}
From this we find that
\begin{align}
\begin{split}
|\nabla\psi^{1}_{\delta}|
\leq&\underbrace{|\nabla\psi^{1}|}_{\begin{cases}
\sqrt{\sin\theta} O(r^{-3})\text{ }\text{ AF}\\
\sqrt{\sin\theta} O(r^{-2})\text{ }\text{ AC}
\end{cases}}
+\underbrace{|\nabla\psi_{0}^{1}|}_{\sin\theta O(r^{-2})}
+\underbrace{|\psi^{1}-\psi^{1}_{0}|}_{\begin{cases}
\sin^{3/2}\theta O(r^{-1})\text{ }\text{ AF}\\
\sin^{3/2}\theta O(1)\text{ }\text{ AC}
\end{cases}}
\underbrace{|\nabla\varphi_{\delta}|}_{O(\delta^{-2})}\\
=&\begin{cases}
\sqrt{\sin\theta} O(r^{-3})\text{ }\text{ in the AF case}\\
\sqrt{\sin\theta} O(r^{-2})\text{ }\text{ in the AC case}
\end{cases},
\end{split}
\end{align}
and similarly
\begin{equation}
|\nabla\psi^{2}_{\delta}|
=\begin{cases}
\sqrt{\cos\theta} O(r^{-3})\text{ }\text{ in the AF case}\\
\sqrt{\cos\theta} O(r^{-2})\text{ }\text{ in the AC case}
\end{cases},
\end{equation}
as well as
\begin{align}
\begin{split}
|\Upsilon_{\delta}|&\leq C\Big(\underbrace{|\nabla\chi|}_{\begin{cases}
\rho O(r^{-7+\kappa})\text{ }\text{ AF}\\
\rho O(r^{-5+\kappa})\text{ }\text{ AC}
\end{cases}}+\underbrace{|\nabla\chi_{0}|}_{\rho O(r^{-4})}
+\underbrace{|\chi-\chi_{0}|}_{\begin{cases}
\rho^{2} O(r^{-7+\kappa})\text{ }\text{ AF}\\
\rho^{2} O(r^{-5+\kappa})\text{ }\text{ AC}
\end{cases}}
\underbrace{|\nabla
\overline{\varphi}_{\delta}|}_{O(\delta^{-2})}
\\
&+\underbrace{(|\psi^{1}|+|\psi^{1}_{0}|)}_{\begin{cases}
\sin^{3/2} O(r^{-1})\text{ }\text{ AF}\\
\sin^{3/2} O(1)\text{ }\text{ AC}
\end{cases}}
\underbrace{(|\nabla\psi^{2}|+|\nabla\psi^{2}_{0}|)}_{\begin{cases}
\sqrt{\cos\theta} O(r^{-3})\text{ }\text{ AF}\\
\sqrt{\cos\theta} O(r^{-2})\text{ }\text{ AC}
\end{cases}}
+\underbrace{(|\psi^{2}|+|\psi^{2}_{0}|)}_{\begin{cases}
\cos^{3/2} O(r^{-1})\text{ }\text{ AF}\\
\cos^{3/2} O(1)\text{ }\text{ AC}
\end{cases}}
\underbrace{(|\nabla\psi^{1}|+|\nabla\psi^{1}_{0}|)}_{\begin{cases}
\sqrt{\sin\theta} O(r^{-3})\text{ }\text{ AF}\\
\sqrt{\sin\theta} O(r^{-2})\text{ }\text{ AC}
\end{cases}}\\
&+\underbrace{(|\psi^{1}|+|\psi^{1}_{0}|)}_{\begin{cases}
\sin^{3/2} O(r^{-1})\text{ }\text{ AF}\\
\sin^{3/2} O(1)\text{ }\text{ AC}
\end{cases}}\underbrace{(|\psi^{2}|+|\psi^{2}_{0}|)}_{\begin{cases}
\cos^{3/2} O(r^{-1})\text{ }\text{ AF}\\
\cos^{3/2} O(1)\text{ }\text{ AC}
\end{cases}}
\underbrace{|\nabla\varphi_{\delta}|}_{O(\delta^{-2})}\Big)\\
=& \begin{cases}
\sqrt{\rho} O(r^{-6+\kappa})\text{ }\text{ in the AF case}\\
\sqrt{\rho} O(r^{-4+\kappa})\text{ }\text{ in the AC case}
\end{cases},
\end{split}
\end{align}
\begin{align}
\begin{split}
|\Theta_{\delta}^{1}|&\leq C\Big[\underbrace{|\nabla\zeta^{1}|}_{\begin{cases}
\rho\sqrt{\sin\theta} O(r^{-8+\kappa})\text{ }\text{ AF}\\
\rho\sqrt{\sin\theta} O(r^{-5+\kappa})\text{ }\text{ AC}
\end{cases}}+\underbrace{|\nabla\zeta^{1}_{0}|}_{\rho \sin^{2}\theta O(r^{-4})}
+\underbrace{|\zeta^{1}-\zeta^{1}_{0}|}_{\begin{cases}
\rho^{2} O(r^{-8+\kappa})\text{ }\text{ AF}\\
\rho^{2} O(r^{-5+\kappa})\text{ }\text{ AC}
\end{cases}}
\underbrace{|\nabla
\varphi_{\delta}|}_{O(\delta^{-2})}\\
&+\underbrace{(|\psi^{1}|+|\psi^{1}_{0}|)}_{\begin{cases}
\sin^{3/2}\theta O(r^{-1})\text{ }\text{ AF}\\
\sin^{3/2}\theta O(1)\text{ }\text{ AC}
\end{cases}}
\big(\underbrace{|\nabla\chi|}_{\begin{cases}
\rho O(r^{-7+\kappa})\text{ }\text{ AF}\\
\rho O(r^{-5+\kappa})\text{ }\text{ AC}
\end{cases}}+\underbrace{|\nabla\chi_{0}|}_{\rho O(r^{-4})}
+\underbrace{|\chi-\chi_{0}|}_{\begin{cases}
\rho^{2} O(r^{-7+\kappa})\text{ }\text{ AF}\\
\rho^{2} O(r^{-5+\kappa})\text{ }\text{ AC}
\end{cases}}
\underbrace{|\nabla
\overline{\varphi}_{\delta}|}_{O(\delta^{-2})}\\
&+\underbrace{(|\psi^{1}|+|\psi^{1}_{0}|)}_{\begin{cases}
\sin^{3/2} O(r^{-1})\text{ }\text{ AF}\\
\sin^{3/2} O(1)\text{ }\text{ AC}
\end{cases}}
\underbrace{(|\nabla\psi^{2}|+|\nabla\psi^{2}_{0}|)}_{\begin{cases}
\sqrt{\cos\theta} O(r^{-3})\text{ }\text{ AF}\\
\sqrt{\cos\theta} O(r^{-2})\text{ }\text{ AC}
\end{cases}}
+\underbrace{(|\psi^{2}|+|\psi^{2}_{0}|)}_{\begin{cases}
\cos^{3/2} O(r^{-1})\text{ }\text{ AF}\\
\cos^{3/2} O(1)\text{ }\text{ AC}
\end{cases}}
\underbrace{(|\nabla\psi^{1}|+|\nabla\psi^{1}_{0}|)}_{\begin{cases}
\sqrt{\sin\theta} O(r^{-3})\text{ }\text{ AF}\\
\sqrt{\sin\theta} O(r^{-2})\text{ }\text{ AC}
\end{cases}}\\
&+\underbrace{(|\psi^{1}|+|\psi^{1}_{0}|)}_{\begin{cases}
\sin^{3/2} O(r^{-1})\text{ }\text{ AF}\\
\sin^{3/2} O(1)\text{ }\text{ AC}
\end{cases}}\underbrace{(|\psi^{2}|+|\psi^{2}_{0}|)}_{\begin{cases}
\cos^{3/2} O(r^{-1})\text{ }\text{ AF}\\
\cos^{3/2} O(1)\text{ }\text{ AC}
\end{cases}}
\underbrace{|\nabla\varphi_{\delta}|}_{O(\delta^{-2})}\big)\Big]\\
=& \begin{cases}
\sqrt{\rho}\sin\theta O(r^{-7+\kappa})\text{ }\text{ in the AF case}\\
\sqrt{\rho}\sin\theta O(r^{-4+\kappa})\text{ }\text{ in the AC case}
\end{cases},
\end{split}
\end{align}
and
\begin{equation}
|\Theta^{2}_{\delta}|= \begin{cases}
\sqrt{\rho}\cos\theta O(r^{-7+\kappa})\text{ }\text{ in the AF case}\\
\sqrt{\rho}\cos\theta O(r^{-4+\kappa})\text{ }\text{ in the AC case}
\end{cases}.
\end{equation}
Therefore
\begin{equation}
I_{5}\leq C\int_{0}^{\frac{\pi}{2}}\int_{\delta}^{2\delta}
\frac{\cos\theta}{\rho^{3}\sin\theta}\underbrace{\frac{e^{-V_{\delta}-6U}}
{\cosh W_{\delta}}}
_{\begin{cases}
O(r^{12})\text{ }\text{ AF}\\
O(r^{6})\text{ }\text{ AC}
\end{cases}}
\underbrace{|\Theta_{\delta}^{1}|^{2}}
_{\begin{cases}
\rho\sin^{2}\theta O(r^{-14+2\kappa})\text{ }\text{ AF}\\
\rho\sin^{2}\theta O(r^{-8+2\kappa})\text{ }\text{ AC}
\end{cases}}
r^{5}\sin(2\theta)dr d\theta\rightarrow 0,
\end{equation}
\begin{equation}
I_{8}\leq C\int_{0}^{\frac{\pi}{2}}\int_{\delta}^{2\delta}
\frac{\cos\theta}{\rho \sin\theta}\underbrace{\frac{e^{-V_{\delta}-2U}}{\cosh W_{\delta}}}_{\begin{cases}
O(r^{4})\text{ }\text{ AF}\\
O(r^{2})\text{ }\text{ AC}
\end{cases}}
\underbrace{|\nabla\psi_{\delta}^{1}|^{2}}
_{\begin{cases}
\sin\theta O(r^{-6})\text{ }\text{ AF}\\
\sin\theta O(r^{-4})\text{ }\text{ AC}
\end{cases}}r^{5}\sin(2\theta)dr d\theta\rightarrow 0,
\end{equation}
\begin{equation}
I_{9}\leq C\int_{0}^{\frac{\pi}{2}}\int_{\delta}^{2\delta}
\rho^{-2}\underbrace{e^{-4U}}_{\begin{cases}
O(r^{8})\text{ }\text{ AF}\\
O(r^{4})\text{ }\text{ AC}
\end{cases}}
\underbrace{|\Upsilon_{\delta}|^{2}}
_{\begin{cases}
\rho O(r^{-12+2\kappa})\text{ }\text{ AF}\\
\rho O(r^{-8+2\kappa})\text{ }\text{ AC}
\end{cases}}r^{5}\sin(2\theta)dr d\theta\rightarrow 0.
\end{equation}
It may be shown in an analogous way that $I_{6}$ and $I_{7}$ also converge to zero.
\end{proof}

Lastly we treat the asymptotic regimes near the axes $\Gamma$ and away from the origin. For this purpose it will be useful to define the domains
\begin{equation}
\mathcal{C}_{\delta,\varepsilon}=
\{\rho\leq\varepsilon\}\cap
\{\delta\leq r\leq 2/\delta\},\text{ }\text{ }\text{ }\text{ }\text{ }\text{ }\text{ }\mathcal{W}_{\delta,\varepsilon}=
\{\varepsilon\leq \rho\leq\sqrt{\varepsilon}\}\cap
\{\delta\leq r\leq 2/\delta\}.
\end{equation}

\begin{lemma}
Set
\begin{equation}
G_{\varepsilon}(\Psi)=(U,V_{\varepsilon},
W_{\varepsilon},\zeta^{1}_{\varepsilon},\zeta^{2}_{\varepsilon},\chi_{\varepsilon},
\psi^{1}_{\varepsilon},\psi^{2}_{\varepsilon})
\end{equation}
with
\begin{align}
\begin{split}
(V_{\varepsilon},
W_{\varepsilon},\zeta^{1}_{\varepsilon},\zeta^{2}_{\varepsilon},
\chi_{\varepsilon},
\psi^{1}_{\varepsilon},\psi^{2}_{\varepsilon})&=
(V_{0},W_{0},\zeta^{1}_{0},\zeta^{2}_{0},
\chi_{0},
\psi^{1}_{0},\psi^{2}_{0})\\
&
+\phi_{\varepsilon}(V-V_{0},W-W_{0},\zeta^{1}-\zeta^{1}_{0},
\zeta^{2}-\zeta^{2}_{0},\chi-\chi_{0},
\psi^{1}-\psi^{1}_{0},\psi^{2}-\psi^{2}_{0}),
\end{split}
\end{align}
so that except for the first component $G_{\varepsilon}(\Psi)$ coincides with $\Psi_{0}$ when $\rho\leq\varepsilon$.
Fix $\delta>0$ and suppose that except for the first component $\Psi$ agrees with $\Psi_{0}$ on $B_{\delta}$, then
$\lim_{\varepsilon\rightarrow 0}\mathcal{I}(G_{\varepsilon}(\Psi))=\mathcal{I}(\Psi)$. This also holds if
$\Psi\equiv\Psi_{0}$ outside $B_{2/\delta}$.
\end{lemma}

\begin{proof}
Observe that
\begin{equation}
\mathcal{I}(G_{\varepsilon}(\Psi))
=\mathcal{I}_{\mathcal{C}_{\delta,\varepsilon}}(G_{\varepsilon}(\Psi))
+\mathcal{I}_{\mathcal{W}_{\delta,\varepsilon}}(G_{\varepsilon}(\Psi))
+\mathcal{I}_{\mathbb{R}^{3}\setminus(\mathcal{C}_{\delta,\varepsilon}
\cup\mathcal{W}_{\delta,\varepsilon})}(G_{\varepsilon}(\Psi)).
\end{equation}
Using the fact that except for the first component
$\Psi$ agrees with $\Psi_{0}$ on $B_{\delta}$, together with the DCT and finite energy of $\Psi_{0}$, shows
$\mathcal{I}_{\mathbb{R}^{3}\setminus(\mathcal{C}_{\delta,\varepsilon}
\cup\mathcal{W}_{\delta,\varepsilon})}(G_{\varepsilon}(\Psi))
\rightarrow \mathcal{I}(\Psi)$.
Furthermore
\begin{align}
\begin{split}
\mathcal{I}_{\mathcal{C}_{\delta,\varepsilon}}(G_{\varepsilon}(\Psi))
&=\int_{\mathcal{C}_{\delta,\varepsilon}}12|\nabla U|^{2}
+|\nabla V_{0}|^{2}+|\nabla W_{0}|^{2}
+\sinh^{2}W_{0}|\nabla(V_{0}+h_{2})|^{2}+\frac{e^{-6h_{1}-6U-h_{2}-V_{0}}}{\cosh W_{0}}|\Theta_{0}^{1}|^{2}\\
&+\int_{\mathcal{C}_{\delta,\varepsilon}}
e^{-6h_{1}-6U+h_{2}+V_{0}}\cosh W_{0}
\left|\Theta_{0}^{2}-e^{-h_{2}-V_{0}}\tanh W_{0}\Theta_{0}^{1}\right|^{2}
+\frac{e^{-2h_{1}-2U-h_{2}-h_{2}-V_{0}}}{\cosh W_{0}}|\nabla\psi^{1}_{0}|^{2}\\
&+\int_{\mathcal{C}_{\delta,\varepsilon}}e^{-2h_{1}-2U+h_{2}+V_{0}}\cosh W_{0}|\nabla\psi_{0}^{2}
-e^{-h_{2}-V_{0}}\tanh W_{0}\nabla\psi_{0}^{1}|^{2}
+e^{-4h_{1}-4U}|\Upsilon_{0}|^{2},
\end{split}
\end{align}
where all but the first term on the right-hand side may be estimated
by the reduced energy of $\Psi_{0}$ (and therefore converge to zero), since
$e^{-U}\leq Ce^{-U_{0}}$ near the origin;
the first term also converges to zero by the DCT.
Now write
\begin{align}\label{070}
\begin{split}
\mathcal{I}_{\mathcal{W}_{\delta,\varepsilon}}(G_{\varepsilon}(\Psi))
=&\underbrace{\int_{\mathcal{W}_{\delta,\varepsilon}}
12|\nabla U|^{2}}_{I_{1}}
+\underbrace{\int_{\mathcal{W}_{\delta,\varepsilon}}
|\nabla V_{\varepsilon}|^{2}}_{I_{2}}
+\underbrace{\int_{\mathcal{W}_{\delta,\varepsilon}}
|\nabla W_{\varepsilon}|^{2}}_{I_{3}}\\
&+\underbrace{\int_{\mathcal{W}_{\delta,\varepsilon}}
\sinh^{2}W_{\varepsilon}|\nabla(V_{\varepsilon}+h_{2})|^{2}}_{I_{4}}
+\underbrace{\int_{\mathcal{W}_{\delta,\varepsilon}}
\frac{\cos\theta}{\rho^{3}\sin\theta}
\frac{e^{-V_{\varepsilon}-6U}}{\cosh W_{\varepsilon}}
|\Theta^{1}_{\varepsilon}|^{2}}_{I_{5}}\\
&+\underbrace{\int_{\mathcal{W}_{\delta,\varepsilon}}
\frac{\sin\theta}{\rho^{3}\cos\theta}
e^{V_{\varepsilon}-6U}\cosh W_{\varepsilon}
|\Theta^{2}_{\varepsilon}-e^{-V_{\varepsilon}}\cot\theta \tanh W_{\varepsilon}
\Theta^{1}_{\varepsilon}|^{2}}_{I_{6}}\\
&+\underbrace{\int_{\mathcal{W}_{\delta,\varepsilon}}
\frac{\sin\theta}{\rho\cos\theta}
e^{V_{\varepsilon}-2U}\cosh W_{\varepsilon}
|\nabla\psi^{2}_{\varepsilon}-e^{-V_{\varepsilon}}\cot\theta \tanh W_{\varepsilon}
\nabla\psi^{1}_{\varepsilon}|^{2}}_{I_{7}}\\
&+\underbrace{\int_{\mathcal{W}_{\delta,\varepsilon}}
\frac{\cos\theta}{\rho\sin\theta}
\frac{e^{-V_{\varepsilon}-2U}}{\cosh W_{\varepsilon}}
|\nabla\psi^{1}_{\varepsilon}|^{2}}_{I_{8}}
+\underbrace{\int_{\mathcal{W}_{\delta,\varepsilon}}
\rho^{-2}e^{-4U}|\Upsilon_{\varepsilon}|^{2}
}_{I_{9}}.
\end{split}
\end{align}
We have
\begin{equation}
I_{2}\leq C\int_{\delta/2}^{3/\delta}\int_{\varepsilon}^{\sqrt{\varepsilon}}
\left(\underbrace{|\nabla V|^{2}}_{O(1)}+\underbrace{|\nabla V_{0}|^{2}}_{O(1)}
+\underbrace{(V-V_{0})^{2}}_{O(1)}\underbrace{|\nabla\phi_{\varepsilon}|^{2}}_{ O\left((\rho\log\varepsilon)^{-2}\right)}\right)\rho d\rho d|z|\rightarrow 0,
\end{equation}
and similarly for $I_{1}$ and $I_{3}$.
Moreover, using
$\sinh W_{\varepsilon}=O(\sqrt{\rho})$
yields
\begin{equation}
I_{4}\leq C\int_{\delta/2}^{3/\delta}\int_{\varepsilon}^{\sqrt{\varepsilon}}
\rho\left(\underbrace{|\nabla V|^{2}}_{O(1)}+
\underbrace{|\nabla V_{0}|^{2}}_{O(1)}
+\underbrace{|\nabla h_{2}|^{2}}_{O(\rho^{-2})}
+\underbrace{(V-V_{0})^{2}}_{O(1)}\underbrace{|\nabla\phi_{\varepsilon}|^{2}}_{ O\left((\rho\log\varepsilon)^{-2}\right)}\right)\rho d\rho d|z|\rightarrow 0.
\end{equation}

Next observe that \eqref{FALLOFF4.1}, \eqref{FALLOFF4.1'}, and \eqref{FALLOFF10.1'} together with $(\chi-\chi_{0})|_{\Gamma}=0$, $\psi^{i}|_{\Gamma_{i}}=0$, and $(\zeta^{i}-\zeta_{0}^{i})|_{\Gamma}=0$, $i=1,2$ give rise to the following estimates on $\mathcal{W}_{\delta,\varepsilon}$:
\begin{equation}\label{002.1''}
|(\chi-\chi_{0})(\rho,z,\phi)|
\leq\int_{0}^{\rho}|\partial_{\rho}(\chi-\chi_{0})
(\tilde{\rho},z,\phi)|d\tilde{\rho}
=O(\rho^{2}),
\end{equation}
\begin{equation}\label{002.2''}
|\psi^{1}(\rho,z,\phi)|
\leq\int_{0}^{\rho}|\partial_{\rho}\psi^{1}
(\tilde{\rho},z,\phi)|d\tilde{\rho}
=\sin^{3/2}\theta O(1),
\end{equation}
\begin{equation}\label{002.3''}
|\psi^{2}(\rho,z,\phi)|
\leq\int_{0}^{\rho}|\partial_{\rho}\psi^{2}
(\tilde{\rho},z,\phi)|d\tilde{\rho}
=\cos^{3/2}\theta O(1),
\end{equation}
\begin{equation}\label{002''}
|(\zeta^{i}-\zeta^{i}_{0})(\rho,z,\phi)|
\leq\int_{0}^{\rho}|\partial_{\rho}(\zeta^{i}-\zeta^{i}_{0})
(\tilde{\rho},z,\phi)|d\tilde{\rho}
=O(\rho^{2}).
\end{equation}
From this we find that
\begin{equation}
|\nabla\psi^{1}_{\varepsilon}|
\leq\underbrace{|\nabla\psi^{1}|}_{\sqrt{\sin\theta}O(1)}
+\underbrace{|\nabla\psi_{0}^{1}|}_{\sin\theta O(1)}
+\underbrace{|\psi^{1}-\psi^{1}_{0}|}_{\sin^{3/2}\theta O(1)}\underbrace{|\nabla\phi_{\varepsilon}|}_{O(|\rho\log\varepsilon|^{-1})}
=\sqrt{\sin\theta}\left(O(1)+\frac{O(|\log\varepsilon|^{-1})}
{\cos\theta}\right),
\end{equation}
and similarly
\begin{equation}
|\nabla\psi^{2}_{\varepsilon}|
=\sqrt{\cos\theta}\left(O(1)+\frac{O(|\log\varepsilon|^{-1})}
{\sin\theta}\right),
\end{equation}
as well as
\begin{align}
\begin{split}
|\Upsilon_{\varepsilon}|\leq& C\left(\underbrace{|\nabla\chi|}_{O(\rho)}
+\underbrace{|\nabla\chi_{0}|}_{O(\rho)}
+\underbrace{|\chi-\chi_{0}|}_{O(\rho^{2})}\underbrace{|\nabla
\phi_{\varepsilon}|}_{O(|\rho\log\varepsilon|^{-1})}
+\underbrace{(|\psi^{1}|+|\psi^{1}_{0}|)}_{\sin^{3/2}\theta O(1)}
\underbrace{(|\nabla\psi^{2}|+|\nabla\psi^{2}_{0}|)}_{\sqrt{\cos\theta}
O(1)}\right.\\
&
+\left.\underbrace{(|\psi^{2}|+|\psi^{2}_{0}|)}_{\cos^{3/2}\theta O(1)}
\underbrace{(|\nabla\psi^{1}|+|\nabla\psi^{1}_{0}|)}_{\sqrt{\sin\theta}
O(1)}
+\underbrace{(|\psi^{1}|+|\psi^{1}_{0}|)}_{\sin^{3/2}\theta O(1)}\underbrace{(|\psi^{2}|+|\psi^{2}_{0}|)}_{\cos^{3/2}\theta O(1)}
\underbrace{|\nabla\phi_{\varepsilon}|}_{O(|\rho\log\varepsilon|^{-1})}\right)\\
=& O(\sqrt{\rho}),
\end{split}
\end{align}
\begin{align}
\begin{split}
|\Theta_{\varepsilon}^{1}|\leq& C\left[\underbrace{|\nabla\zeta^{1}|}_{ \sqrt{\sin\theta}O(\rho)}+\underbrace{|\nabla\zeta^{1}_{0}|}_{
\sin^{2}\theta O(\rho)}
+\underbrace{|\zeta^{1}-\zeta^{1}_{0}|}_{
O(\rho^{2})}\underbrace{|\nabla
\phi_{\varepsilon}|}_{O(|\rho\log\varepsilon|^{-1})}
+\underbrace{(|\psi^{1}|+|\psi^{1}_{0}|)}_{\sin^{3/2}\theta O(1)}\left(
\underbrace{|\nabla\chi|}_{O(\rho)}
+\underbrace{|\nabla\chi_{0}|}_{O(\rho)}\right.\right.\\
&
+\underbrace{|\chi-\chi_{0}|}_{O(\rho^{2})}\underbrace{|\nabla
\phi_{\varepsilon}|}_{O(|\rho\log\varepsilon|^{-1})}
+\underbrace{(|\psi^{1}|+|\psi^{1}_{0}|)}_{\sin^{3/2}\theta O(1)}
\underbrace{(|\nabla\psi^{2}|+|\nabla\psi^{2}_{0}|)}_{\sqrt{\cos\theta}
O(1)}
+\underbrace{(|\psi^{2}|+|\psi^{2}_{0}|)}_{\cos^{3/2}\theta O(1)}
\underbrace{(|\nabla\psi^{1}|+|\nabla\psi^{1}_{0}|)}_{\sqrt{\sin\theta}
O(1)}\\
&+\left.\left.\underbrace{(|\psi^{1}|+|\psi^{1}_{0}|)}_{\sin^{3/2}\theta
O(1)}\underbrace{(|\psi^{2}|+|\psi^{2}_{0}|)}_{\cos^{3/2}\theta
O(1)}
\underbrace{|\nabla\phi_{\varepsilon}|}_{O(|\rho\log\varepsilon|^{-1})}\right)\right]\\
=& \left(\sqrt{\sin\theta}+|\log\varepsilon|^{-1}
+\frac{\sin\theta}{\sqrt{\cos\theta}}|\log\varepsilon|^{-1}\right)O(\rho),
\end{split}
\end{align}
and
\begin{equation}
|\Theta^{2}_{\varepsilon}|= \left(\sqrt{\cos\theta}+|\log\varepsilon|^{-1}
+\frac{\cos\theta}{\sqrt{\sin\theta}}|\log\varepsilon|^{-1}\right)O(\rho).
\end{equation}
Therefore
\begin{equation}
I_{5}\leq C\int_{\delta/2}^{3/\delta}\int_{\varepsilon}^{\sqrt{\varepsilon}}
\frac{\cos\theta}{\rho^{3}\sin\theta}\underbrace{\frac{e^{-V_{\varepsilon}-6U}}
{\cosh W_{\varepsilon}}}_{O(1)}\underbrace{|\Theta_{\varepsilon}^{1}|^{2}}
_{\left(\sin\theta+(\log\varepsilon)^{-2}
+\frac{\sin^{2}\theta}{\cos\theta}(\log\varepsilon)^{-2}\right)O(\rho^{2})}\rho d\rho d|z|=O(|\log\varepsilon|^{-1})\rightarrow 0,
\end{equation}
\begin{equation}
I_{8}\leq C\int_{\delta/2}^{3/\delta}\int_{\varepsilon}^{\sqrt{\varepsilon}}
\frac{\cos\theta}{\rho \sin\theta}\underbrace{\frac{e^{-V_{\varepsilon}-2U}}
{\cosh W_{\varepsilon}}}_{O(1)}\underbrace{|\nabla\psi_{\varepsilon}^{1}|^{2}}
_{\sin\theta\left(O(1)+\frac{O((\log\varepsilon)^{-2})}{\cos^{2}\theta}\right)}\rho d\rho d|z|=O(|\log\varepsilon|^{-1})\rightarrow 0,
\end{equation}
\begin{equation}
I_{9}\leq C\int_{\delta/2}^{3/\delta}\int_{\varepsilon}^{\sqrt{\varepsilon}}
\rho^{-2}\underbrace{e^{-4U_{\varepsilon}}}_{O(1)}
\underbrace{|\Upsilon_{\varepsilon}|^{2}}
_{O(\rho)}\rho d\rho d|z|=O(\sqrt{\varepsilon})\rightarrow 0.
\end{equation}
It may be shown in an analogous way that $I_{6}$ and $I_{7}$ also converge to zero.
\end{proof}

Consider now the composition of the three cut-and-paste operations defined above, namely
\begin{equation}
\Psi_{\delta,\varepsilon}
=G_{\varepsilon}\left(F_{\delta}\left(
\overline{F}_{\delta}(\Psi)\right)\right).
\end{equation}
Taken together, the previous three lemmas prove the following approximation result.

\begin{prop}\label{Proposition1}
Let $\varepsilon\ll\delta\ll 1$ and suppose that $\Psi$ satisfies the hypotheses of Theorem \ref{minimum}. Then
$\Psi_{\delta,\varepsilon}$ satisfies \eqref{54} and
\begin{equation}\label{106}
\lim_{\delta\rightarrow 0}\lim_{\varepsilon\rightarrow 0}
\mathcal{I}(\Psi_{\delta,\varepsilon})=\mathcal{I}(\Psi).
\end{equation}
\end{prop}

\section{Proof of the Gap Bound: Theorem \ref{minimum}}
\label{secproof}

Consider the geodesic in $G_{2(2)}/SO(4)$ connecting $\tilde{\Psi}_{0}$ to
$\tilde{\Psi}_{\delta,\varepsilon}$ as described in section \ref{dirichlet}, and denote it by $\tilde{\Psi}^{t}_{\delta,\varepsilon}$.
Since $\Psi_{\delta,\varepsilon}$ satisfies
\eqref{54}, we have
$U^{t}_{\delta,\varepsilon}=U_{0}+t(U_{\delta,\varepsilon}-U_{0})$ and $V^{t}_{\delta,\varepsilon}=V_{0}$ on $\mathcal{A}_{\delta,\varepsilon}$. Observe that
\begin{equation}
\frac{d^{2}}{dt^{2}}\mathcal{I}(\Psi^{t}_{\delta,\varepsilon})
=\underbrace{\frac{d^{2}}{dt^{2}}\mathcal{I}_{\Omega_{\delta,\varepsilon}}
(\Psi^{t}_{\delta,\varepsilon})}_{I_{1}}+
\underbrace{\frac{d^{2}}{dt^{2}}\mathcal{I}_{\mathcal{A}_{\delta,\varepsilon}}
(\Psi^{t}_{\delta,\varepsilon})}_{I_{2}}.
\end{equation}
Then convexity of the harmonic energy implies
\begin{align}\label{integral0}
\begin{split}
I_{1}=&\frac{d^{2}}{dt^{2}}E_{\Omega_{\delta,\varepsilon}}
(\tilde{\Psi}^{t}_{\delta,\varepsilon})
-\frac{d^{2}}{dt^{2}}\int_{\partial\Omega_{\delta,\varepsilon}
\cap\partial\mathcal{A}_{\delta,\varepsilon}}
12\left[h_{1}+2(U_{0}
+t(U_{\delta,\varepsilon}-U_{0}))\right]\partial_{\nu}h_{1}\\
&-\frac{d^{2}}{dt^{2}}\int_{\partial\Omega_{\delta,\varepsilon}
\cap\partial\mathcal{A}_{\delta,\varepsilon}}(h_{2}+V_{0})\partial_{\nu}h_{2}\\
\geq & 2\int_{\Omega_{\delta,\varepsilon}}
|\nabla\operatorname{dist}_{G_{2(2)}/SO(4)}
(\Psi_{\delta,\varepsilon},\Psi_{0})|^{2},
\end{split}
\end{align}
and the fact that $\operatorname{dist}_{G_{2(2)}/SO(4)}(\Psi_{\delta,\varepsilon},\Psi_{0})
=\sqrt{12}|U_{\delta,\varepsilon}-U_{0}|$ on $\mathcal{A}_{\delta,\varepsilon}$ yields
\begin{align}\label{integral}
\begin{split}
I_{2}=&\int_{\mathcal{A}_{\delta,\varepsilon}}
24|\nabla(U_{\delta,\varepsilon}-U_{0})|^{2}
+36(U_{\delta,\varepsilon}-U_{0})^{2}
\frac{e^{-6h_{1}-6U_{\delta,\varepsilon}^{t}-h_{2}-V_{0}}}{\cosh W_{0}}|\Theta^{1}_{0}|^{2}\\
&+\int_{\mathcal{A}_{\delta,\varepsilon}}36(U_{\delta,\varepsilon}-U_{0})^{2}
e^{-6h_{1}-6U_{\delta,\varepsilon}^{t}+h_{2}+V_{0}}\cosh W_{0}|e^{-h_{2}-V_{0}}
\tanh W_{0}\Theta^{1}_{0}-\Theta_{0}^{2}|^{2}\\
&+\int_{\mathcal{A}_{\delta,\varepsilon}}4(U_{\delta,\varepsilon}-U_{0})^{2}
e^{-2h_{1}-2U^{t}_{\delta,\varepsilon}+h_{2}+V_{0}}
\cosh W_{0}|e^{-h_{2}-V_{0}}\tanh W_{0}\nabla\psi_{0}^{1}
-\nabla\psi_{0}^{2}|^{2}\\
&+\int_{\mathcal{A}_{\delta,\varepsilon}}4(U_{\delta,\varepsilon}-U_{0})^{2}
\left(\frac{e^{-2h_{1}-2U^{t}_{\delta,\varepsilon}-h_{2}-V_{0}}}{\cosh W_{0}}|\nabla\psi^{1}_{0}|^{2}+4e^{-4h_{1}-4U^{t}_{\delta,\varepsilon}}
|\Upsilon_{0}|^{2}\right)\\
\geq &2\int_{\mathcal{A}_{\delta,\varepsilon}}
|\nabla\operatorname{dist}_{G_{2(2)}/SO(4)}
(\Psi_{\delta,\varepsilon},\Psi_{0})|^{2},
\end{split}
\end{align}
as long as interchanging
$\frac{d^{2}}{dt^{2}}$ and the integral in \eqref{integral} is justified. In order to show that this is the case, it is enough to prove that each of the terms in \eqref{integral} is uniformly integrable. First note that $U_{\delta,\varepsilon}$ and $U_{0}$ have square integrable derivatives on $\mathbb{R}^3$, and thus the first term satisfies the desired property. All remaining terms may be treated similarly to the second term, which we now examine. Clearly uniform integrability will hold if
$(U_{\delta,\varepsilon}-U_{0})^{2}e^{-6t(U_{\delta,\varepsilon}-U_{0})}$ is uniformly bounded in
$\mathcal{A}_{\delta,\varepsilon}=\mathcal{C}_{\delta,\varepsilon}
\cup B_{\delta}$, as the entire second term would then be dominated by the reduced energy of $\Psi_{0}$. We have that $U$ and $U_{0}$ are bounded on $\mathcal{C}_{\delta,\varepsilon}$, and
$|U_{\delta,\varepsilon}-U_{0}|\leq C|\log r|$ on $B_{\delta}$. Moreover, since $r^{6t}(\log r)^{2}$ remains uniformly bounded for $0<t_{0}<t\leq 1$, where $t_{0}\neq 0$ is arbitrarily small, the desired result follows away from $t=0$. Therefore, combining \eqref{integral0} and \eqref{integral} establishes
\eqref{55} for $t\in(0,1]$.

The next task at hand is to prove \eqref{56} for $\Psi_{\delta,\varepsilon}$, which will follow
from the harmonic map equations for $\Psi_{0}$ (see Appendix \ref{app2}).
Fix $\varepsilon_{0}<\varepsilon$ and
$\delta_{0}<\delta$ and consider
\begin{equation}
\frac{d}{dt}\mathcal{I}(\Psi^{t}_{\delta,\varepsilon})
=\underbrace{\frac{d}{dt}\mathcal{I}_{\Omega_{\delta_{0},\varepsilon_{0}}}(\Psi^{t}_{\delta,\varepsilon})}_{I_{3}}+
\underbrace{\frac{d}{dt}\mathcal{I}_{\mathcal{A}_{\delta_{0},\varepsilon_{0}}}(\Psi^{t}_{\delta,\varepsilon})}_{I_{4}}.
\end{equation}
Note that we may interchange $t$-derivatives and integration when $t\in(0,1]$, for reasons that are analogous to those outlined above. Thus, applying the harmonic map equations for $\Psi_{0}$, and the fact that $\frac{d}{dt}\Psi_{\delta,\varepsilon}^{t}|_{t=0}=(U_{\delta,\varepsilon}-U_{0})\partial_{u}$
on $\mathcal{A}_{\delta_{0},\varepsilon_{0}}$, implies that for $t$ small
\begin{equation}
I_{3}=O(t)-\int_{\partial B_{\delta_{0}}}24(U_{\delta,\varepsilon}-U_{0})\partial_{\nu}U_{0}
-\int_{\partial\mathcal{C}_{\delta_{0},\varepsilon_{0}}}
24(U_{\delta,\varepsilon}-U_{0})\partial_{\nu}U_{0}.
\end{equation}
Furthermore, since
$U^{t}_{\delta,\varepsilon}=U_{0}+t(U_{\delta,\varepsilon}-U_{0})$
and
\begin{equation}
\frac{d}{dt}V^{t}_{\delta,\varepsilon}=
\frac{d}{dt}W^{t}_{\delta,\varepsilon}
=\frac{d}{dt}\zeta^{1,t}_{\delta,\varepsilon}=\frac{d}{dt}\zeta^{2,t}_{\delta,\varepsilon}
=\frac{d}{dt}\chi^{t}_{\delta,\varepsilon}=\frac{d}{dt}\psi^{1,t}_{\delta,\varepsilon}
=\frac{d}{dt}\psi^{2,t}_{\delta,\varepsilon}=0\text{ }\text{ }\text{ on }\text{ } \mathcal{A}_{\delta_{0},\varepsilon_{0}},
\end{equation}
we have
\begin{align}\label{mnbv}
\begin{split}
I_{4}=&O(t)+\int_{\mathcal{A}_{\delta_{0},\varepsilon_{0}}}
24\nabla U_{0}\cdot\nabla(U_{\delta,\varepsilon}-U_{0})
-6(U_{\delta,\varepsilon}-U_{0})
\frac{e^{-6h_{1}-6U_{\delta,\varepsilon}^{t}-h_{2}-V_{0}}}{\cosh W_{0}}|\Theta^{1}_{0}|^{2}\\
&-\int_{\mathcal{A}_{\delta_{0},\varepsilon_{0}}}6(U_{\delta,\varepsilon}-U_{0})
e^{-6h_{1}-6U_{\delta,\varepsilon}^{t}+h_{2}+V_{0}}\cosh W_{0}|e^{-h_{2}-V_{0}}
\tanh W_{0}\Theta^{1}_{0}-\Theta_{0}^{2}|^{2}\\
&-\int_{\mathcal{A}_{\delta_{0},\varepsilon_{0}}}2(U_{\delta,\varepsilon}-U_{0})
e^{-2h_{1}-2U^{t}_{\delta,\varepsilon}+h_{2}+V_{0}}
\cosh W_{0}|e^{-h_{2}-V_{0}}\tanh W_{0}\nabla\psi_{0}^{1}
-\nabla\psi_{0}^{2}|^{2}\\
&-\int_{\mathcal{A}_{\delta_{0},\varepsilon_{0}}}2(U_{\delta,\varepsilon}-U_{0})
\left(\frac{e^{-2h_{1}-2U^{t}_{\delta,\varepsilon}-h_{2}-V_{0}}}{\cosh W_{0}}|\nabla\psi^{1}_{0}|^{2}+2e^{-4h_{1}-4U^{t}_{\delta,\varepsilon}}
|\Upsilon_{0}|^{2}\right).
\end{split}
\end{align}
Now observe that
\begin{equation}
\left|\int_{\partial B_{\delta_{0}}}\underbrace{(U_{\delta,\varepsilon}-U)}_{O(|\log\delta_{0}|)}
\underbrace{\partial_{\nu}U_{0}}_{O(\delta^{-2}_{0})}\right|
\leq C|\log\delta_{0}|\delta_{0}^{2}\rightarrow 0 \text{ }\text{ }\text{ }\text{ and }\text{ }\text{ }\text{ }\left|\int_{\partial \mathcal{C}_{\delta_{0},\varepsilon_{0}}}\underbrace{(U_{\delta,\varepsilon}-U)}
_{O(1)}\underbrace{\partial_{\nu}U_{0}}_{O(1)}\right|
\leq C\varepsilon_{0}\rightarrow 0
\end{equation}
allow an integration by parts in \eqref{mnbv}, from which we obtain $I_{3}+I_{4}=O(t)$ after using the harmonic map equation for $U_{0}$ together with the fact that $\mathcal{A}_{\delta_{0},\varepsilon_{0}}=B_{\delta_{0}}\cup\mathcal{C}_{\delta_{0},\varepsilon_{0}}$.
It follows that \eqref{56} holds for $\Psi_{\delta,\varepsilon}$.

Now combine \eqref{55} and \eqref{56} to find
\begin{align}
\begin{split}
\mathcal{I}(\Psi_{\delta,\varepsilon})-\mathcal{I}(\Psi_{0})
&\geq 2\int_{\mathbb{R}^{3}}|\nabla\operatorname{dist}_{G_{2(2)}/SO(4)}
(\Psi_{\delta,\varepsilon},\Psi_{0})|^{2}dx\\
&\geq C\left(\int_{\mathbb{R}^{3}}
\operatorname{dist}_{G_{2(2)}/SO(4)}^{6}(\Psi_{\delta,\varepsilon},\Psi_{0})dx
\right)^{\frac{1}{3}},
\end{split}
\end{align}
where the second line arises from the Gagliardo-Nirenberg-Sobolev inequality.
By the triangle
inequality
\begin{align}\label{AAAA}
\begin{split}
&\left(\operatorname{dist}_{G_{2(2)}/SO(4)}(\Psi_{\delta,\varepsilon},\Psi)
-\operatorname{dist}_{G_{2(2)}/SO(4)}(\Psi,\Psi_{0})\right)^{6}
-\operatorname{dist}_{G_{2(2)}/SO(4)}^{6}(\Psi,\Psi_{0})\\
\leq&\operatorname{dist}_{G_{2(2)}/SO(4)}^{6}(\Psi_{\delta,\varepsilon},\Psi_{0})
-\operatorname{dist}_{G_{2(2)}/SO(4)}^{6}(\Psi,\Psi_{0})\\
\leq&\left(\operatorname{dist}_{G_{2(2)}/SO(4)}(\Psi_{\delta,\varepsilon},\Psi)
+\operatorname{dist}_{G_{2(2)}/SO(4)}(\Psi,\Psi_{0})\right)^{6}
-\operatorname{dist}_{G_{2(2)}/SO(4)}^{6}(\Psi,\Psi_{0}).
\end{split}
\end{align}
Therefore if
\begin{equation}\label{LKJ}
\lim_{\delta\rightarrow 0}\lim_{\varepsilon\rightarrow 0}\int_{\mathbb{R}^{3}}
\operatorname{dist}_{G_{2(2)}/SO(4)}^{6}(\Psi_{\delta,\varepsilon},\Psi)dx
=0,
\end{equation}
then the proof of Theorem \ref{minimum} will be complete in light of Proposition \ref{Proposition1}. By the triangle inequality
\begin{align}\label{aaa}
\begin{split}
&\operatorname{dist}_{G_{2(2)}/SO(4)}(\Psi_{\delta,\varepsilon},\Psi)\\
\leq&\operatorname{dist}_{G_{2(2)}/SO(4)}
((U_{\delta,\varepsilon},V_{\delta,\varepsilon},W_{\delta,\varepsilon},
\zeta^{1}_{\delta,\varepsilon},\zeta^{2}_{\delta,\varepsilon},\chi_{\delta,\varepsilon},
\psi^{1}_{\delta,\varepsilon},\psi^{2}_{\delta,\varepsilon}),
(U,V_{\delta,\varepsilon},W_{\delta,\varepsilon},
\zeta^{1}_{\delta,\varepsilon},\zeta^{2}_{\delta,\varepsilon},\chi_{\delta,\varepsilon},
\psi^{1}_{\delta,\varepsilon},\psi^{2}_{\delta,\varepsilon}))\\
&+\operatorname{dist}_{G_{2(2)}/SO(4)}
((U,V_{\delta,\varepsilon},W_{\delta,\varepsilon},
\zeta^{1}_{\delta,\varepsilon},\zeta^{2}_{\delta,\varepsilon},\chi_{\delta,\varepsilon},
\psi^{1}_{\delta,\varepsilon},\psi^{2}_{\delta,\varepsilon}),
(U,V,W_{\delta,\varepsilon},
\zeta^{1}_{\delta,\varepsilon},\zeta^{2}_{\delta,\varepsilon},\chi_{\delta,\varepsilon},
\psi^{1}_{\delta,\varepsilon},\psi^{2}_{\delta,\varepsilon}))\\
&+\cdots+
\operatorname{dist}_{G_{2(2)}/SO(4)}
((U,V,W,
\zeta^{1},\zeta^{2},\chi,
\psi^{1},\psi^{2}_{\delta,\varepsilon}),
(U,V,W,
\zeta^{1},\zeta^{2},\chi,
\psi^{1},\psi^{2}))\\
\leq& C\left[|U-U_{\delta,\varepsilon}|+|V-V_{\delta,\varepsilon}|
+|W-W_{\delta,\varepsilon}|
+e^{-3U-3h_{1}}\left(e^{-\tfrac{1}{2}V-\tfrac{1}{2}h_{2}}
|\zeta^{1}-\zeta^{1}_{\delta,\varepsilon}|
+e^{\tfrac{1}{2}V+\tfrac{1}{2}h_{2}}
|\zeta^{2}-\zeta^{2}_{\delta,\varepsilon}|\right)\right]\\
&+ Ce^{-3U-3h_{1}}\left(e^{-\tfrac{1}{2}V-\tfrac{1}{2}h_{2}}
(|\psi^1|+|\psi^{1}_{0}|)
+e^{\tfrac{1}{2}V+\tfrac{1}{2}h_{2}}
(|\psi^2|+|\psi^{2}_{0}|)\right)|\chi-\chi_{\delta,\varepsilon}|\\
&+Ce^{-3U-3h_{1}}\left((|\psi^1|+|\psi^{1}_{0}|)
(|\psi^2|+|\psi^{2}_{0}|)e^{-\tfrac{1}{2}V-\tfrac{1}{2}h_{2}}
+(|\psi^2|+|\psi^{2}_{0}|)^{2}e^{\tfrac{1}{2}V+\tfrac{1}{2}h_{2}}\right)
|\psi^{1}-\psi^{1}_{\delta,\varepsilon}|\\
&+Ce^{-3U-3h_{1}}\left((|\psi^1|+|\psi^{1}_{0}|)^{2}
e^{-\tfrac{1}{2}V-\tfrac{1}{2}h_{2}}
+(|\psi^1|+|\psi^{1}_{0}|)(|\psi^2|+|\psi^{2}_{0}|)
e^{\tfrac{1}{2}V+\tfrac{1}{2}h_{2}}\right)
|\psi^{2}-\psi^{2}_{\delta,\varepsilon}|\\
&+Ce^{-2U-2h_{1}}\left(|\chi-\chi_{\delta,\varepsilon}|
+(|\psi^2|+|\psi^{2}_{0}|)|\psi^{1}-\psi^{1}_{\delta,\varepsilon}|
+(|\psi^1|+|\psi^{1}_{0}|)|\psi^{2}-\psi^{2}_{\delta,\varepsilon}|\right)\\
&+Ce^{-U-h_{1}}\left(e^{-\tfrac{1}{2}V-\tfrac{1}{2}h_{2}}|\psi^{1}-\psi^{1}_{\delta,\varepsilon}|
+e^{\tfrac{1}{2}V+\tfrac{1}{2}h_{2}}|\psi^{2}-\psi^{2}_{\delta,\varepsilon}|\right),
\end{split}
\end{align}
where it was used that distances between points of $G_{2(2)}/SO(4)$ are dominated by the length of connecting coordinate lines.

Each term on the right-hand side of \eqref{aaa} involves the difference of a component of $\Psi$ with the corresponding component of $\Psi_{\delta,\varepsilon}$. Since such expressions vanish outside of the domains $\mathbb{R}^{3}\setminus B_{1/\delta}$, $B_{2\delta}$, and $\mathcal{C}_{\delta,\sqrt{\varepsilon}}$, it is sufficient to estimate integrals on these three regions. Below we carry this out for a single term only, as the rest may be verified in a similar manner. Consider
\begin{align}
\begin{split}
&\int_{\mathbb{R}^{3}}
e^{-18U-18h_{1}}\left((|\psi^{1}|+|\psi^{1}_{0}|)^{6}
(|\psi^{2}|+|\psi^{2}_{0}|)^{6}e^{-3V-3h_{2}}
+(|\psi^{1}|+|\psi^{1}_{0}|)^{12}e^{3V+3h_{2}}\right)
|\psi^{1}-\psi^{1}_{\delta,\varepsilon}|^{6}dx\\
\leq&C\left(\int_{\mathbb{R}^{3}\setminus B_{1/\delta}}
+\int_{\mathcal{C}_{\delta,\sqrt{\varepsilon}}}
+\int_{B_{2\delta}}\frac{e^{-18U}}{\rho^{9}}
\frac{\cos^{3}\theta}{\sin^{3}\theta}(|\psi^{1}|+|\psi^{1}_{0}|)^{6}
(|\psi^{2}|+|\psi^{2}_{0}|)^{6}|\psi^{1}-\psi^{1}_{\delta,\varepsilon}|^{6}\right)\\
&+C\left(\int_{\mathbb{R}^{3}\setminus B_{1/\delta}}
+\int_{\mathcal{C}_{\delta,\sqrt{\varepsilon}}}
+\int_{B_{2\delta}}\frac{e^{-18U}}{\rho^{9}}
\frac{\sin^{3}\theta}{\cos^{3}\theta}(|\psi^{2}|+|\psi^{2}_{0}|)^{12}
|\psi^{1}-\psi^{1}_{\delta,\varepsilon}|^{6}\right).
\end{split}
\end{align}
Observe that \eqref{FALLOFF4}, \eqref{FALLOFF8.0}, \eqref{002.2},  \eqref{002.3}, \eqref{002.2'}, and \eqref{002.3'} imply that as $\delta\rightarrow 0$ we have
\begin{equation}
\int_{\mathbb{R}^{3}\setminus B_{1/\delta}}\underbrace{\frac{e^{-18U}}{\rho^{9}}}_{O(\rho^{-9})}
\left(\frac{\cos^{3}\theta}{\sin^{3}\theta}
\underbrace{(|\psi^{1}|+|\psi^{1}_{0}|)^{6}}_{\sin^{9}\theta O(r^{-6\kappa})}
\underbrace{(|\psi^{2}|+|\psi^{2}_{0}|)^{6}}_{\cos^{9}\theta O(r^{-6\kappa})}
+\frac{\sin^{3}\theta}{\cos^{3}\theta}
\underbrace{(|\psi^{2}|+|\psi^{2}_{0}|)^{12}}_{\cos^{18}\theta O(r^{-12\kappa})}\right)
\underbrace{|\psi^{1}-\psi^{1}_{\delta,\varepsilon}|^{6}}_{\sin^{9}\theta
O(r^{-6\kappa})}\rightarrow 0,
\end{equation}
and
\begin{align}
\begin{split}
&\int_{B_{2\delta}}\underbrace{\frac{e^{-18U}}{\rho^{9}}}_{
\begin{cases}
\rho^{-9}O(r^{36})\text{ }\text{ in the AF case}\\
\rho^{-9} O(r^{18})\text{ }\text{ in the AC case}
\end{cases}}
\frac{\cos^{3}\theta}{\sin^{3}\theta}
\underbrace{(|\psi^{1}|+|\psi^{1}_{0}|)^{6}
(|\psi^{2}|+|\psi^{2}_{0}|)^{6}
|\psi^{1}-\psi^{1}_{\delta,\varepsilon}|^{6}}_{\begin{cases}
\sin^{18}\theta\cos^{9}\theta O(r^{-18})\text{ }\text{ in the AF case}\\
\sin^{18}\theta\cos^{9}\theta O(1)\text{ }\text{ in the AC case}
\end{cases}}\\
+&\int_{B_{2\delta}}\underbrace{\frac{e^{-18U}}{\rho^{9}}}_{
\begin{cases}
\rho^{-9}O(r^{36})\text{ }\text{ in the AF case}\\
\rho^{-9} O(r^{18})\text{ }\text{ in the AC case}
\end{cases}}
\frac{\sin^{3}\theta}{\cos^{3}\theta}
\underbrace{(|\psi^{2}|+|\psi^{2}_{0}|)^{12}
|\psi^{1}-\psi^{1}_{\delta,\varepsilon}|^{6}}_{\begin{cases}
\cos^{18}\theta\sin^{9}\theta O(r^{-18})\text{ }\text{ in the AF case}\\
\cos^{18}\theta\sin^{9}\theta O(1)\text{ }\text{ in the AC case}
\end{cases}}\rightarrow 0.
\end{split}
\end{align}
Furthermore \eqref{FALLOFF8.1}, \eqref{002.2''}, and \eqref{002.3''} yield
\begin{equation}
\int_{\mathcal{C}_{\delta,\sqrt{\varepsilon}}}\underbrace{\frac{e^{-18U}}{\rho^{9}}}_{O(\rho^{-9})}
\left(\frac{\cos^{3}\theta}{\sin^{3}\theta}
\underbrace{(|\psi^{1}|+|\psi^{1}_{0}|)^{6}}_{\sin^{9}\theta O(1)}
\underbrace{(|\psi^{2}|+|\psi^{2}_{0}|)^{6}}_{\cos^{9}\theta O(1)}
+\frac{\sin^{3}\theta}{\cos^{3}\theta}
\underbrace{(|\psi^{2}|+|\psi^{2}_{0}|)^{12}}_{\cos^{18}\theta O(1)}\right)
\underbrace{|\psi^{1}-\psi^{1}_{\delta,\varepsilon}|^{6}}_{\sin^{9}\theta
O(1)}\rightarrow 0
\end{equation}
as $\varepsilon\rightarrow 0$. As mentioned above, all remaining terms of \eqref{aaa} may be treated similarly;
it follows that \eqref{LKJ} holds.

\section{Proof of the Main Result: Theorem \ref{MainTheorem}}

We may assume that $Q=|Q|$ by replacing $E$ with $-E$ if necessary. If such a replacement is made, then a change in orientation is also required so that $\star\rightarrow -\star$ which preserves the constraint equation \eqref{Const3}; the change of orientation does not affect the sign of $Q$ since both $\star$ and the integral change signs. Having chosen an orientation to fix the sign of the charge, and in particular the direction of rotation for the Killing fields $\eta_{(l)}$, we are not able to simultaneously guarantee the signs of the angular momenta $\mathcal{J}_{l}$. In this case the proof below will yield \eqref{maininequality}. Alternatively, we may assume without loss of generality that $\mathcal{J}_{l}=|\mathcal{J}_{l}|$, $l=1,2$ by replacing $\eta_{(l)}$ with $-\eta_{(l)}$ if necessary, but cannot simultaneously guarantee the sign of $Q$. In this situation, the proof presented below will yield \eqref{maininequality1}.

If $ab+q=0$, then we may take a perturbation of the initial data to achieve $ab+q\neq 0$ while at the same time preserving all the hypotheses of the theorem. Thus, establishing the inequality \eqref{maininequality} and \eqref{maininequality1} under the condition $ab+q\neq 0$, as is done below, also yields the desired result when $ab+q=0$ by letting the perturbation converge to zero.
Let us now assume that $ab+q\neq 0$, so that there is an extreme charged Myers-Perry black hole solution yielding the harmonic map $\tilde{\Psi}_{0}$ which satisfies the asymptotics \eqref{FALLOFF4}-\eqref{FALLOFF10.1'}. As mentioned in the introduction, the nonvanishing of $ab+q$ is required in order to have a proper black hole arising from the extreme charged Myers-Perry family; if $ab+q=0$, then the corresponding extreme charged Myers-Perry solution has a naked singularity, and such data do not satisfy the appropriate asymptotics.
In Appendix \ref{app1} it is shown that the hypotheses concerning the asymptotics of the initial data $(M,g,k,E,B)$, imply that $(U,V,W,\zeta^{1},\zeta^{2},\chi,\psi^{1},\psi^{2})$ satisfy the asymptotics
\eqref{FALLOFF1}-\eqref{FALLOFF4.1'}.
Therefore Theorem \ref{minimum} may be applied,
and the mass-angular momentum-charge inequality \eqref{maininequality} follows from \eqref{POI}, \eqref{53}, and the fact that
\begin{equation}
\mathcal{M}(\Psi_{0})=\frac{27\pi}{8}\frac{\left(\mathcal{J}_1+\mathcal{J}_2\right)^2}
{\left(2\mathcal{M}(\Psi_{0})+\sqrt{3}Q\right)^2}
+\sqrt{3}Q.
\end{equation}

Next consider the situation when equality is achieved in \eqref{maininequality} or \eqref{maininequality1}, again with $ab+q\neq 0$.  Then by \eqref{POI} and \eqref{53} we have
\begin{equation}\label{yhn}
\mu_{SG}=0,\text{ }\text{ }\text{ }\text{ }\text{
}\text{ }
A^{i}_{\rho,z}=A^{i}_{z,\rho},\text{ }\text{ }\text{ }\text{ }\text{ }i=1,2,
\text{ }\text{ }\text{ }\text{ }\text{ }E(e_{3})=E(e_{4})=0,
\end{equation}
\begin{equation}\label{ujm}
B(e_{i},e_{j})=B(e_{3},e_{3})=B(e_{3},e_{4})=B(e_{4},e_{4})=0,\text{ }\text{ }\text{ }\text{ }\text{ }i,j\neq 3, 4,
\end{equation}
\begin{equation}\label{ol.}
k(e_{i},e_{j})=k(e_{3},e_{3})=k(e_{3},e_{4})=k(e_{4},e_{4})=0,\text{ }\text{ }\text{ }\text{ }\text{ }i,j\neq 3, 4,
\end{equation}
and
\begin{equation}\label{ghj}
(U,V,W,\zeta^{1},\zeta^{2},\chi,\psi^{1},\psi^{2})
=(U_0,V_0,W_0,\zeta^{1}_0,\zeta^{2}_0,\chi_{0},\psi^{1}_{0},\psi^{2}_{0}).
\end{equation}
Observe that \eqref{sc1} and \eqref{yhn}-\eqref{ghj} produce
\begin{align}
\begin{split}
R=&16\pi\mu_{SG}+|k|^{2}+\frac{1}{2}|E|^{2}+\frac{1}{4}|B|^{2}\\
=&
\frac{e^{-8U_{0}-2\alpha+2\log r}}{2\rho^{2}}\lambda^{ij}_{0}\Theta_{0}^{i}\cdot\Theta_{0}^{j}
+\frac{e^{-6U_{0}-2\alpha+2\log r}}{2\rho^{2}}|\Upsilon_{0}|^{2}
+\frac{e^{-4U_{0}-2\alpha+2\log r}}{2}\lambda_{0}^{ij}\nabla\psi_{0}^{i}\cdot\nabla\psi_{0}^{j}\\
=&e^{2(\alpha_{0}-\alpha)}R_{0},
\end{split}
\end{align}
where $R_{0}$ and $\alpha_{0}$ are associated with the extreme charged Myers-Perry solution. Furthermore, from the basic formula for the scalar curvature of Brill data \eqref{scalarcurvature}, along with \eqref{yhn} and \eqref{ghj} we find
\begin{align}
\begin{split}
e^{2U+2\alpha-2\log r}R
=&-6\Delta U_{0}-2\Delta_{\rho,z}\alpha-6|\nabla U_{0}|^2
+\frac{\det\nabla\lambda_{0}}{2\rho^{2}}\\
=&e^{2U_{0}+2\alpha_{0}-2\log r}R_{0}+2\Delta_{\rho,z}(\alpha_{0}-\alpha).
\end{split}
\end{align}
This shows that $\Delta_{\rho,z}(\alpha_{0}-\alpha)=0$. Moreover $(\alpha_{0}-\alpha)|_{\Gamma}=0$ since there are no conical singularities on the axes \eqref{ConeCondition},
and $(\alpha_{0}-\alpha)\rightarrow 0$ as $r\rightarrow\infty$. Hence, by the maximum principle $\alpha=\alpha_{0}$.

In order to establish that $(M,g)$ is isometric to the canonical slice of the extreme charged Myers-Perry black hole, note that
according to \eqref{yhn} the 1-forms $A_{\rho}^{i}d\rho+A_{z}^{i}dz$, $i=1,2$ are closed, thereby yielding potentials satisfying
$\partial_{\rho}f^{i}=A_{\rho}^{i}$ and $\partial_{z}f^{i}=A_{z}^{i}$, $i=1,2$. Now change  coordinates by
$\widetilde{\phi}^{i}=\phi^{i}+f^{i}(\rho,z)$ so that the metric becomes
\begin{equation}
g=\frac{e^{2U_0+2\alpha_{0}}}{2\sqrt{\rho^{2}+z^{2}}}(d\rho^{2}+dz^{2})
+e^{2U_0}(\lambda_{0})_{ij}d\widetilde{\phi}^{i}d\widetilde{\phi}^{j},
\end{equation}
and $g\cong g_0$. Finally \eqref{Bcomponents}, \eqref{Ecomponents}, \eqref{kcomponents}, \eqref{yhn}-\eqref{ghj}, and $\alpha=\alpha_{0}$ show that $k$, $E$, and $B$ agree with their counterparts in the canonical slice of the extreme charged Myers-Perry spacetime, and in particular the non-electromagnetic linear momentum vanishes $J_{SG}=0$.
\hfill\qedsymbol

\appendix
\numberwithin{equation}{section}

\section{Relations Between Asymptotics}
\label{app1}

In order to apply Theorem \ref{minimum} in the proof of Theorem \ref{MainTheorem}, it is necessary to show that the asymptotics \eqref{Asym1}-\eqref{Asym12} of the initial data $(M,g,k,E,B)$, imply that the resulting harmonic map data $(U,V,W,\zeta^{1},\zeta^{2},\chi,\psi^{1},\psi^{2})$ satisfy the asymptotics \eqref{FALLOFF1}-\eqref{FALLOFF4.1'}. The purpose of this appendix is to establish this fact.
Note that the asymptotics of $U$ are given directly, and those of $V$ and $W$ come from \eqref{variablesvw}. Moreover, the asymptotics of the potentials arise from those of $|E|_{g}$, $|B|_{g}$, and $|k|_{g}$, in the following way. First observe that \eqref{1.11} implies
\begin{equation}
\sum_{i=1,2}|\nabla\psi^{i}|
=\sum_{i=1,2}(|\partial_{\rho}\psi^i|^{2}+|\partial_{z}\psi^i|^{2})^{1/2}
\leq Cr^{-1}\rho e^{2U+\alpha}\sum_{i=1,2}\left(|B(\theta^1,\theta^{i+2})|
+|B(\theta^{2},\theta^{i+2})|\right),
\end{equation}
and
\begin{equation}
\sum_{l=1,2}\lambda_{ij}B(\theta^l,\theta^{i+2})B(\theta^l,\theta^{j+2})\leq|B|_{g}^{2},
\end{equation}
so that the asymptotics of $\psi^1$ and $\psi^2$ may be obtained from those of $|B|_{g}$. Using this we find the asymptotics for $\chi$ in terms of those for $|E|_{g}$, since \eqref{Ecomponents} yields
\begin{align}
\begin{split}
|\nabla\chi|=&(|\partial_{\rho}\chi|^{2}+|\partial_{z}\chi|^{2})^{1/2}\\
\leq &C \left[r^{-1}\rho e^{3U+\alpha}(|E(e_{1})|+|E(e_{2})|)+
(|\psi^1|+|\psi^2|)(|\nabla\psi^1|+|\nabla\psi^2|)\right],
\end{split}
\end{align}
and in addition
\begin{equation}
|E(e_{1})|^2+|E(e_{2})|^2\leq|E|_{g}^{2}.
\end{equation}
Lastly, with \eqref{kcomponents} the asymptotics for the potentials $\zeta^1$ and $\zeta^2$ may be derived from
\begin{equation}
|\nabla \zeta^{i}|=\left(|\partial_{\rho}\zeta^i|^{2}+|\partial_{z}\zeta^i|^{2}\right)^{\frac{1}{2}}
\leq Cr^{-1}\rho e^{4U+\alpha}\left(|k(e_{1},e_{i+2})|+|k(e_{2},e_{i+2})|\right),
\end{equation}
and
\begin{equation}
\sum_{l=1,2}\lambda^{ij}k(e_{l},e_{i+2})k(e_{l},e_{j+2})\leq |k|_{g}^{2}.
\end{equation}
In conclusion, the asymptotics for Brill data produce the following asymptotics for the corresponding harmonic map data.
As $r\rightarrow\infty$ the following decay occurs
\begin{equation}
U,V=O(r^{-1-\kappa}),\text{ }\text{ }\text{ }W=\rho O(r^{-5-\kappa}),\text{ }\text{ }\text{ }\psi^{1}=\sin\theta O(r^{-\kappa}),\text{ }\text{ }\text{ }\psi^{2}=\cos\theta O(r^{-\kappa}),
\end{equation}
\begin{equation}
|\nabla U|=O(r^{-3-\kappa}),\text{ }\text{ }\text{ }\text{ }|\nabla V|=O(r^{-3-\kappa}),\text{ }\text{ }\text{ }\text{ }|\nabla W|=O(r^{-5-\kappa}),
\end{equation}
\begin{equation}
|\nabla\psi^{1}|=\sin\theta O(r^{-2-\kappa}),\text{ }\text{ }\text{ }\text{ }|\nabla\psi^{2}|=\cos\theta O(r^{-2-\kappa}),\text{ }\text{ }\text{ }
\end{equation}
\begin{equation}
|\nabla\chi|=\rho O(r^{-3-\kappa}),\text{ }\text{ }\text{ }\text{ }\text{ }\text{ }\text{ }|\nabla \zeta^{1}|=\rho\sin\theta O(r^{-2-\kappa}),\text{ }\text{ }\text{ }\text{ }\text{ }\text{ }\text{ }|\nabla \zeta^{2}|=\rho\cos\theta O(r^{-2-\kappa}).
\end{equation}
Consider now the nondesignated end, in which the asymptotics are broken up into two cases. In the asymptotically flat case, as $r\rightarrow 0$, we have
\begin{equation}
U=-2\log r+O(1),\text{ }\text{ }\text{ } W=\rho O(r^{-\frac{1}{2}+\kappa}),\text{ }\text{ }\text{ }\psi^{1}=\sin\theta O(r^{\kappa}),\text{ }\text{ }\text{ }
\psi^{2}=\cos\theta O(r^{\kappa}),
\end{equation}
\begin{equation}
|\nabla U|=O(r^{-2}),\text{ }\text{ }\text{ }\text{ }\text{ }r^{-2}|V|+|\nabla V|=O(r^{-1+\kappa}),\text{ }\text{ }\text{ }\text{ }\text{ }|\nabla W|= O(r^{-1/2+\kappa}),
\end{equation}
\begin{equation}
|\nabla\psi^{1}|=\sin\theta O(r^{-2+\kappa}),\text{ }\text{ }\text{ }\text{ }|\nabla\psi^{2}|=\cos\theta O(r^{-2+\kappa}),\text{ }\text{ }\text{ }\text{ }
\end{equation}
\begin{equation}
|\nabla\chi|=\rho O(r^{-5+\kappa}),\text{ }\text{ }\text{ }\text{ }\text{ }\text{ }\text{ }|\nabla \zeta^{1}|=\rho\sin\theta O(r^{-6+\kappa}),\text{ }\text{ }\text{ }\text{ }\text{ }\text{ }\text{ }|\nabla \zeta^{2}|=\rho\cos\theta O(r^{-6+\kappa}).
\end{equation}
In the asymptotically cylindrical case, as $r\rightarrow 0$, we have
\begin{equation}
U=-\log r+O(1),\text{ }\text{ }\text{ } W=\rho O(r^{-2}),\text{ }\text{ }\text{ }\psi^{1}=\sin\theta O(r^{2+\kappa}),\text{ }\text{ }\text{ }
\psi^{2}=\cos\theta O(r^{2+\kappa}),
\end{equation}
\begin{equation}
|\nabla U|=O(r^{-2}),\text{ }\text{ }\text{ }\text{ }\text{ }r^{-2}|V|+|\nabla V|=O(r^{-1+\kappa}),\text{ }\text{ }\text{ }\text{ }\text{ }|\nabla W|= O(r^{-2}),
\end{equation}
\begin{equation}
|\nabla\psi^{1}|=\sin\theta O(r^{\kappa}),\text{ }\text{ }\text{ }\text{ }|\nabla\psi^{2}|=\cos\theta O(r^{\kappa}),\text{ }\text{ }\text{ }\text{ }
\end{equation}
\begin{equation}
|\nabla\chi|=\rho O(r^{-2+\kappa}),\text{ }\text{ }\text{ }\text{ }\text{ }\text{ }\text{ }
|\nabla \zeta^{1}|=\rho\sin\theta O(r^{-2+\kappa}),\text{ }\text{ }\text{ }\text{ }\text{ }\text{ }\text{ }|\nabla \zeta^{2}|=\rho\cos\theta O(r^{-2+\kappa}).
\end{equation}
Moreover the asymptotics near the axis, that is, as $\rho\rightarrow 0$ with $\delta\leq r\leq \frac{2}{\delta}$, are given by
\begin{equation}
U,V=O(1),\text{ }\text{ }\text{ } W=O(\rho),\text{ }\text{ }\text{ }\psi^{1}=\sin\theta O(\rho),\text{ }\text{ }\text{ }
\psi^{2}=\cos\theta O(\rho),
\end{equation}
\begin{equation}
|\nabla U|=O(1),\text{ }\text{ }\text{ }\text{ }\text{ }|\nabla V|=O(1),\text{ }\text{ }\text{ }\text{ }\text{ }|\nabla W|=O(1),
\end{equation}
\begin{equation}
|\nabla\psi^{1}|=\sin\theta O(1),\text{ }\text{ }\text{ }\text{ }|\nabla\psi^{2}|=\cos\theta O(1),\text{ }\text{ }\text{ }\text{ }
\end{equation}
\begin{equation}
|\nabla\chi|=O(\rho),\text{ }\text{ }\text{ }\text{ }\text{ }\text{ }\text{ }|\nabla \zeta^{1}|=\sin\theta O(\rho),\text{ }\text{ }\text{ }\text{ }\text{ }\text{ }\text{ }|\nabla \zeta^{2}|=\cos\theta O(\rho).
\end{equation}

\section{The Charged Myers-Perry Black Hole in 5D Minimal Supergravity}
\label{app2}

Consider 5-dimensional minimal supergravity with action \eqref{sugraaction}. The charged Myers-Perry solution \cite{Chong:2005hr} may be interpreted as a natural generalization
of the Kerr-Newman black hole to 5 dimensions.
In Boyer-Lindquist coordinates the charged Myers-Perry metric takes the form
\begin{align}
\begin{split}
& -dt^2-\frac{2q}{\Sigma}\left(dt-a\sin^2\theta d\phi^1-b\cos^2\theta d\phi^2\right)
\left(b\sin^2\theta d\phi^1+a\cos^2\theta d\phi^2\right) \\
&+\frac{\mathfrak{m}\Sigma-q^2}{\Sigma^2}\left(dt-a\sin^2\theta d\phi^1-b\cos^2\theta d\phi^2\right)^2+\frac{\tilde{r}^2\Sigma}{\Delta}d\tilde{r}^2
+ \Sigma d\theta^2\\
&+\left(\tilde{r}^2+a^2\right)\sin^2\theta (d\phi^1)^2
+\left(\tilde{r}^2+b^2\right)\cos^2\theta (d\phi^2)^2,
\end{split}
\end{align}
where
\begin{equation}
\Sigma=\tilde{r}^2+b^2\sin^2\theta+a^2\cos^2\theta,\qquad \Delta=\left(\tilde{r}^2+a^2\right)\left(\tilde{r}^2+b^2\right)+q^2+2abq-\mathfrak{m} \tilde{r}^2.
\end{equation}
The gauge field which defines the field strength $F=dA$ is given by
\begin{equation}
A=\frac{\sqrt{3}q}{\Sigma}\left(dt-a\sin^2\theta d\phi^1-b\cos^2\theta d\phi^2\right).
\end{equation}
Notice that there are four parameters $(\mathfrak{m},a,b,q)$ that characterize each solution. These parameters represent the mass, angular momenta, and charge through the formulae
\begin{gather}
m=\frac{3}{8}\pi\mathfrak{m},\qquad \mathcal{J}_{1}=\frac{2am+\sqrt{3}Qb}{3},\qquad \mathcal{J}_{2}=\frac{2bm+\sqrt{3}Qa}{3},\qquad Q=\frac{\sqrt{3}\pi q}{4}.
\end{gather}
The black hole is extreme if $\mathfrak{m}=2q+(a+b)^2$, and in this case
\begin{equation}
\left(2m+\sqrt{3}Q\right)^2\left(m-\sqrt{3}Q\right)=\frac{27\pi}{8} \left( \mathcal{J}_{1}+ \mathcal{J}_{2}\right)^2:=\mathcal{J},
\end{equation}
from which one may solve for the mass to find
\begin{align}
\begin{split}
m=&\frac{3Q^2}{2}\left(\mathcal{J} + 3\sqrt{3}Q^3 + \sqrt{\mathcal{J}^2 + 6\sqrt{3} \mathcal{J} Q^3}\right)^{-1/3}\\
 &+ \frac{1}{2}\left(\mathcal{J} + 3\sqrt{3} Q^3 + \sqrt{\mathcal{J}^2 + 6\sqrt{3}\mathcal{J} Q^3}\right)^{1/3}.
\end{split}
\end{align}
Horizons are located at the roots of $\Delta$ and are given by
\begin{equation}
\tilde{r}_{\pm}=\pm\sqrt{\frac{\mathfrak{m}-a^2-b^2
		 +\sqrt{\left[\mathfrak{m}+2q-\left(a-b\right)^2\right]\left[\mathfrak{m}-2q-\left(a+b\right)^2\right]}}{2}},
\end{equation}
whereas spacetime singularities for nonvanishing $a$ and $b$ with $|a|\neq |b|$ are found at the roots of $\Sigma$. Moreover this spacetime exhibits an orthogonally transitive isometry group $\mathbb{R}\times U(1)^2$, where $U(1)^2$ represents the rotational symmetry generated by $\partial_{\phi^i}$, $i=1,2$, and $\mathbb{R}$ gives the time translation symmetry.

Consider now the metric on a constant time slice. We will focus on the exterior region $\tilde{r}>\tilde{r}_{+}$, with the remaining variables taking values $\theta\in (0,\pi/2)$ and $\phi^{i}\in (0,2\pi)$, $i=1,2$.
In this domain a new radial coordinate $r\in(0,\infty)$ may be defined by
\begin{equation}\label{fjgh}
\tilde{r}^2=r^2+\frac{1}{2}\left(\mathfrak{m}-a^2-b^2\right)
+\frac{\left[\mathfrak{m}+2q-\left(a-b\right)^2\right]\left[\mathfrak{m}-2q-\left(a+b\right)^2\right]}{16r^2},\text{ }\text{ }\text{ }\text{ }\text{ }\text{ }\mathfrak{m}\neq 2q+(a+b)^2,
\end{equation}
\begin{equation}
\tilde{r}^2=r^2+ab+q,\text{ }\text{ }\text{ }\text{ }\text{ }\text{ }\mathfrak{m}=2q+(a+b)^2.
\end{equation}
It turns out that the right-hand side of \eqref{fjgh} has a critical point at the horizon, and that on either side of this surface are isometric copies of the outer region. The purpose of the new coordinates $(r,\theta,\phi^1,\phi^2)$ is to put the spatial metric in Brill form
\begin{equation}
g=\frac{\Sigma}{r^2}\left(dr^2
+r^2 d\theta^2\right)+\Lambda_{ij}d\phi^i d\phi^j,
\end{equation}
where
\begin{gather}
\Lambda_{11}=\frac{a\left[a\left(\mathfrak{m}\Sigma-q^2\right)+2bq\Sigma\right]}{\Sigma^2}\sin^4\theta+(\tilde{r}^2+a^2)\sin^2\theta,\\
\Lambda_{12}=\frac{ab\left(\mathfrak{m}\Sigma-q^2\right)+(a^2+b^2)q\Sigma}{\Sigma^2}\sin^2\theta\cos^2\theta,\\
\Lambda_{22}=\frac{b\left[b\left(\mathfrak{m}\Sigma-q^2\right)+2aq\Sigma\right]}{\Sigma^2}\cos^4\theta+(\tilde{r}^2+b^2)\cos^2\theta.
\end{gather}
The cylindrical version of Brill coordinates arises from the typical transformation $\rho=\frac{1}{2}r^{2}\sin(2\theta)$, $z=\frac{1}{2}r^{2}\cos(2\theta)$, and in this setting the metric becomes
\begin{equation}
g=\frac{e^{2U+2\alpha}}{2\sqrt{\rho^2+z^2}}(d\rho^2+dz^2)
+e^{2U}\lambda_{ij}d\phi^i d\phi^j,
\end{equation}
with
\begin{equation}\label{htildeMP}
e^{2U}=\frac{\sqrt{\det\Lambda}}{\rho}, \qquad e^{2\alpha}=\frac{\rho \Sigma}{r^2\sqrt{\det\Lambda}}, \qquad \lambda_{ij}=\frac{\rho}{\sqrt{\det\Lambda}}\Lambda_{ij}.
\end{equation}

In order to extract the electric field $E=\iota_n F$ and magnetic 2-form field $B=\iota_n \star_{5} F$, note that if $h=\sqrt{-g^{tt}}$ then the unit normal to a constant time slice is $n = -h^{-1}dt$, and also
\begin{equation}
\sqrt{\det g} = \frac{\tilde{r}h\Sigma \sin 2\theta}{2}.
\end{equation}
Thus, we find that
\begin{equation}
E = \frac{\sqrt{3} q}{h \Sigma^2}\left[2\frac{\tilde{r} (\Delta + m \tilde{r}^2 - abq - q^2)}{ \Delta} d \tilde{r} - (a^2-b^2) \sin 2\theta d\theta\right].
\end{equation}
The explicit expression for the magnetic 2-form $B$ is fairly complicated, but the Hodge dual is more natural and is given by
\begin{align}
\begin{split}
\star B =& \frac{2 \sqrt{3} a q r \sin^2\theta}{\Sigma^2} dr \wedge d\phi^1  + \frac{2 \sqrt{3} b q r \cos^2\theta}{\Sigma^2} dr \wedge d\phi^2 \\
&- \frac{\sqrt{3} a q (a^2+r^2) \sin 2\theta}{\Sigma^2} d\theta \wedge d\phi^1  +\frac{\sqrt{3} b q (b^2+r^2) \sin 2\theta}{\Sigma^2} d\theta \wedge d\phi^2.
\end{split}
\end{align}

The potentials may be derived as follows. Write $A=A_tdt+A_id\phi^i$, then
\begin{equation}
d\psi^{i}=\iota_{\eta_{(i)}}F=\iota_{\eta_{(i)}}dA
=\mathfrak{L}_{\eta_{(i)}}A-d\iota_{\eta_{(i)}}A =-dA_{i}
\end{equation}
and so
\begin{equation}
\psi^{1}=-A_{1}= \frac{\sqrt{3} a q \sin^2\theta}{\Sigma}, \text{ }\text{ }\text{ }\text{ }\text{ }\text{ }\text{ }\text{ } \psi^{2} =-A_{2}= \frac{\sqrt{3} b q \cos^2\theta}{\Sigma}.
\end{equation}
Further computations yield the remaining potentials
\begin{equation}
\chi=-\sqrt{3}q\sin^2\theta+\frac{\sqrt{3}q(a^2-b^2)\sin^2\theta\cos^2\theta}{\Sigma},
\end{equation}
and
\begin{align}
\begin{split}
\zeta^1=& \frac{(a \mathfrak{m} + b q)(\cos 4\theta - 4 \cos 2\theta)}{8} - \frac{2(a^2-b^2)(2 a q^2 + 2(a \mathfrak{m} + bq)\Sigma)\cos^2 \theta \sin^4\theta}{4\Sigma^2}, \\
\zeta^2=&  -\frac{(b \mathfrak{m} + a q)(\cos 4\theta + 4 \cos 2\theta)}{8} - \frac{2(a^2-b^2)(2 b q^2 + 2(b \mathfrak{m} + aq)\Sigma)\cos^4 \theta \sin^2\theta}{4\Sigma^2}.
\end{split}
\end{align}

From the explicit expressions above, we may calculate the asymptotics in the non-extreme case. As $r\rightarrow\infty$ the following decay occurs
\begin{equation}
U,V=O(r^{-2}),\text{ }\text{ }\text{ }W=\rho O(r^{-6}),\text{ }\text{ }\text{ }\psi^{1}=\sin^{2}\theta O(r^{-4}),\text{ }\text{ }\text{ }\psi^{2}=\cos^{2}\theta O(r^{-4}),
\end{equation}
\begin{equation}
|\nabla U|=O(r^{-4}),\text{ }\text{ }\text{ }\text{ }|\nabla V|=O(r^{-4}),\text{ }\text{ }\text{ }\text{ }|\nabla W|=O(r^{-6}),
\end{equation}
\begin{equation}
|\nabla\psi^{1}|=\sin\theta O(r^{-5}),\text{ }\text{ }\text{ }\text{ }|\nabla\psi^{2}|=\cos\theta O(r^{-5}),\text{ }\text{ }\text{ }
\end{equation}
\begin{equation}
|\nabla\chi|=\rho O(r^{-4}),\text{ }\text{ }\text{ }\text{ }\text{ }\text{ }\text{ }|\nabla \zeta^{1}|=\rho\sin^{2}\theta O(r^{-4}),\text{ }\text{ }\text{ }\text{ }\text{ }\text{ }\text{ }|\nabla \zeta^{2}|=\rho\cos^{2}\theta O(r^{-4}).
\end{equation}
In the nondesignated asymptotically flat end case, as $r\rightarrow 0$, we have
\begin{equation}
U=-2\log r+O(1),\text{ }\text{ }\text{ } W=\rho O(r^{2}),\text{ }\text{ }\text{ }\psi^{1}=\sin^{2}\theta O(r^{2}),\text{ }\text{ }\text{ }
\psi^{2}=\cos^{2}\theta O(r^{2}),
\end{equation}
\begin{equation}
|\nabla U|=O(r^{-2}),\text{ }\text{ }\text{ }\text{ }\text{ }r^{-2}|V|+|\nabla V|=O(1),\text{ }\text{ }\text{ }\text{ }\text{ }|\nabla W|= O(r^{2}),
\end{equation}
\begin{equation}
|\nabla\psi^{1}|=\sin\theta O(1),\text{ }\text{ }\text{ }\text{ }|\nabla\psi^{2}|=\cos\theta O(1),\text{ }\text{ }\text{ }\text{ }
\end{equation}
\begin{equation}
|\nabla\chi|=\rho O(r^{-4}),\text{ }\text{ }\text{ }\text{ }\text{ }\text{ }\text{ }|\nabla \zeta^{1}|=\rho\sin^{2}\theta O(r^{-4}),\text{ }\text{ }\text{ }\text{ }\text{ }\text{ }\text{ }|\nabla \zeta^{2}|=\rho\cos^{2}\theta O(r^{-4}).
\end{equation}
Furthermore the asymptotics near the axis, that is, as $\rho\rightarrow 0$ with $\delta\leq r\leq \frac{2}{\delta}$, are given by
\begin{equation}
U,V=O(1),\text{ }\text{ }\text{ } W=O(\rho),\text{ }\text{ }\text{ }\psi^{1}=\sin^{2}\theta O(1),\text{ }\text{ }\text{ }
\psi^{2}=\cos^{2}\theta O(1),
\end{equation}
\begin{equation}
|\nabla U|=O(1),\text{ }\text{ }\text{ }\text{ }\text{ }|\nabla V|=O(1),\text{ }\text{ }\text{ }\text{ }\text{ }|\nabla W|=O(1),
\end{equation}
\begin{equation}
|\nabla\psi^{1}|=\sin\theta O(1),\text{ }\text{ }\text{ }\text{ }|\nabla\psi^{2}|=\cos\theta O(1),\text{ }\text{ }\text{ }\text{ }
\end{equation}
\begin{equation}
|\nabla\chi|=O(\rho),\text{ }\text{ }\text{ }\text{ }\text{ }\text{ }\text{ }|\nabla \zeta^{1}|=\sin^{2}\theta O(\rho),\text{ }\text{ }\text{ }\text{ }\text{ }\text{ }\text{ }|\nabla \zeta^{2}|=\cos^{2}\theta O(\rho).
\end{equation}
Asymptotics for the extreme charged Myers-Perry solution are recorded in
\eqref{FALLOFF4}-\eqref{FALLOFF10.1'}.

Lastly we state the Euler-Lagrange equations satisfied by the the extreme charged Myers-Perry harmonic map $\tilde{\Psi}_{0}:\mathbb{R}^{3}\setminus\Gamma
\rightarrow G_{2(2)}/SO(4)$, namely
\begin{align}
\begin{split}
	12\Delta u+\frac{3e^{-6u-v}}{\cosh w}|\Theta^1|^2+3e^{-6u+v}\cosh w|e^{-v}\tanh w\Theta^1-\Theta^2|^2+\frac{e^{-2u-v}}{\cosh w}|\nabla\psi^1|^2&\\
	+e^{-2u+v}\cosh w|e^{-v}\tanh w\nabla\psi^1-\nabla\psi^2|^2+2e^{-4u}\Upsilon^2&=0,\\
	2\operatorname{div}\left(\cosh^2w\nabla v\right)+e^{-6u}\cosh w\left\{e^{-v}|\Theta^1|^2-e^{v}|\Theta^2|^2\right\}&\\
	+e^{-2u}\cosh w\left\{e^{-v}|\nabla\psi^1|^2-e^{v}|\nabla\psi^2|^2\right\}&=0,\\
	-2\Delta w+\sinh 2w|\nabla v|^2-2e^{-2u}\cosh w\left\{\delta_3(\nabla\psi^1,\nabla\psi^2)+e^{-4u}
\delta_3(\Theta^1,\Theta^2)\right\}&\\
+e^{-6u}\sinh w\left\{e^{-v}|\Theta^1|^2+e^{v}|\Theta^2|^2\right\}+e^{-2u}\sinh w\left\{e^{-v}|\nabla\psi^1|^2+e^{v}|\nabla\psi^2|^2\right\}&=0,\\
	\operatorname{div}\left(e^{-6u-v}\cosh w\Theta^1-e^{-6u}\sinh w\Theta^2\right)&=0,\\
	\operatorname{div}\left(e^{-6u+v}\cosh w\Theta^2-e^{-6u}\sinh w\Theta^1\right)&=0,\\
	\operatorname{div}\left\{e^{-6u}\cosh w\left(e^{-v}\Theta^1\psi^1+e^{v}\Theta^2\psi^2\right)-e^{-6u}\sinh w\left(\Theta^2\psi^1+\Theta^1\psi^2\right)+e^{-4u}\Upsilon\right\}&=0,\\
	\operatorname{div}\left\{e^{-6u}\psi^1\psi^2\left(\sinh w\Theta^2-e^{-v}\cosh w\Theta^1\right)+(\psi^2)^2e^{-6u}\left(\sinh w\Theta^1-e^{v}\cosh w\Theta^2\right)\right.&\\
	\left. +3\sqrt{3}e^{-2u}\left(e^{-v}\cosh w\nabla\psi^1-\sinh w\nabla\psi^2\right)-3e^{-4u}\psi^2\Upsilon\right\}&\\ +e^{-6u}\delta_{3}\left(\left(3\sqrt{3}\nabla\chi+\left(2\psi^1\nabla\psi^2
-\psi^2\nabla\psi^1\right)\right),\left(\sinh w\Theta^2-e^{-v}\cosh w\Theta^1\right)\right)&\\
	+e^{-6u}\delta_{3}(\psi^2\nabla\psi^2,(\sinh w\Theta^1-e^{v}\cosh w\Theta^2))-3e^{-4u}\delta_3(\nabla\psi^2,\Upsilon)&=0,\\
	\operatorname{div}\left\{e^{-6u}\psi^1\psi^2\left(\sinh w\Theta^1-e^{v}\cosh w\Theta^2\right)+(\psi^1)^2e^{-6u}\left(e^{-v}\cosh w\Theta^1-\sinh w\Theta^2\right)\right.&\\
	\left. -3\sqrt{3}e^{-2u}\left(\sinh w\nabla\psi^1-e^{v}\cosh w\nabla\psi^2\right)+3e^{-4u}\psi^1\Upsilon\right\}&\\ +e^{-6u}\delta_{3}\left(\left(3\sqrt{3}\nabla\chi
+\left(\psi^1\nabla\psi^2-2\psi^2\nabla\psi^1\right)\right)
,\left(\sinh w\Theta^1-e^{v}\cosh w\Theta^2\right)\right)&\\
+e^{-6u}\delta_{3}(\psi^1\nabla\psi^1,(e^{-v}\cosh w\Theta^1-\sinh w\Theta^2))
+3e^{-4u}\delta_3(\nabla\psi^1,\Upsilon)&=0.
\end{split}
\end{align}

\section{Spacetime Construction of Potentials and Charges}
\label{app3}

In Section \ref{secpot} potentials for the electromagnetic field, as well as the angular momentum, were constructed from the initial data perspective. Here we show how the same potentials arise from the spacetime point of view. We do not assume stationarity, but do impose the spacetime field equations \eqref{FIELDeqns} and \eqref{FIELDeqns1}.
Consider the two magnetic 1-forms
$\mathcal{B}_{(i)} = -\iota_{\eta_{(i)}} F$, $i=1,2$. Since $dF =0$ and $\eta_{(i)}$ is a Killing field, Cartan's formula \eqref{cartanmagic} implies that these forms are closed
\begin{equation}
d\mathcal{B}_{(i)}=\iota_{\eta_{(i)}} dF-\mathfrak{L}_{\eta_{(i)}}F=0.
\end{equation}
Assuming that the spacetime is simply connected, there then exist magnetic potentials such that $\mathcal{B}_{(i)} = d\psi^{i}$, $i=1,2$. Now consider the electric 1-form
\begin{equation}
\mathcal{E} = -\iota_{\eta_{(1)}} \iota_{\eta_{(2)}} \star_{5} F,
\end{equation}
which vanishes on any axis of rotation where some linear combination of the $\eta_{(i)}$ vanish. This satisfies
\begin{equation}
d\mathcal{E} = -\iota_{\eta_{(1)}} \iota_{\eta_{(2)}}d\mathcal{S},
\end{equation}
where $S$ is defined from the Maxwell equations
\begin{equation}
d\left(\star_{5} F - \mathcal{S}\right) =0, \text{ }\text{ }\text{ }\text{ }\text{ }\text{ } \mathcal{S} = -\frac{1}{\sqrt{3}} A \wedge F.
\end{equation}
Therefore
\begin{equation}
d\mathcal{E} = d\sigma,\text{ }\text{ }\text{ }\text{ }\text{ }\text{ } \sigma = \frac{1}{\sqrt{3}}(\psi^1 d\psi^2 - \psi^2 d\psi^1),
\end{equation}
showing the existence of an electric potential with $d\chi = \mathcal{E} - \sigma$.
Finally we show how to construct the charged twist potentials, which encode angular momentum. Consider the twist 1-forms
\begin{equation}
\omega_{(i)} = \star_{5}(\eta_{(1)} \wedge \eta_{(2)} \wedge d\eta_{(i)}), \text{ }\text{ }\text{ }\text{ }\text{ }\text{ }\text{ }i=1,2.
\end{equation}
By the Frobenius theorem these 1-forms represent the obstruction to integrability of the distribution orthogonal to the 2-planes spanned by $\eta_{(1)}$ and $\eta_{(2)}$, and in addition they satisfy
\begin{equation}
d\omega_{(i)} = 2\star(\eta_{(1)} \wedge \eta_{(2)} \wedge \operatorname{Ric}(\eta_{(i)}))
\end{equation}
where $\operatorname{Ric}$ denotes the spacetime Ricci tensor. Then a computation \cite{Kunduri:2011zr} utilizing the Einstein equations shows that
\begin{equation}
d\omega_{(i)}=\mathcal{E} \wedge \mathcal{B}_{(i)} = d\left[\psi^{i} \left(d\chi + \frac{1}{3\sqrt{3}}(\psi^1 d\psi^2 - \psi^2 d\psi^1)\right)\right].
\end{equation}
It follows that twist potentials exist such that
\begin{equation}
d\zeta^i = \omega_i - \psi^i \left(d\chi + \frac{1}{3\sqrt{3}}(\psi^1 d\psi^2
 - \psi^2 d\psi^1)\right).
\end{equation}

In $3+1$-dimensional Einstein-Maxwell theory, one may integrate the closed 2-forms $F$ and $\star F$ over 2-cycles which enclose an asymptotically flat end, in order to obtain an appropriate definition of total magnetic and electric charge, respectively. In $4+1$-dimensional minimal supergravity, $F$ is still closed, however there are no 2-cycles which enclose an asymptotically flat end, and thus there is no natural notion of total magnetic charge. If the spacetime possesses nontrivial 2-cycles then one may integrate $F$ over these surfaces in order to obtain a notion of quasi-local magnetic charge. These surfaces are often referred to as bubbles supported by magnetic flux \cite{Marolf:2000cb}, but play no role in the current paper. On the other hand, there is a natural notion of total electric charge in 5-dimensional minimal supergravity, however it is not obtain by integrating $\star_5 F$, as this form is no longer closed. Rather, total electric charge is obtain by integrating
$\star_5 F + \frac{1}{\sqrt{3} }A \wedge F$ over a nontrivial 3-cycle, as this form is closed in light of the minimal supergravity equations.

\section{Conventions and Formulas for Forms}
\label{app4}

On an $n$-dimensional manifold let
\begin{equation}
\omega=\frac{1}{p!}\omega_{i_{1}\cdots i_{p}} e^{i_{1}}\wedge\cdots\wedge e^{i_{p}}
\end{equation}
be a $p$-form, where $\omega_{i_{1}\cdots i_{p}}$ is an antisymmetric covariant $p$-tensor and $e^i$, $i=1,\ldots,n$ form a basis for the cotangent space.
If $\alpha$ is a $q$-form then the wedge product of these two forms is a $p+q$-form given by
\begin{equation}
(\omega\wedge\alpha)_{i_{1}\cdots i_{p}j_{1}\cdots j_{q}}=\frac{(p+q)!}{p!q!}
\omega_{[i_{1}\cdots i_{p}}\alpha_{j_{1}\cdots j_{q}]},
\end{equation}
where for any tensor $T$ its antisymmetric part is
\begin{equation}
T_{[i_{1}\cdots i_{p}]}=\frac{1}{p!}\left(T_{i_{1}\cdots i_{p}}+\text{even permutations}
-\text{odd permutations}\right).
\end{equation}
The exterior derivative of $\omega$ produces a $p+1$-form
\begin{equation}
(d\omega)_{i_{1}\cdots i_{p+1}}=(p+1)\partial_{[i_{1}}\omega_{i_{2}\cdots i_{p+1}]},
\end{equation}
and the Hodge star operation is expressed in component form by
\begin{equation}
(\star \omega)_{j_{1}\cdots j_{n-p}}=\frac{1}{p!}
\epsilon_{j_{1}\cdots j_{n-p}}^{\phantom{j_{1}\cdots j_{n-p}}i_{1}\cdots i_{p}}
\omega_{i_{1}\cdots i_{p}}.
\end{equation}
For a metric with $t$ negative eigenvalues we have
\begin{equation}
\star\star\omega=(-1)^{p(n-p)+t}\omega,
\end{equation}
and
\begin{equation}
(\star d\star\omega)_{i_{1}\cdots i_{p-1}}=(-1)^{p(n-p+1)+t+1}\nabla^{l}\omega_{li_{1}\cdots i_{p-1}}.
\end{equation}
A useful formula for contracting volume forms is
\begin{equation}
\epsilon^{i_{1}\cdots i_{n-p} l_{1}\cdots l_{p}}\epsilon_{j_{1}\cdots j_{n-p} l_{1}\cdots l_{p}}
=(-1)^{t}p!(n-p)!\delta^{i_{1}\cdots i_{n-p}}_{j_{1}\cdots j_{n-p}},
\end{equation}
where
\begin{equation}
\delta^{i_{1}\cdots i_{q}}_{j_{1}\cdots j_{q}}
=\delta^{[i_{1}}_{[j_{1}}\cdots \delta^{i_{q}]}_{j_{q}]}.
\end{equation}
If $X$ is a vector field then Cartan's formula is
\begin{equation}\label{cartanmagic}
\mathfrak{L}_{X}w = d\iota_{X}\omega+\iota_{X}d\omega,
\end{equation}
where $\iota$ denotes the interior product
\begin{equation}
(\iota_{X}\omega)_{i_{1}\cdots i_{p-1}}=X^{j}\omega_{ji_{1}\cdots i_{p-1}},
\end{equation}
which also satisfies
\begin{equation}
\iota_{X}(\omega\wedge\alpha)=(\iota_{X}\omega)\wedge\alpha+(-1)^{p}\omega\wedge(\iota_{X}\alpha).
\end{equation}
 

\bibliographystyle{abbrv}
\bibliography{masterfile}

\begin{thebibliography}{10}

\bibitem{alaee2014thesis}
A.~Alaee.
\newblock Geometric inequalities for initial data with symmetries.
\newblock {\em Ph.D. thesis, Memorial University}, 2015.

\bibitem{alaee2015global}
A.~Alaee, M.~Khuri, and H.~Kunduri.
\newblock Proof of the mass-angular momentum inequality for bi-axisymmetric
  black holes with spherical topology.
\newblock {\em Adv. Theor. Math. Phys.}, 20(6):1397--1441, 2016,
  arXiv:1510.06974.

\bibitem{Alaeeremarks2015}
A.~Alaee and H.~K. Kunduri.
\newblock Remarks on mass and angular momenta for $u(1)^2$-invariant initial
  data.
\newblock {\em Preprint, arXiv:1508.02337}, 2015.

\bibitem{Breckenridge:1996is}
J.~C. Breckenridge, R.~C. Myers, A.~W. Peet, and C.~Vafa.
\newblock {D-branes and spinning black holes}.
\newblock {\em Phys. Lett. B}, 391:93--98, 1997.

\bibitem{cha2014deformations}
Y.~S. Cha and M.~A. Khuri.
\newblock Deformations of axially symmetric initial data and the mass-angular
  momentum inequality.
\newblock {\em Annales Henri Poincar{\'e}}, 16(3):841--896, 2015,
  arXiv:1401.3384.

\bibitem{cha2015deformations}
Y.~S. Cha and M.~A. Khuri.
\newblock Deformations of charged axially symmetric initial data and the
  mass-angular momentum-charge inequality.
\newblock {\em Annales Henri Poincar{\'e}}, 16(12):2881--2918, 2015,
  arXiv:1407.3621.

\bibitem{Chong:2005hr}
Z.~W. Chong, M.~Cvetic, H.~Lu, and C.~N. Pope.
\newblock {General non-extremal rotating black holes in minimal
  five-dimensional gauged supergravity}.
\newblock {\em Phys. Rev. Lett.}, 95:161301, 2005.

\bibitem{chrusciel2008masspositivity}
P.~T. Chru{\'s}ciel.
\newblock Mass and angular-momentum inequalities for axi-symmetric initial data
  sets i. positivity of mass.
\newblock {\em Annals of Physics}, 323(10):2566--2590, 2008.

\bibitem{Chrusciel2009}
P.~T. Chru{\'s}ciel and J.~L. Costa.
\newblock Mass, angular-momentum and charge inequalities for axisymmetric
  initial data.
\newblock {\em Classical and Quantum Gravity}, 26(23):235013, 2009.

\bibitem{chrusciel2008mass}
P.~T. Chru{\'s}ciel, Y.~Li, and G.~Weinstein.
\newblock Mass and angular-momentum inequalities for axi-symmetric initial data
  sets. ii. angular momentum.
\newblock {\em Annals of Physics}, 323(10):2591--2613, 2008.

\bibitem{Compere}
G.~Comp\`{e}re, S.~de~Buyl, E.~Jamsin, and A.~Virmani.
\newblock $g_2$ dualities in $d = 5$ supergravity and black strings.
\newblock {\em Classical and Quantum Gravity}, 26(12):125016, 2009.

\bibitem{Costa2010}
J.~L. Costa.
\newblock Proof of a dain inequality with charge.
\newblock {\em Journal of Physics A: Mathematical and Theoretical},
  43(28):285202, 2010.

\bibitem{Cremmer:1980gs}
E.~Cremmer.
\newblock {Supergravities in 5 Dimensions}.
\newblock In {\em {In *Salam, A. (ed.), Sezgin, E. (ed.): Supergravities in
  diverse dimensions, vol. 1* 422-437. (In *Cambridge 1980, Proceedings,
  Superspace and supergravity* 267-282) and Paris Ec. Norm. Sup. - LPTENS 80-17
  (80,rec.Sep.) 17 p. (see Book Index)}}, 1980.

\bibitem{Cremmer:1997ct}
E.~Cremmer, B.~Julia, H.~Lu, and C.~N. Pope.
\newblock {Dualization of dualities. 1.}
\newblock {\em Nucl. Phys. B}, 523:73--144, 1998.

\bibitem{Cremmer:1998px}
E.~Cremmer, B.~Julia, H.~Lu, and C.~N. Pope.
\newblock {Dualization of dualities. 2. Twisted self-duality of doubled fields,
  and superdualities}.
\newblock {\em Nucl. Phys. B}, 535:242--292, 1998.

\bibitem{Cremmer:1999du}
E.~Cremmer, B.~Julia, H.~Lu, and C.~N. Pope.
\newblock {Higher dimensional origin of D = 3 coset symmetries}.
\newblock 1999.

\bibitem{Cvetic:1996xz}
M.~Cvetic and D.~Youm.
\newblock {General rotating five-dimensional black holes of toroidally
  compactified heterotic string}.
\newblock {\em Nucl. Phys. B}, 476:118--132, 1996.

\bibitem{Dain2008}
S.~Dain.
\newblock Proof of the angular momentum-mass inequality for axisymmetric black
  holes.
\newblock {\em J. Differ. Geom}, 79(1):33--67, 2008.

\bibitem{dain2012geometric}
S.~Dain.
\newblock Geometric inequalities for axially symmetric black holes.
\newblock {\em Classical and Quantum Gravity}, 29(7):073001, 2012.

\bibitem{dain2013lower}
S.~Dain, M.~Khuri, G.~Weinstein, and S.~Yamada.
\newblock Lower bounds for the area of black holes in terms of mass, charge,
  and angular momentum.
\newblock {\em Phys. Rev. D}, 88(2):024048, 2013, arXiv:1306.4739.

\bibitem{Elvang:2004ds}
H.~Elvang, R.~Emparan, D.~Mateos, and H.~S. Reall.
\newblock {Supersymmetric black rings and three-charge supertubes}.
\newblock {\em Phys. Rev.}, D71:024033, 2005.

\bibitem{emparan2008black}
R.~Emparan and H.~S. Reall.
\newblock Black holes in higher dimensions.
\newblock {\em Living Rev. Rel}, 11(6):0801--3471, 2008.

\bibitem{Gibbons:1994vm}
G.~W. Gibbons, G.~T. Horowitz, and P.~K. Townsend.
\newblock {Higher dimensional resolution of dilatonic black hole
  singularities}.
\newblock {\em Class. Quant. Grav.}, 12:297--318, 1995.

\bibitem{Gibbons:1993xt}
G.~W. Gibbons, D.~Kastor, L.~A.~J. London, P.~K. Townsend, and J.~H. Traschen.
\newblock {Supersymmetric selfgravitating solitons}.
\newblock {\em Nucl. Phys. B}, 416:850--880, 1994.

\bibitem{KhuriSokolowsky}
M.~Khuri and B.~Sokolowsky.
\newblock Existence of brill coordinates for initial data with asymptotically
  cylindrical ends and applications.
\newblock {\em in preparation}, 2016.

\bibitem{khuri2015positive}
M.~Khuri and G.~Weinstein.
\newblock The positive mass theorem for multiple rotating charged black holes.
\newblock {\em Calc. Var. Partial Differential Equations}, 55:1--29, 2016,
  arXiv:1502.06290.

\bibitem{Kunduri:2011zr}
H.~K. Kunduri and J.~Lucietti.
\newblock {Constructing near-horizon geometries in supergravities with hidden
  symmetry}.
\newblock {\em JHEP}, 07:107, 2011.

\bibitem{maison1979ehlers}
D.~Maison.
\newblock Ehlers-harrison-type transformations for jordan's extended theory of
  gravitation.
\newblock {\em General Relativity and Gravitation}, 10(8):717--723, 1979.

\bibitem{Marolf:2000cb}
D.~Marolf.
\newblock {Chern-Simons terms and the three notions of charge}.
\newblock In {\em {Quantization, gauge theory, and strings. Proceedings,
  International Conference dedicated to the memory of Professor Efim Fradkin,
  Moscow, Russia, June 5-10, 2000. Vol. 1+2}}, pages 312--320, 2000.

\bibitem{Mizoguchi}
S.~Mizoguchia and N.~Ohta.
\newblock More on the similarity between $d=5$ simple supergravity and $m$
  theory.
\newblock {\em Physics Letters B}, 441(1-4):123--132, 1998.

\bibitem{Myers2011}
R.~C. Myers.
\newblock Myers-perry black holes.
\newblock {\em Preprint, arXiv:1111.1903}, 2011.

\bibitem{schoen2013convexity}
R.~Schoen and X.~Zhou.
\newblock Convexity of reduced energy and mass angular momentum inequalities.
\newblock {\em Annales Henri Poincar{\'e}}, 14(7):1747--1773, 2013.

\bibitem{strominger1996microscopic}
A.~Strominger and C.~Vafa.
\newblock Microscopic origin of the bekenstein-hawking entropy.
\newblock {\em Physics Letters B}, 379(1):99--104, 1996.

\bibitem{Tomizawa:2009ua}
S.~Tomizawa, Y.~Yasui, and A.~Ishibashi.
\newblock {Uniqueness theorem for charged rotating black holes in
  five-dimensional minimal supergravity}.
\newblock {\em Phys. Rev. D}, 79:124023, 2009.

\bibitem{witten1981new}
E.~Witten.
\newblock A new proof of the positive energy theorem.
\newblock {\em Comm. Math. Phys.}, 80(3):381--402, 1981.

\bibitem{Yokota}
I.~Yokota.
\newblock {\em Exceptional Lie Groups}.
\newblock arXiv:0902.0431.

\end{thebibliography}

\end{document}